\documentclass{article}
\usepackage{hyperref}
\usepackage{amsmath}
\usepackage{amssymb}
\usepackage{amsmath}
\usepackage{amsthm}
\usepackage{enumitem}
\usepackage{graphicx}
\usepackage{subcaption}
\usepackage{tikz}
\usepackage{tkz-graph}
\usepackage{biblatex}
\usepackage[margin=1in]{geometry}
\usetikzlibrary{shapes}

\newtheorem{theorem}{Theorem}
\newtheorem*{theorem*}{Theorem}
\newtheorem{proposition}{Proposition}

\newtheorem{lemma}{Lemma}
\newtheorem{corollary}{Corollary}

\theoremstyle{remark}
\newtheorem{remark}{Remark}

\theoremstyle{definition}
\newtheorem{definition}{Definition} 

\renewcommand{\r}{\mathbb{R}}
\renewcommand{\c}{\mathbb{C}}
\DeclareMathOperator{\qut}{Qut}
\DeclareMathOperator{\arity}{arity}
\newcommand{\pn}{\mathcal{P}_{\fc}}
\renewcommand{\k}{\mathbf{K}}

\DeclareMathOperator{\s}{\mathbf{S}}
\DeclareMathOperator{\f}{\mathbf{F}}
\DeclareMathOperator{\e}{\mathbf{E}}
\DeclareMathOperator{\ii}{\mathbf{I}}
\DeclareMathOperator{\fc}{\mathcal{F}}
\DeclareMathOperator{\gc}{\mathcal{G}}
\renewcommand{\u}{\mathcal{U}}
\DeclareMathOperator{\q}{\mathcal{Q}}
\DeclareMathOperator{\gf}{\mathfrak{G}}
\DeclareMathOperator{\qk}{\mathfrak{Q}^{\cal P}_{\cal F}}
\DeclareMathOperator{\one}{\mathbf{1}}
\DeclareMathOperator{\spn}{span}
\newcommand{\tcwd}[1]{\langle #1 \rangle_{+, \circ, \otimes ,\dagger}}
\newcommand{\tcwdn}[1]{\langle #1 \rangle_{\circ, \otimes ,\dagger}}
\DeclareMathOperator{\vx}{\mathbf{x}}
\DeclareMathOperator{\vy}{\mathbf{y}}
\DeclareMathOperator{\vz}{\mathbf{z}}

\DeclareMathOperator{\vf}{\mathbf{f}}
\DeclareMathOperator{\vg}{\mathbf{g}}
\DeclareMathOperator{\vq}{\mathbf{Q}}

\newcommand{\vc}{\mathcal{V}}
\DeclareMathOperator{\plholant}{Pl-Holant}
\DeclareMathOperator{\holant}{Holant}
\DeclareMathOperator{\eq}{\mathcal{EQ}}

\newcommand\numberthis{\addtocounter{equation}{1}\tag{\theequation}}
\renewcommand*{\VertexSmallMinSize}{6pt}
\newcommand{\DEdge}[2]{
    \draw[thick] (#1) -- (#2) node[draw, fill=white, kite, kite vertex angles = 120, minimum size = 4pt, inner sep = 1pt, pos=0.3, sloped] {};
}

\title{Planar \#CSP Equality Corresponds to Quantum Isomorphism -- A Holant Viewpoint}

\author{Jin-Yi Cai\footnote{Department of Computer Sciences, University of Wisconsin-Madison}
        \hspace{4cm} 
        Ben Young\footnote{Department of Computer Sciences, University of Wisconsin-Madison}\\
    \texttt{\hspace{0.7cm}\href{mailto:jyc@cs.wisc.edu}{jyc@cs.wisc.edu}} \hspace{2.2cm}
\texttt{\href{mailto:benyoung@cs.wisc.edu}{benyoung@cs.wisc.edu}}}
\date{}

\begin{document}
\maketitle

\begin{abstract}
\vspace{.1in}
Recently, Man{\v c}inska and Roberson proved~\cite{planar} that two graphs $G$ and $G'$
are \emph{quantum isomorphic} if and only if
they admit the same number of homomorphisms from all \emph{planar} graphs.
We extend this result to planar \#CSP with any pair of sets $\fc$ and $\fc'$ of real-valued, arbitrary-arity constraint functions. Graph homomorphism is the
special case where each of $\fc$ and $\fc'$ contains a single symmetric 0-1-valued binary constraint function.
Our treatment uses the framework of planar Holant problems.
To prove that quantum isomorphic constraint function sets give the same value on any planar \#CSP instance, we apply a novel form of
\emph{holographic transformation} of Valiant~\cite{valiant},
using the quantum permutation matrix $\u$ defining the quantum isomorphism.
Due to the noncommutativity of $\u$'s entries, it turns out that this
form of holographic transformation 
is only applicable to planar Holant.
To prove the converse, we introduce the
quantum automorphism group $\qut(\fc)$ of a set of constraint functions $\fc$, 
and characterize the intertwiners of $\qut(\fc)$ as the signature matrices of planar 
$\holant(\fc\,|\,\eq)$ quantum gadgets.
Then we define a new notion of (projective) connectivity for constraint functions and 
reduce arity while preserving the quantum automorphism group. 
Finally,  to address the challenges
posed by generalizing from 0-1 valued to real-valued constraint functions, 
we adapt a  technique of Lov\'asz~\cite{lovasz} in the classical
setting for isomorphisms of real-weighted graphs 
to the setting of quantum isomorphisms.
\end{abstract}

\newpage
\section{Introduction}
\label{sec:introduction}
\paragraph*{Graph Homomorphism and \#CSP.}
A homomorphism from graph $K$ to graph $X$ is an edge-preserving map from the vertex set $V(K)$ of $K$
to the vertex set $V(X)$ of $X$. A well-studied problem in complexity theory is to count the number of
distinct homomorphisms from $K$ to $X$, which can be expressed as 
\[
    \sum_{\sigma: V(K) \to V(X)}\prod_{(u,v)\in E(K)} (A_X)_{\sigma(u), \sigma(v)},
\]
the value of the \emph{partition function} of $X$ evaluated on $K$, 
where $A_X$ is the adjacency matrix of $X$. From this perspective,
graph homomorphism naturally generalizes to a counting constraint satisfaction problem (\#CSP) by replacing 
$\{A_X\}$ with a set $\fc$ of $\r$ or $\c$-valued \emph{constraint functions}
on one or more inputs from a finite domain $V(\fc)$, and replacing $K$ with a set of
constraints and variables, where each constraint applies a constraint function to a sequence of variables.
The problem is to compute the partition function, which is the sum over all variable assignments of the product of the
constraint function evaluations. Letting $V(\fc) = V(X)$ and the constraint and variable sets be
$E(K)$ and $V(K)$, respectively, with each edge/constraint applying $A_X$ to its two endpoints,
we recover the special case of counting homomorphisms from $K$ to $X$.

Bulatov~\cite{bulatov_2013} proved that every problem
\#CSP$(\fc)$, parameterized by a finite set $\fc$ of 0-1-valued constraint functions, is either (1)
 solvable in
polynomial-time or (2) \#P-complete.
Dyer and Richerby \cite{dyer_richerby} proved that
this \emph{complexity dichotomy} has a decidable
criterion. This dichotomy was further extended
to nonnegative real-valued, and then to all complex-valued constraint functions~\cite{cai-chen-lu,cai-chen-complexity}.
When we restrict to planar
\#CSP instances (for which the bipartite constraint-variable incidence graph is planar), a further \emph{complexity trichotomy} is known for the Boolean domain (where $V(\fc) = \{0,1\}$)
\cite{guo}, that there are 
exactly three classes: (1) polynomial-time solvable; (2) \#P-hard for general instances but solvable in polynomial-time over planar structures; and (3) \#P-hard over planar structures. 
Furthermore, Valiant's holographic algorithm with matchgates~\cite{valiant} is universal 
for all problems in class (2):
\emph{Every} \#P-hard
\#CSP problem that is solvable
in polynomial-time in the planar setting
is solvable by this one algorithmic strategy.
However, for planar \#CSP
on domains of  size greater than 2, a full complexity classification is open. 


\paragraph*{Holant Problems.}
We carry out much of our work in the planar Holant framework from counting complexity, which
we find natural to this theory, and of which
planar \#CSP itself is a special case. Like a \#CSP problem, a Holant problem is parameterized by a
set $\fc$ of constraint functions.
The input to a planar Holant problem is a \emph{signature grid}, a planar graph
where each edge represents a variable and every vertex is assigned a constraint function from $\fc$.
A vertex's constraint function is applied to its incident edges. 
This is dual to the \#CSP view of graph homomorphism, where each edge
is a (necessarily binary) constraint and each vertex is a variable.
As with \#CSP, the computational problem is to compute the Holant value -- the sum over all
variable (edge) assignments, of the product of the evaluations of the constraint functions.
A (planar) \emph{gadget} is a (planar) Holant signature grid  with a number of \emph{dangling} edges,
representing external variables. Each gadget has an associated \emph{signature matrix}, which stores
the Holant value for each fixed assignment to the dangling edges. The study of Holant problems is
motivated by Valiant's \emph{holographic transformations} \cite{valiant}, which are certain Holant value-preserving
transformations of the constraint functions by invertible matrices.

\paragraph*{Classical and Quantum Isomorphism.}
As suggested above, one can view a $q$-vertex real-weighted graph $X$, via its adjacency matrix $A_X \in \r^{q \times q}$,
as an $\mathbb{R}$-valued binary (i.e. two input variables) constraint function.
Two $q$-vertex graphs $X$ and $Y$ are
isomorphic if one can apply a permutation to the rows and columns of $A_X$ to obtain $A_Y$. 
Equivalently, if we convert $A_X$ and $A_Y$ to vectors $a_X, a_Y \in \mathbb{R}^{q^2}$,
there is a permutation matrix $P$ satisfying $P^{\otimes 2} a_X = a_Y$.
For $n$-ary constraint functions $F, G \in \r^{[q]^n}$,
a natural generalization applies. $F$ and $G$ are isomorphic
if there is a permutation matrix $P$ satisfying
$P^{\otimes n}f = g$, where $f, g \in \r^{q^n}$ are the vector versions of $F$ and $G$. 

Quantum isomorphism of (undirected, unweighted) graphs, introduced in \cite{asterias}, is
a relaxation of classical isomorphism. Graphs $X$ and $Y$ are \emph{quantum
isomorphic} if there is a perfect winning strategy in a two-player \emph{graph isomorphism game}
in which the players share and can perform measurements on an entangled quantum state.
This condition is equivalent to the existence of a \emph{quantum permutation matrix} matrix $\u$ 
-- a relaxation of a permutation matrix whose entries do not necessarily commute -- 
satisfying $\u^{\otimes 2}a_X = a_Y$ \cite{lupini_nonlocal_2018}. Analogously to classical isomorphism, 
in this work we define $n$-ary constraint functions $F$ and $G$ to be 
\emph{quantum isomorphic} if there is a quantum permutation matrix
$\u$ satisfying $\u^{\otimes n}f = g$. Sets $\fc$ and $\gc$ of constraint functions of equal
cardinality are quantum isomorphic if there is a single quantum permutation matrix defining
a quantum isomorphism between every pair of corresponding functions in $\fc$ and $\gc$.

In \cite{lovasz_operations}, Lov\'asz proved that two graphs are isomorphic if and only if they admit the same number of
homomorphisms from every graph. Fifty years later, Man{\v c}inska and Roberson \cite{planar} proved that two graphs are \emph{quantum}
isomorphic if and only if they admit the same number of
homomorphisms from all \emph{planar} graphs. We
generalize this result to \#CSP and sets of constraint functions. We achieve this via graph combinatorics, results from quantum group theory, and a novel form of holographic transformation,
establishing new connections between
planar Holant, \#CSP, quantum permutation matrices, 
and quantum isomorphism. 

While quantum permutation matrices, quantum isomorphism, and other
quantum constructions in this paper are somewhat abstract and technical, we believe it is precisely these
concepts' abstractness that makes the connections we develop between them and the very concrete,
combinatorial concept of planarity so fascinating and potentially fruitful.
Our result that quantum isomorphism exactly captures planarity could lead to entirely novel, 
algebraic methods of studying the complexity of planar \#CSP and Holant.


\paragraph*{Our Results.}
Our main result is the following theorem, a broad extension of the
main result of Man{\v c}inska and Roberson \cite{planar}, recast into the well-studied Holant and
\#CSP frameworks.
\begin{theorem*}[\autoref{thm:result}, informal]
    Sets $\fc$ and $\gc$ of $\r$-valued constraint functions are quantum isomorphic 
    iff the partition function of every planar \#CSP$(\fc)$ 
    instance is preserved upon replacing every constraint function in $\fc$ with
    the corresponding function in $\gc$.
\end{theorem*}
Our general constraint functions add significant complexity relative to the graph homomorphism special
case in \cite{planar}, since, unlike unweighted graph adjacency matrices,
they can be (1) asymmetric (i.e. permuting the argument order affects their value), (2) $n$-ary, for 
$n > 2$, and (3) arbitrary real-valued. Each of these three extensions adds intricacies and challenges not present in
\cite{planar}, which we address with novel approaches that reveal new, deeper connections between
quantum permutation matrices and planar graphs. 


First, in \autoref{sec:decomposition} we give a procedure for decomposing any planar Holant signature grid 
corresponding to a planar \#CSP instance into a small set of simple gadgets. Here arise the first
new complications associated with higher-arity signatures. The dangling edges of simple gadgets extracted from the
signature grid may not be oriented correctly, so we must use certain other gadgets to pivot them to
the correct orientation, respecting planarity.

With some preparation in \autoref{sec:invariance}, we prove the quantum Holant theorem in 
\autoref{sec:holanttheorem}. The forward direction of 
\autoref{thm:result} is a direct corollary, giving a more graphical and more intuitive proof than
that of the graph homomorphism special case in \cite{planar}.
The gadget decomposition gives an expression for the Holant value as a product of
the component gadgets' signature matrices. So,
assuming $\fc$ and $\gc$ are quantum isomorphic,
we use the quantum permutation matrix $\u$ defining the quantum isomorphism as a \emph{quantum holographic transformation}, inserting
tensor powers of $\u$ and its inverse between every pair of
signature matrices in the product without changing the Holant value.
Then a sequence of these holographic transformations 
converts every signature in $\fc$ to the corresponding signature in $\gc$.
The quantum holographic transformation does not work on general signature grids, since
viewing $\u$ itself as a constraint function in the signature grid is not in general well-defined, as 
$\u$'s entries do not commute and the partition function does not specify an order to multiply the
constraint function evaluations. However, the planarity of the signature grid and the resulting gadget 
decomposition and matrix product expression for the Holant value implicitly provide a multiplication
order.
Quantum holographic transformations apply to the planar version of the general Holant problem parameterized by a set $\fc$ of constraint functions (not just
the special case of \#CSP), and should be of independent interest.

The success of the quantum holographic transformation for asymmetric signatures is also dependent
on the fact that the holographic transformation action of a quantum permutation matrix is 
invariant under gadget rotations and reflections.
The asymmetry and rotation and reflection issues are only relevant in the context of planar signature 
grids, since in nonplanar grids, one can simply cross and twist the incident edges to achieve the 
desired input order. Hence this is another interesting connection between quantum permutation matrices
and the structural properties of planar graphs.

In \autoref{sec:backward}, to prove the reverse direction of \autoref{thm:result}, we turn to the theory of quantum groups
\cite{woronowicz_compact_1987, wang_quantum_1998}. We
introduce the quantum automorphism group $\qut(\fc)$ of a
set $\fc$ of signatures, an abstraction of the classical automorphism group satisfying many of the
same properties.
Using the planar gadget decomposition, we prove that
the signature matrices of planar Holant gadgets in the context of \#CSP$(\fc)$, a very concrete,
combinatorial concept, exactly
capture the abstract \emph{intertwiner space} of $\qut(\fc)$.

A natural approach to the rest of the proof breaks down for constraint functions of arity $> 2$.
Hence we introduce a method to reduce a constraint function's
arity while maintaining its inclusion in the original intertwiner space. Then
we say a constraint function is \emph{projectively connected} if this
procedure yields a connected graph upon reaching arity 2. 
Finally, we show that if $\fc$ and $\gc$ are projectively connected and the quantum automorphism group
of the disjoint union of $\fc$ and $\gc$ maps a `vertex' of $\fc$ to a `vertex' of $\gc$, then $\fc$ and $\gc$
are quantum isomorphic (analogous to the familiar classical fact for graphs).
For 0-1 valued functions, Man{\v c}inska and Roberson~\cite{planar} ensured projective connectivity 
by taking complements. 
However, for real-valued functions $F$ and $G$
this method does not work: we cannot take the complement to assume they are projectively connected. 
Instead, we adapt to the quantum setting
a technique of Lov\'asz~\cite{lovasz} in the classical
setting for real-weighted graphs, and extract a quantum isomorphism to complete the proof.

All of the above results extend to sets of constraint functions over $\mathbb{C}$
that are closed under conjugation and for which
the quantum isomorphism respects conjugation (both properties
are trivially satisfied by constraint functions over $\mathbb{R}$), and our proof is carried out in
this setting.

In \autoref{sec:connectivity}, we give an alternate approach 
for enforcing constraint function connectivity due to Roberson \cite{roberson}, which adds new binary
connected constraint functions to $\fc$ and $\gc$ rather than modify the existing constraint functions 
to be projectively connected.
In \autoref{sec:game}, we extend the connection between 
quantum isomorphism and nonlocal games. We define graph
isomorphism games for complex-weighted directed graphs and prove
the following generalization of a result in \cite{planar}: real-weighted graphs
$F$ and $G$ admit the same number of homomorphisms from all
planar graphs if and only if there is a perfect quantum commuting strategy for the $(F,G)$-isomorphism game. 
Finally, in \autoref{sec:category}, we discuss how pivoting dangling edges around a gadget and horizontally reflecting gadgets, graphical manipulations that arise naturally
throughout our work, correspond to the dual and adjoint operations in the pivotal dagger category
of gadgets.

We hope that our results, in particular the quantum holographic transformation technique in
\autoref{thm:quantumholant}, will lead to further applications of quantum group theory to the study of planar
\#CSP and Holant complexity.

\section{Preliminaries}
\label{sec:preliminaries}

For notational brevity, following 
\cite{petrescu, dudek}, and others working with $n$-ary structures,
write $x_1^r$ to mean
$x_1, \ldots, x_r$ or $x_1 \ldots x_r$, depending on context, and $x_r^1$ to mean
$x_r, \ldots, x_1$ or $x_r \ldots x_1$.
When the index range is clear, we simply write 
$\vx = x_1, \ldots, x_r$.

A \emph{graph} $X$ consists of a vertex set $V(X)$ and an edge
set $E(X)$ of subsets of $V(X)$ of size 1 or 2. That is, edges
are undirected and loops are allowed. The \emph{adjacency matrix} $A_X$ of $X$ is a symmetric matrix defined by
\[
    (A_X)_{x,y} = \begin{cases} 1 & \{x,y\} \in E(X) \\
    0 & \text{otherwise}
    \end{cases}.
\]
for $x,y \in V(X)$. If $(A_X)_{x,y} = 1$ for distinct vertices $x$ and $y$, we say $x$ and $y$
are \emph{adjacent}.
A \emph{multigraph} can have multiple (undirected) edges between
a single pair of vertices.
More generally, a \emph{$\c$-weighted graph} has directed edges
and loops
assigned values in $\c$. Any matrix in $\c^{q \times q}$
serves as the adjacency matrix of a $q$-vertex $\c$-weighted graph.

Following \cite{lupini_nonlocal_2018}, we refer to graphs and
$\c$-weighted graphs by the letters $X$, $Y$, and $Z$ (as
$G$ is reserved to denote a general constraint function).

Let $X$ be a graph with no loops, and $Y$ be any graph. A map $\varphi: V(X) \to V(Y)$ is a
\emph{homomorphism} from $X$ to $Y$ if $\varphi$ maps edges of
$X$ to edges of $Y$ -- that is,
if $x$ and $y$ are adjacent in $X$, then either $\varphi(x) \neq \varphi(y)$ and $\varphi(x)$ and $\varphi(y)$ are adjacent in $E(Y)$, or $\varphi(x) = \varphi(y)$ and this vertex has a loop
in $Y$.

$A^{\dagger}$ denotes the conjugate transpose of matrix $A$ with entries in $\c$.

\subsection{Constraint functions and \#CSP}
\begin{definition}[Constraint function, $V(F)$, $V(\fc)$, CC]
Think of a tensor $F \in \c^{[q]^n}$, for $q, n \geq 1$, as being a function on $n$ variables taking values in $[q]$ -- a
\emph{constraint function} of domain size $q$ and arity $n$. For $\vx \in [q]^n$, write
$F_{\vx} = F_{x_1\ldots x_n} = F(x_1,\ldots,x_n) \in \c$.

Motivated by the cases where $n = 2$ and $F$ is the adjacency
matrix of a $\c$-weighted graph, we often will write $V(F)$ instead of $[q]$ as the variable
domain -- that is, $F \in \c^{V(F)^n}$.
Whenever we specify a set $\fc$ of constraint functions, it is
assumed that all constraint functions in $\fc$ have the same
domain, which we call $V(\fc)$, with $|V(\fc)| = q$.

For a complex-valued constraint function $F$, let $\overline{F}$ be the entrywise conjugate of $F$ -- that is,
$\left(\overline{F}\right)_{\vx} = \overline{F_{\vx}}$. A set $\fc$ of
constraint functions is \emph{conjugate closed} (CC) if
$F \in \fc \iff \overline{F} \in \fc$.
\end{definition}
Clearly, any set of real-valued constraint functions is conjugate closed.

\begin{definition}[\#CSP, $Z$]
A \emph{\#CSP problem} \#CSP$(\fc)$ is parameterized by a set
$\fc$ of constraint functions.
For a single constraint function $F$, write \#CSP$(F)$ = \#CSP$(\{F\})$.
A \emph{\#CSP instance} $K$ in the context of \#CSP$(\fc)$ 
(called a \#CSP$(\fc)$ instance) is defined by a pair $(V,C)$, where $V$ is a 
set of variables and $C$ is a multiset of \emph{constraints}.
Each constraint consists of a constraint
function $F \in \fc$ and
an ordered tuple of variables to which $F$ is applied. 

Define the \emph{partition function} $Z$  on an input \#CSP$(\fc)$ instance $K = (V,C)$ by
\[
    Z(K) = \sum_{\sigma: V \to V(\fc)} \prod_{c \in C} F^c(\sigma|_c),
\]
where
$F^c(\sigma |_c) = F^c_{\sigma(v_1), \ldots, \sigma(v_n)}$ if $c$ specifies that $F^c \in \fc$ is applied
to variables $v_1,\ldots,v_n$.
\end{definition}

\begin{definition}[Compatible constraint function sets, $K_{\fc\to\gc}$]
\label{def:compatible}
For integer $t$, let $\fc = \{F_i\}_{i\in [t]}$, $\gc = \{G_i\}_{i\in [t]}$ be two sets of constraint functions.
We say $\fc$ and $\gc$ are \emph{compatible} if
there exists $q \geq 1$ and $n_i \ge 1$, such that $F_i, G_i \in \c^{[q]^{n_i}}$
for all $i \in [t]$ (in other words, all constraint functions
have the same domain size, and equal-indexed constraint functions
have the same arity). If $\fc$ and $\gc$ are CC, then we
further assume 
that $F_i = \overline{F_j} \iff G_i = \overline{G_j}$.
(For real-valued constraint functions
this requirement is trivial.)
We call $F_i$ and $G_i$ \emph{corresponding} constraint functions.

For compatible constraint function sets $\fc$ and $\gc$ and
any \#CSP$(\fc)$ instance $K$, define a \#CSP$(\gc)$ instance
$K_{\fc\to\gc}$ by replacing every constraint in $K$ with 
the corresponding constraint
function in $\gc$ applied to the same variable tuple.

For $F, G \in \c^{[q]^n}$, define $K_{F\to G} = K_{\{F\}\to\{G\}}$.
\end{definition}

\begin{definition}[$k$-labeled \#CSP instance, $Z^{\psi}$, $Z^v$]
    \label{def:labeledcsp}
    A \#CSP instance $K = (V,C)$ is \emph{$k$-labeled} if there is a $k$-tuple of variables
    $(v_1, \ldots, v_k) \in V^k$, not necessarily distinct, such that
    variable $v_i$ is labeled $i$, for $1 \leq i \leq k$.
    Let $K = (V,C)$ be a $k$-labeled instance and let $\psi: [k] \to [q]$ be a map fixing, or \emph{pinning}, the values 
    of the labeled variables. If $\psi$ assigns some multiply-labeled variable different values -- that
    is, there exist $i \neq j \in [q] \text{ such that } v_i = v_j \text{ but } \psi(i) \neq \psi(j)$,
    define $Z^{\psi}(K) = 0$. Otherwise, let
    \[
        Z^{\psi}(K) = \sum_{\sigma: V \to [q] \text{ extends } \psi} \prod_{c \in C} F^c(\sigma|_c),
    \]
    and $\sigma$ extends $\psi$ means $\sigma$ assigns value $\psi(i)$ to the variable labeled $i$.

    If $k=1$, then any $\psi: [1] \to [q]$ may be identified with $\psi(1)$, and we write
    $Z^{\psi(1)} = Z^{\psi}$.
\end{definition}
\autoref{def:labeledcsp} is a generalization of the definition
of a $k$-labeled graph in \cite{schrijver}. It is similar to
the definitions of a $k$-labeled graph in 
\cite{freedman_reflection, lovasz}, but in these definitions each variable has at most one label.
We allow multiple labels on a variable to make the definition of a $k$-labeled \#CSP instance match that of a
bi-labeled graph from \cite{planar}.

\begin{definition}[Constraint function isomorphism]
    Tensors/constraint functions $F, G \in \c^{[q]^n}$ are \emph{isomorphic} ($F \cong G$) if there is a $\varphi \in S_q$
    such that $G_{\varphi(x_1), \ldots, \varphi(x_n)} = F_{x_1,\ldots,x_n}$ for all $\vx \in [q]^n$.
    
    Compatible constraint function sets $\fc$ and $\gc$ of cardinality $s$ and common domain $[q]$ are
    isomorphic ($\fc \cong \gc$) if there is a $\varphi \in S_q$
    which is an isomorphism between $F_i$ and $G_i$ for all
    $i \in [s]$.
\end{definition}
Note that we require that every pair $F_i$ and $G_i$
be isomorphic via the same $\varphi$.

For a constraint function $F$ often it will be useful to `flatten' it into a matrix:
\begin{definition}[$F^{m,d}$, $f$]
    \label{def:flatten}
    For $F \in \c^{[q]^{n}}$ and any $m,d \geq 0$, $m+d = n$, let $F^{m,d} \in \c^{q^m \times q^d}$ be the $q^m \times q^d$ matrix
    defined by
    \[
        F^{m,d}_{x_1\ldots x_m, x_{n} \ldots x_{m+1}} = F_{\vx}
    \]
        for $\vx \in [q]^{m+d}$, where $x_1 \ldots x_m \in \mathbb{N}$ is the base-$q$ integer with
    the most significant digit $x_1$, and similarly for $x_{n} \ldots x_{m+1}$.
%
    We will frequently have $m = n$ and $d=0$, in which case we call $F^{n,0} \in \c^{q^n}$ the
    \emph{signature vector} of $F$ and denote it by a lowercase $f = F^{n,0}$.
\end{definition}
Observe that the column index bits are reversed. We do this to make the definition of a flattened
constraint function match \autoref{def:sigmatrix} of a signature matrix.


\subsection{Quantum permutation matrices and quantum tensor isomorphism}
We will also consider matrices with entries from an arbitrary (not necessarily commutative)
unital $C^*$-algebra. See \cite{cstaralgebra} for a reference
on $C^*$-algebras.
For such a matrix $\u = (u_{ij})$, let $\u^{\dagger}$ be the conjugate transpose of $\u$, whose
$(i,j)$ entry is $u^*_{ji}$. Let $\one$ be the $C^*$-algebra unit.
\begin{definition}[Quantum permutation matrix]
    A matrix $\u = (u_{ij})$ with entries from a unital $C^*$-algebra is called a \emph{quantum permutation matrix}
    if it satisfies the following for all $i,j$:
    \begin{itemize}
        \item $u_{ij}^2 = u_{ij} = u_{ij}^*$ (in other words, every entry of $\u$ is a \emph{projector}).
        \item $\sum_{j} u_{ij} = \sum_i u_{ij} = \one$.
    \end{itemize}
    It is well-known that these two conditions additionally imply 
    \begin{itemize}
        \item $u_{ij}u_{kj} = \delta_{ik} u_{ij}$ and $u_{ij}u_{ik} = \delta_{jk}u_{ij}$ for all $i,j,k$.
        \item $\u \u^{\dagger} = \u^{\dagger} \u = I \otimes \one$ ($\u$ is unitary).
    \end{itemize}
\end{definition}
For $q_1 \times q_2$ matrix $\u$ and $p_1 \times p_2$ matrix $\mathcal{V}$, both with entries from
a $C^*$-algebra,
let the Kronecker product $\u \otimes \mathcal{V}$ be the $q_1p_1 \times q_2p_2$ 
matrix defined by $(\u \otimes \mathcal{V})_{(x_1,x_2),(y_1,y_2)} = u_{x_1y_1} v_{x_2y_2}$. 
For $q_1 \times q_2$ matrix $\u$, define
the $q_1^n \times q_2^n$ matrix $\u^{\otimes n}$ inductively.
For base-$q_1$ integer $x_1\ldots x_n$ and base-$q_2$ integer
$y_1\ldots y_n$, index $(\u^{\otimes n})_{x_1^n,y_1^n} = u_{x_1y_1} \ldots u_{x_ny_n}$.

The following well-known identities for matrices 
hold:
For any appropriately-sized matrices 
$\u$, $\mathcal{V}$, and $M,N$ with entries from $\c$ (or scalar multiples of $\one$),
\begin{equation}
    \label{eq:tensordistribute}
    (\u \otimes \mathcal{V})(M \otimes N) = (\u M) \otimes (\mathcal{V} N),
    \qquad
    (M \otimes N)(\u \otimes \mathcal{V}) = (M \u) \otimes (N \mathcal{V}).
\end{equation}
If $\u$ is a quantum permutation matrix, then
\begin{equation}
    \u^{\otimes k}(\u^{\otimes k})^{{\dagger}} = I^{\otimes k} \otimes \one
    = (\u^{\otimes k})^{{\dagger}} \u^{\otimes k}.
    \label{eq:tensorinverse}
\end{equation}

The following concept was introduced in \cite{asterias}, and
expressed in this form in \cite{lupini_nonlocal_2018}.
\begin{definition}[Graph $\cong_{qc}$]
    \label{def:commutingiso}
    Graphs $X$ and $Y$ with adjacency matrices $A_X$ and $A_Y$ are \emph{quantum isomorphic} ($X \cong_{qc} Y$) if there is a quantum permutation matrix
    $\u$ satisfying $\u A_X = A_Y \u$.
\end{definition}

One can view a quantum permutation matrix as a generalization of a classical permutation matrix, which 
has rows and columns sum to 1, and each entry (0 or 1) is idempotent. In fact, if $\c$ is the
$C^*$-algebra over which $\u$ is defined, then $\u$ is a classical permutation matrix.

The following definition is motivated by the fact that $F,G \in \c^{[q]^n}$ are isomorphic if and
only if there is a permutation matrix $P \in \c^{q \times q}$ satisfying $P^{\otimes n}f = g$, and
by the discussion in the previous paragraph.
\begin{definition}[$\cong_{qc}$]
    \label{def:generalqc}
    $F,G \in \c^{[q]^n}$ are \emph{quantum isomorphic} ($F \cong_{qc} G)$ if there is a $q \times q$ quantum permutation matrix 
    matrix $\u$ satisfying $\u^{\otimes n}f = g$.
    
    Compatible sets $\fc$ and $\gc$ of constraint functions with cardinality $s$ and common domain $[q]$ are
    \emph{quantum isomorphic} ($\fc \cong_{qc} \gc$) if
    there is a $q \times q$ quantum permutation matrix $\u$
    satisfying $\u^{\otimes \arity(F_i)}f_i = g_i$ for every
    $i$.
\end{definition}
As with classical isomorphism, we require every corresponding pair of functions
in $\fc$ and $\gc$ to be quantum isomorphic via the same $\u$.
Throughout, we identify scalar multiples $\alpha \one$ of the $C^*$-algebra unit with the scalar
$\alpha \in \c$. In particular, we will write $\u^{\otimes n}f = g$ instead of 
$\u^{\otimes n}f = g \otimes \one$.
If $n=2$ and $F$ and $G$ are symmetric, 0-1 valued tensors,
then $F$ and $G$ are adjacency matrices of graphs.
For graphs $X$ and $Y$, by \autoref{lem:tensorcftog} below, we may
rewrite the condition $\u A_X = A_Y \u$ defining $X \cong_{qc} Y$ in \autoref{def:commutingiso} as $\u^{\otimes 2} A_X^{2,0} = A_Y^{2,0}$.
Thus \autoref{def:generalqc} is a generalization of
\autoref{def:commutingiso}.

\begin{definition}[$E_n, E^{m,d}$, $\eq$]
    \label{def:e}
    For fixed $q, n$, define the \emph{equality} constraint function $E_{q,n} \in \c^{[q]^n}$ by
    \[
        E_{q,n}(x_1,\ldots,x_n) = \begin{cases} 1 & x_1 = \ldots = x_n \\ 0 & \text{otherwise} \end{cases}
    \]
    for $\vx \in [q]^n$.
    We will almost always refer to $E_{q,n}$ as $E_n$, as $q$ is usually clear from context.
    When we flatten $E_n$ into a matrix $E^{m,d}_n$, we abbreviate $E^{m,d} := E_n^{m,d}$, as we
    must have $m+d = n$.

    For fixed $q$, define $\eq_q = \bigcup_n E_{q,n}$. We again will usually omit the $q$.
\end{definition}

\subsection{Holant, gadgets and signature matrices}
\label{sec:holant}
A Holant problem $\holant(\mathcal{F})$, like a  \#CSP problem,
is parameterized by a set $\mathcal{F}$ of constraint functions,
usually called \emph{signatures}.
The input to $\holant(\mathcal{F})$ is a \emph{signature grid} 
$\Omega$, which consists of an underlying multigraph $X$ with vertex set $V$ and edge set $E$. Each  vertex
$v \in V$ is assigned a signature $F_v \in \mathcal{F}$ of arity $\deg(v)$.
The incident edges  $E(v)$ of $v$ are labeled as input variables to $F_v$ taking values in $V(\fc)$.
We use $\plholant(\fc)$ to specify that input signature grids
must have planar underlying multigraphs. For planar Holant, the input variables of $F_v$
are labeled in cyclic order (either clockwise or counterclockwise) starting with
one particular edge:
$E(v) = (e_1^v,\ldots,e_{\deg(v)}^v)$.
$\Omega$ is said to be a signature grid \emph{in the context of} $\holant(\mathcal{F})$. 
The output on input $\Omega$ is
\begin{equation}\label{eqn:def-holant}
    \holant_{\Omega}(\mathcal{F}) = \sum_{\sigma: E \to V(\fc)} \prod_{v \in V} F_v(\sigma |_{E(v)}),
\end{equation}
where $F_v(\sigma |_{E(v)}) = (F_v)_{\sigma(e_1^v),\ldots,\sigma(e_{\deg(v)}^v)}$.
For example, counting perfect matchings 
is expressed by the {\sc Exact-One} function; counting proper edge colorings
is expressed by the {\sc Disequality} function.
For sets $\fc$ and $\gc$ of signatures, define the problem $\holant(\fc \,|\, \gc)$ as
follows. A signature grid in the context of $\holant(\mathcal{F} \,|\, \gc)$ has a bipartite underlying
multigraph with bipartition $V = V_1 \sqcup V_2$ such that the vertices in $V_1$ and $V_2$ are assigned
signatures from $\mathcal{F}$ and $\gc$, respectively.
Holant is likely a more expressive framework than \#CSP. We show in the next paragraph that 
\#CSP is expressible in the Holant framework, but it is unknown whether, 
for every signature set $\fc$, there is a signature set $\fc'$ such that the class of problems
$\holant(\fc)$ is equivalent to the class \#CSP$(\fc')$. This inexpressibility is proven when we restrict
from \#CSP to graph homomorphism, even with complex vertex and edge weights~\cite{lovasz,cai-govorov}. 

Let $\fc$ be a set of constraint
functions. To each \#CSP$(\fc)$ instance $K = (V,C)$ we associate a signature grid $\Omega_K$ in the context of
$\holant(\fc \,|\,  \mathcal{EQ})$ defined as follows: 
For every constraint $c \in C$, if $c$ applies function $F$ of arity $n$, create a degree-$n$ vertex assigned $F$, called a \emph{constraint vertex}.
For each variable $v \in V$, if $v$ appears in the multiset of constraints $C_v \subseteq C$, 
create a degree-$|C_v|$ vertex assigned $E_{|C_v|} \in \eq$, and edges $(v, c)$ for every $c \in C_v$.
If $c$'s variables are $v_1,\ldots,v_n$ (in order), draw the edges around the vertex corresponding to $c$ in
cyclic (either clockwise or counterclockwise) order of the edge to the vertex representing $v_1$, then the edge to the vertex
representing $v_2$, and so on until $v_n$.
Vertices assigned signatures in $\eq$ are called \emph{equality} vertices.
Each $\sigma$ assigns all incident edges of an equality vertex the same
value (or else the term corresponding to $\sigma$ is 0), so we can view $\sigma$ as variable assignent
in the \#CSP sense. Comparing the definitions of $Z(K)$ and $\holant_{\Omega_K}(\fc \mid \mathcal{EQ})$,
we have $Z(K) = \holant_{\Omega_K}(\fc \mid \mathcal{EQ})$.

For example, to compute the number of homomorphisms $K \to X$, consider a \#CSP$(A_X)$ instance $K'$ where the vertices of $K$ are variables and each edge of $K$ is a constraint
applying function $A_X \in \r^{V(X)^2}$ ($X$'s adjacency matrix) to the edge's two endpoints.
The corresponding Holant signature grid $\Omega_K$ starts with underlying graph $K$, with $K$'s vertices
assigned the appropriate equality signature from $\eq$, and we subdivide each of $K$'s edges by placing 
degree-2 constraint vertices, assigned signature $A_X$, to the labeled equality vertices. See
\autoref{fig:gadget_vs_bilabeled}.
We depict equality and constraint vertices as circles and squares, respectively.

\begin{figure}[ht!]
    \center
    \begin{tikzpicture}[scale=1]
    \GraphInit[vstyle=Classic]
    \SetUpEdge[style=-]
    \SetVertexMath

    \def\wirelen{0.8}
    \def\xshift{6}

    \Vertex[x=0,y=1,NoLabel]{a1}
    \Vertex[x=2,y=3,NoLabel]{b1}
    \Vertex[x=2,y=1,NoLabel]{c1}
    \Vertex[x=1,y=0,NoLabel]{d1}
    \Vertex[x=3,y=2,NoLabel]{e1}

    \Vertex[x=0+\xshift,y=1,L=E_3,Lpos=270]{a2}
    \Vertex[x=2+\xshift,y=3,L=E_2,Lpos=0]{b2}
    \Vertex[x=2+\xshift,y=1,L=E_4,Lpos=330,Ldist=-0.1cm]{c2}
    \Vertex[x=1+\xshift,y=0,L=E_2,Lpos=270]{d2}
    \Vertex[x=3+\xshift,y=2,L=E_1,Lpos=90]{e2}

    \foreach \xs/\num in {0/1,\xshift/2} {
        \Edges(a\num,b\num,c\num,d\num,a\num,c\num,e\num)
    };

    \tikzset{VertexStyle/.style = {shape=rectangle, fill=black, minimum size=5pt, inner sep=1pt, draw}}
    \Vertex[x=1+\xshift,y=2,NoLabel]{ax1}
    \Vertex[x=1+\xshift,y=1,NoLabel]{ax2}
    \Vertex[x=0.5+\xshift,y=0.5,NoLabel]{ax3}
    \Vertex[x=1.5+\xshift,y=0.5,NoLabel]{ax4}
    \Vertex[x=2+\xshift,y=2,NoLabel]{ax5}
    \Vertex[x=2.5+\xshift,y=1.5,NoLabel]{ax6}

    \Vertex[x=3+\xshift,y=0,L={=A_X}]{ax7}

    \node at (\xshift/1.34, 1.5) {$\rightsquigarrow$};

    \node at (0.1, 2.7) {$K$};
    \node at (0.1+\xshift, 2.7) {$\Omega_K$};

\end{tikzpicture}
    \caption{A graph $K$ and the corresponding
    $\holant(A_X \mid \eq)$ gadget $\Omega_K$ for computing the number of homomorphisms from $K$ to
$X$. Square vertices are assigned signature $A_X$.}
    \label{fig:gadget_vs_bilabeled}
\end{figure}
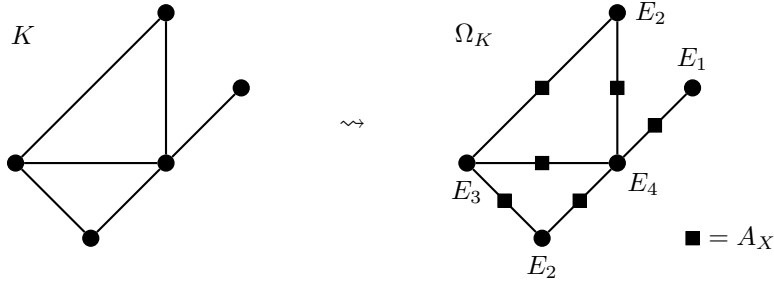

\begin{definition}[Planar, connected \#CSP instance]
    A \#CSP instance $K$ is \emph{planar/connected} if the underlying multigraph of the corresponding
    Holant signature grid is planar/connected.
\end{definition}

We now have the notation to state our main theorem.
\begin{theorem}
    \label{thm:result}
    Let $\fc$, $\gc$ be conjugate closed compatible sets of constraint functions.
    Then $\fc \cong_{qc} \gc$ if and only if $Z(K) = Z(K_{\fc\to\gc})$ for every planar \#CSP$(\fc)$ instance $K$.
\end{theorem}
If $\fc = \{A_X\}$ and $\gc = \{A_Y\}$, then \autoref{thm:result} specializes to the result of
\cite{planar}: graphs $X$ and $Y$ are quantum isomorphic iff they admit the same number of homomorphisms
from every planar graph $K$.

The next construction is very useful, and is the reason why
we look through the Holant lens in this work.
\begin{definition}[Gadget]
    \label{def:gadget}
    A \emph{gadget} is a Holant signature grid equipped with an ordered set of dangling edges (edges with only one endpoint), defining external variables.
    
    Our gadgets will usually be in the context of 
    $\holant(\fc \mid \eq)$, the Holant problem equivalent to
    \#CSP$(\fc)$. In this case, we specify that all dangling
    edges must be attached to equality vertices.
\end{definition}

\begin{definition}[$M(\k)$, $M(\k)$]
    \label{def:sigmatrix}
    Let $\k$ be a gadget with $n$ dangling edges and containing signatures of domain size $q$. 
    For any $m, d \geq 0$, $m + d = n$,
    define $\k$'s $(m,d)$-\emph{signature matrix}
    $M(\k) \in \c^{q^{m} \times q^d}$ by letting $M(\k)_{\vx,\vy}$ be the Holant value (summing over all
    assignments to non-dangling edges)
    when the first $m$ dangling edges (called \emph{output} or \emph{row} dangling edges) are assigned $x_1,\ldots,x_m$ and the last $d$
    dangling edges (called \emph{input} or \emph{column} dangling edges) are assigned $y_{d},\ldots,y_1$. We draw the output
    dangling edges extending to the left of the gadget, and
    the input dangling edges to the right.
\end{definition}

\begin{definition}[$\gf(m,d)$, Gadget $\circ, \otimes, {\dagger}$]
    \label{def:gadgetops}
    For a gadget $\k$ with $m+d$ dangling edges, write
    $\k \in \gf(m,d)$ to mean we consider $\k$ with $m$ output and $d$ input dangling edges.
    
    \begin{itemize}
        \item Given $\mathbf{K_1} \in \gf(m,d), \mathbf{K_2} \in \gf(d, m)$, define the composition $\mathbf{K_1} \circ \mathbf{K_2}
        \in \gf(m, w)$ as the gadget formed by placing $\mathbf{K_2}$ to the right of $\mathbf{K_1}$, and connecting each
        input dangling edge of $\mathbf{K_1}$ with the corresponding output dangling edge of $\mathbf{K_2}$.
        \begin{itemize}
            \item If composition makes vertices assigned 
            $E_a, E_b \in \eq$ adjacent, contract the edge between them, merging them into a single vertex
            assigned $E_{a+b-2}$.
        \end{itemize}
        \item For gadgets $\mathbf{K_1} \in \gf(m_1, d_1), \mathbf{K_2} \in \gf(m_2, d_2)$, define the tensor product
        $\mathbf{K_1} \otimes \mathbf{K_2} \in \gf(m_1 + m_2, d_1 + d_2)$ as the gadget formed by 
        taking the disjoint union of the multigraphs underlying
        $\mathbf{K_1}$ and $\mathbf{K_2}$, with $\mathbf{K_1}$ placed above $\mathbf{K_2}$.
        Define the resulting order
        of dangling edges as $\mathbf{K_1}$'s output, then $\mathbf{K_2}$'s output,
        then $\mathbf{K_2}$'s input, then $\mathbf{K_1}$'s input dangling edges.
        \item For $\k \in \gf(m,d)$, define the conjugate transpose $\k^{\dagger} \in \gf(d, m)$ by reflecting the
        underlying multigraph of $\k$ horizontally, then replacing every signature $F$ with $\overline{F}$. Maintain the edge input order of each signature,
        so a clockwise-oriented signature becomes counterclockwise-oriented, and vice-versa.
    \end{itemize}
\end{definition}
See \autoref{fig:gadgetops}. Dangling edges are drawn lighter and thinner than internal edges.
\begin{figure}[ht!]
    \begin{center}
    \begin{subfigure}{0.8\textwidth}
        \centering
        \begin{tikzpicture}[scale=0.7]
    \GraphInit[vstyle=Classic]
    \SetUpEdge[style=-]
    \SetVertexMath

    \def\wirelen{0.7}
    \def\xsh{7}

    \draw[thin, color=gray] (-\wirelen,-0.4) .. controls (-\wirelen/2,-0.4) .. (0,0);
    \draw[thin, color=gray] (-\wirelen,0.4) .. controls (-\wirelen/2,0.4) .. (0,0);
    \draw[thin, color=gray] (-\wirelen,2) -- (2,2);
    \draw[thin, color=gray] (2,0) -- (3+\wirelen,0);
    \draw[thin, color=gray] (3,1) -- (3+\wirelen,1);

    \tikzset{VertexStyle/.style = {shape=circle, fill=black, minimum size=5pt, inner sep=1pt, draw}}
    \Vertex[x=0,y=0,NoLabel]{a2}
    \Vertex[x=2,y=2,NoLabel]{b2}
    \Vertex[x=2,y=0,NoLabel]{c2}
    \Vertex[x=3,y=1,NoLabel]{e2}

    \Edges(a2,b2,c2)
    \Edges(c2,e2)

    \tikzset{VertexStyle/.style = {shape=rectangle, fill=black, minimum size=4pt, inner sep=1pt, draw}}
    \Vertex[x=1,y=1,NoLabel]{ax1}
    \Vertex[x=2,y=1,NoLabel]{ax5}
    \Vertex[x=2.5,y=0.5,NoLabel]{ax6}

    \Edge(c2)(ax1)

    \node at (1.5,2.5) {$\mathbf{K}$};

    \def\ysh{0.5}
    \draw[thin, color=gray] (\xsh+1+\wirelen,0+\ysh) -- (\xsh+1,0+\ysh);
    \draw[thin, color=gray] (\xsh+-1,0+\ysh) -- (\xsh-1-\wirelen,0+\ysh);
    \draw[thin, color=gray] (\xsh+-1,1+\ysh) -- (\xsh-1-\wirelen,1+\ysh);

    \tikzset{VertexStyle/.style = {shape=circle, fill=black, minimum size=5pt, inner sep=1pt, draw}}
    \Vertex[x=\xsh+1,y=0+\ysh,NoLabel]{a3}
    \Vertex[x=\xsh+1,y=1+\ysh,NoLabel]{b3}
    \Vertex[x=\xsh-1,y=0+\ysh,NoLabel]{c3}
    \Vertex[x=\xsh-1,y=1+\ysh,NoLabel]{e3}

    \node at (\xsh, 1.5+\ysh) {$\mathbf{L}$};

    \tikzset{VertexStyle/.style = {shape=rectangle, fill=black, minimum size=4pt, inner sep=1pt, draw}}
    \Vertex[x=\xsh-1,y=0.5+\ysh,NoLabel]{ax3}
    \Vertex[x=\xsh+0,y=0.5+\ysh,NoLabel]{bx3}

    \Edges(e3,c3,b3)
    \Edge(e3)(a3)
\end{tikzpicture}
    \end{subfigure}

    \begin{subfigure}{0.4\textwidth}
        \begin{tikzpicture}[scale=0.7]
    \GraphInit[vstyle=Classic]
    \SetUpEdge[style=-]
    \SetVertexMath

    \def\wirelen{0.7}
    \def\xsh{4}

    \node at (\xsh-2, 2.7) {$\mathbf{K} \circ \mathbf{L}$};

    \draw[thin, color=gray] (-\wirelen,-0.4) .. controls (-\wirelen/2,-0.4) .. (0,0);
    \draw[thin, color=gray] (-\wirelen,0.4) .. controls (-\wirelen/2,0.4) .. (0,0);
    \draw[thin, color=gray] (-\wirelen,2) -- (2,2);

    \tikzset{VertexStyle/.style = {shape=circle, fill=black, minimum size=5pt, inner sep=1pt, draw}}
    \Vertex[x=0,y=0,NoLabel]{a2}
    \Vertex[x=2,y=2,NoLabel]{b2}
    \Vertex[x=2,y=0,NoLabel]{merge2}
    \Vertex[x=3,y=1,NoLabel]{merge1}

    \Edges(a2,b2,merge2)

    \tikzset{VertexStyle/.style = {shape=rectangle, fill=black, minimum size=4pt, inner sep=1pt, draw}}
    \Vertex[x=1,y=1,NoLabel]{ax1}
    \Vertex[x=2,y=1,NoLabel]{ax5}
    \Vertex[x=2.34,y=0.66,NoLabel]{ax6}

    \Edge(merge2)(ax1)

    \draw[thin, color=gray] (\xsh+1+\wirelen,0) -- (\xsh+1,0);

    \tikzset{VertexStyle/.style = {shape=circle, fill=black, minimum size=5pt, inner sep=1pt, draw}}
    \Vertex[x=\xsh+1,y=0,NoLabel]{a3}
    \Vertex[x=\xsh+1,y=1,NoLabel]{b3}

    \tikzset{VertexStyle/.style = {shape=rectangle, fill=black, minimum size=4pt, inner sep=1pt, draw}}
    \Vertex[x=\xsh-1.21,y=0.5,NoLabel]{ax3}
    \Vertex[x=\xsh-0.2,y=0.6,NoLabel]{bx3}

    \Edge[style={bend left}](merge1)(merge2)
    \Edge[style={bend right}](merge1)(merge2)
    \Edge(merge2)(b3)
    \Edge(merge1)(a3)
\end{tikzpicture}
        \vspace{0.5cm}
    \end{subfigure}
    \begin{subfigure}{0.25\textwidth}
        \begin{tikzpicture}[scale=0.7]
    \GraphInit[vstyle=Classic]
    \SetUpEdge[style=-]
    \SetVertexMath

    \def\wirelen{0.7}
    \def\xsh{1}
    \def\ysh{2}

    \node at (1, \ysh+2.7) {$\mathbf{K} \otimes \mathbf{L}$};

    \draw[thin, color=gray] (-\wirelen,\ysh-0.4) .. controls (-\wirelen/2,\ysh-0.4) .. (0,\ysh+0);
    \draw[thin, color=gray] (-\wirelen,\ysh+0.4) .. controls (-\wirelen/2,\ysh+0.4) .. (0,\ysh+0);
    \draw[thin, color=gray] (-\wirelen,\ysh+2) -- (2,\ysh+2);
    \draw[thin, color=gray] (2,\ysh+0) -- (3+\wirelen,\ysh+0);
    \draw[thin, color=gray] (3,\ysh+1) -- (3+\wirelen,\ysh+1);

    \tikzset{VertexStyle/.style = {shape=circle, fill=black, minimum size=5pt, inner sep=1pt, draw}}
    \Vertex[x=0,y=\ysh+0,NoLabel]{a2}
    \Vertex[x=2,y=\ysh+2,NoLabel]{b2}
    \Vertex[x=2,y=\ysh+0,NoLabel]{c2}
    \Vertex[x=3,y=\ysh+1,NoLabel]{e2}

    \Edges(a2,b2,c2)
    \Edges(c2,e2)

    \tikzset{VertexStyle/.style = {shape=rectangle, fill=black, minimum size=4pt, inner sep=1pt, draw}}
    \Vertex[x=1,y=\ysh+1,NoLabel]{ax1}
    \Vertex[x=2,y=\ysh+1,NoLabel]{ax5}
    \Vertex[x=2.5,y=\ysh+0.5,NoLabel]{ax6}

    \Edge(c2)(ax1)

    \draw[thin, color=gray] (\xsh+1+\wirelen+1,0) -- (\xsh+1,0);
    \draw[thin, color=gray] (\xsh+-1,0) -- (\xsh-1-\wirelen,0);
    \draw[thin, color=gray] (\xsh+-1,1) -- (\xsh-1-\wirelen,1);

    \tikzset{VertexStyle/.style = {shape=circle, fill=black, minimum size=5pt, inner sep=1pt, draw}}
    \Vertex[x=\xsh+1,y=0,NoLabel]{a3}
    \Vertex[x=\xsh+1,y=1,NoLabel]{b3}
    \Vertex[x=\xsh-1,y=0,NoLabel]{c3}
    \Vertex[x=\xsh-1,y=1,NoLabel]{e3}

    \tikzset{VertexStyle/.style = {shape=rectangle, fill=black, minimum size=4pt, inner sep=1pt, draw}}
    \Vertex[x=\xsh-1,y=0.5,NoLabel]{ax3}
    \Vertex[x=\xsh+0,y=0.5,NoLabel]{bx3}

    \Edges(e3,c3,b3)
    \Edge(e3)(a3)
\end{tikzpicture}
    \end{subfigure}
    \hspace{0.6cm}
    \begin{subfigure}{0.25\textwidth}
        \begin{tikzpicture}[scale=0.7]
    \GraphInit[vstyle=Classic]
    \SetUpEdge[style=-]
    \SetVertexMath

    \def\wirelen{0.7}

    \draw[thin, color=gray] (-1-\wirelen,0) -- (-1,0);
    \draw[thin, color=gray] (1,0) -- (1+\wirelen,0);
    \draw[thin, color=gray] (1,1) -- (1+\wirelen,1);

    \tikzset{VertexStyle/.style = {shape=circle, fill=black, minimum size=5pt, inner sep=1pt, draw}}
    \Vertex[x=-1,y=0,NoLabel]{a3}
    \Vertex[x=-1,y=1,NoLabel]{b3}
    \Vertex[x=1,y=0,NoLabel]{c3}
    \Vertex[x=1,y=1,NoLabel]{e3}

    \node at (0, 1.5) {$\mathbf{L}^{\dagger}$};

    \tikzset{VertexStyle/.style = {shape=rectangle, fill=black, minimum size=4pt, inner sep=1pt, draw}}
    \Vertex[x=1,y=0.5,NoLabel]{ax3}
    \Vertex[x=0,y=0.5,NoLabel]{bx3}

    \Edges(e3,c3,b3)
    \Edge(e3)(a3)
\end{tikzpicture}
        \vspace{1cm}
    \end{subfigure}
    \end{center}
    \caption{Operations on gadgets $\k$ and $\mathbf{L}$ in the context of $\plholant(\fc \mid \eq)$.}
    \label{fig:gadgetops}
\end{figure}
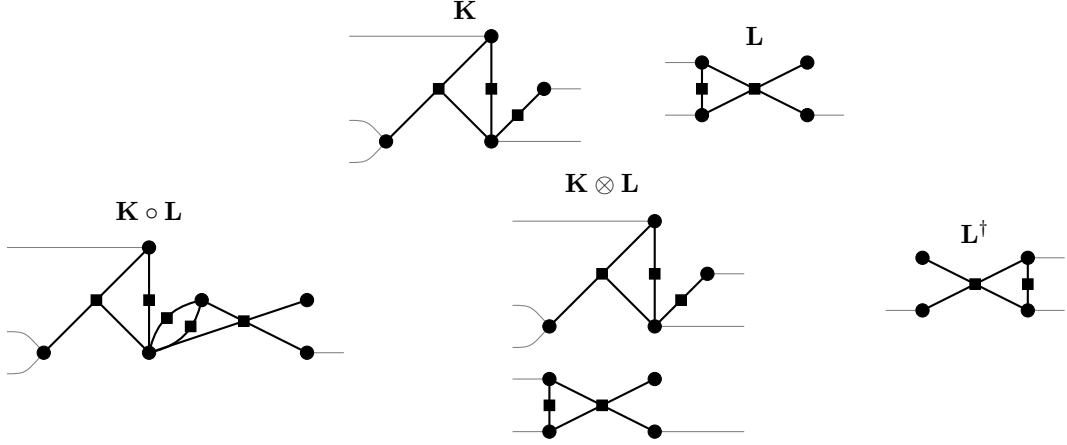

We will work with gadgets in the context
of $\holant(\fc \,|\, \eq)$. Recall from \autoref{def:gadget} that all dangling edges of such a gadget are attached to equality vertices. Normally, then, composing two such
gadgets would violate
bipartiteness, since the new edge formed by merging the dangling
edges would be between two equality vertices.
However, all edges incident to adjacent equalities must take the same value, so contracting the edge in the special case of
composition in \autoref{def:gadgetops} preserves the Holant
value, and preserves bipartiteness.

The next lemma shows that the above operations on gadgets correspond to the same operations on their
signature matrices. See e.g. \cite{cai_chen_2017}.
\begin{lemma}
    \label{lem:gadgetmatrix}
    \quad
    \begin{itemize}
        \item For $\mathbf{K_1} \in \gf(m,d), \mathbf{K_2} \in \gf(d,w)$, $M(\mathbf{K_1} \circ \mathbf{K_2}) = M(\mathbf{K_1})M(\mathbf{K_2})$;
        \item For $\mathbf{K_1} \in \gf(m_1, d_1), \mathbf{K_2} \in \gf(m_2, d_2)$, $M(\mathbf{K_1} \otimes \mathbf{K_2}) = M(\mathbf{K_1}) \otimes M(\mathbf{K_2})$;
        \item For $\k \in \gf(m,d)$, $M(\mathbf{K}^{\dagger}) = M(\mathbf{K})^{\dagger}$;
    \end{itemize}
\end{lemma}

\begin{definition}
\label{def:planargadget}
    Gadget $\k$ is \emph{planar} if its underlying multigraph has an embedding such that no
    edges (dangling or not) cross, and the dangling edges are in cyclic order in the outer face.
\end{definition}
    
For a plane embedding of a planar gadget, its dangling edges are considered in counterclockwise
cyclic order. We draw an embedded gadget's
output dangling edges on the left in order from top to
bottom (hence the first dangling edge is drawn at the top left),
and its input dangling edges on the right in order from bottom
to top. When constructing a gadget's signature matrix (\autoref{def:sigmatrix}), we consider the
input dangling edges in reverse order, so from top to bottom. 
When we compose $\mathbf{K_1} \circ \mathbf{K_2}$, this input
reversal ensures we match $\mathbf{K_1}$'s top input with
$\mathbf{K_2}$'s top output and so on down to
$\mathbf{K_1}$'s bottom input with
$\mathbf{K_2}$'s bottom output, so $\circ$ preserves planarity.

$k$-labeled \#CSP$(\fc)$ instances $K$ correspond to
$\holant(\fc \mid \eq)$ gadgets $\k \in \gf(m,d)$, for any $m+d = k$.
Given $K$, with corresponding $\holant(\fc \mid \eq)$ signature grid $\Omega_K$, attach
a dangling edge to the equality vertex in $\Omega_K$ representing variable $v$ for every label on $v$, matching the
dangling edge order with $K$'s label order. Letting $m$ dangling edges be outputs and $d$ be 
inputs, this creates gadget $\k \in \gf(m,d)$. For 
$\vx \in V(\fc)^{m}$ and $\vy \in V(\fc)^{d}$, and $\psi: [k] \to V(\fc)$ defined by
$\psi(i) = x_i$ for $1 \leq i \leq m$ and $\psi(i) = y_i$ for $m < i \leq k$,
we have $M(\k)_{\vx,\vy} = Z^{\psi}(K)$. In particular, for $\k \in \gf(1,0)$ and $x \in V(\fc)$,
\begin{equation}
    Z^x(K) = M(\k)_x. \label{eq:onelabelequiv}
\end{equation}

\section{The Planar Gadget Decomposition}
\label{sec:decomposition}

Throughout, let $F \in \c^{[q]^n}$ denote a constraint function
in $\fc$, a set of constraint functions.
All gadgets will be in the context of $\plholant(\fc\,|\, \eq)$.

If $\fc = \{A_X\}$ for undirected, unweighted graph $X$, the $\holant(A_X \mid \eq)$ 
gadgets $\k$ corresponding to
$k$-labeled \#CSP$(\fc)$ instances as discussed at the end of \autoref{sec:preliminaries}
(see \autoref{fig:gadget_vs_bilabeled}, but with the addition of dangling edges) can be identified
with the ``bi-labeled graphs'' introduced in \cite{planar}. Furthermore, the ``homomorphism
matrices'' of these bi-labeled graphs are exactly the signature matrices of these gadgets, and the
three bi-labeled graph operations $\circ, \otimes, *$ correspond to our gadget operations
$\circ, \otimes, \dagger$ (all dangling edges are incident to equality vertices, hence are always
contracted during composition, matching bi-labeled graph composition). The one difference is that
bi-labeled graphs do not contain the 
degree-2 vertices assigned $A_X$ in our $\holant(A_X \mid \eq)$ gadgets. 
However, subdividing the edges of a bi-labeled graph to create the corresponding gadget
preserves planarity of the underlying
graph. Thus planar bi-labeled graphs (those satisfying \cite[{Definition 5.3}]{planar}) correspond to
planar gadgets. We will make implicit use of this fact throughout this section.

Generalizing graph homomorphism to \#CSP entails replacing $A_X$ with an arbitrary set $\fc$ of 
constraint functions, and replacing the degree-2 vertices assigned $A_X$ with arbitrary-degree
constraint vertices assigned signatures from $\fc$. This adds a layer of complexity to the below planar
gadget decomposition relative to the bi-labeled graph decomposition in \cite{planar}. The latter
decomposition treated the (hidden) degree-2 $A_X$ vertices as edges and extracted only equality
vertices, but higher-degree constraint vertices must now be extracted separately.

\begin{definition}[$\pn$, $\pn(m,d)$]
    Let $\pn$ be the set of all planar gadgets in the context of
    $\holant(\fc\,|\,\eq)$. Recall that all dangling edges of
    such a gadget are attached to vertices assigned signatures
    in $\eq$.

    For $m,d \geq 0$, let $\pn(m,d) \subseteq \pn$ be the subset of gadgets with $m$
    output dangling edges and $d$ input dangling edges.
\end{definition}

Unlike the special case of (undirected) graph homomorphism, our constraint functions can be
\emph{asymmetric}: their
value is not necessarily preserved under permutation of their inputs. 
This adds a layer of complexity to our planar gadget decomposition relative to the analogous
decomposition in \cite{planar}.
By planarity, permutations that cross
$F$'s input edges are not allowed, but the dihedral group actions -- rotations and reflection -- are
allowed, producing a new constraint function. 
\begin{definition}[$F^{(r)}$, $F^{\top}$, $F^{\dagger}$]
    For $F \in \c^{[q]^n}$, define $F^{(0)} = F$, and for $1 \leq r \leq n-1$ define the $r$th rotation
    $F^{(r)} \in \c^{[q]^n}$ by
    \[
        F^{(r)}_{x_1^n} = F_{x_{r+1}^n x_1^r}.
    \]
    Define the reflection $F^{\top}$ by $(F^{\top})_{x_1^n} = F_{x_n^1}$, and define $F^{\dagger} = \overline{F^{\top}}$
\end{definition}
Observe that $\left(F^{(r)}\right)^{(n-r)} = F$.
The following fact motivates the notation $F^{\top}$: 
For any $m+d = \arity(F)$ and $x_1^d, y_1^m$,
\[
    \left(F^{m,d}\right)^{\top}_{x_1^d,y_1^m}
    = \left(F^{m,d}\right)_{y_1^m,x_1^d}
    = F_{y_1^mx_d^1}
    = (F^{\top})_{x_1^dy_m^1}
    = (F^{\top})^{d,m}_{x_1^d,y_1^m},
\]
hence, for any $m+d = \arity(F)$,
\begin{equation}
    \label{eq:transposeequiv}
    \left(F^{m,d}\right)^{\top} = (F^{\top})^{d,m},
    \qquad
    \left(F^{m,d}\right)^{\dagger} = (F^{\dagger})^{d,m},
\end{equation}
where the second equation follows by taking the entrywise
conjugate of the first.

We now define some useful special gadgets for $\holant(\fc \,|\, \eq)$.
\begin{definition}[$\e^{m,d}$, $(\f^{(r)})^{m,d}$, $\f^{m,d}$]
    For $m,d \geq 0$, let $\e^{m,d}$ be the gadget consisting of a single vertex, assigned $E_{m+d}$, with
    $m$ output and $d$ input dangling edges.

    For $m,d \geq 0$ and $(m+d)$-ary signature function $F$, let $(\f^{(r)})^{m,d}$ be the gadget consisting of a central degree-$(m+d)$ vertex assigned $F$, adjacent
    to $m$ vertices $u_1, \ldots, u_m$ drawn from top to bottom on the left and $d$ vertices $v_1, \ldots, v_m$ drawn from bottom to top on the
    right. The central vertex's inputs proceed counterclockwise, and its first input is the edge to
    $u_r$ if $r \leq m$ or to $v_{r-m}$ if $r > m$.
    The $i$th output and input dangling edges are incident to $u_i$ and $v_i$, respectively,
    and each $u_i$ and $v_i$ is assigned $E_2$.
    Call the $u_i$ and $v_i$ vertices and their edges to
    the central vertex the \emph{arms} of $\f^{m,d}$.

    Define $\f^{m,d} = (\f^{(0)})^{m,d}$
    and $\ii = \e^{1,1}$.
\end{definition}

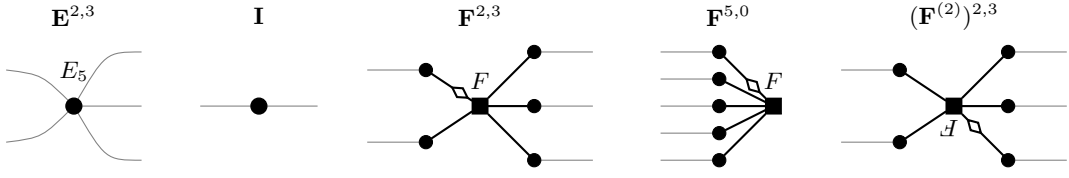
\begin{figure}[ht!]
    \center
    \begin{tikzpicture}[scale=.6]
    \tikzstyle{every node}=[font=\small]
    \GraphInit[vstyle=Classic]
    \SetUpEdge[style=-]
    \SetVertexMath

    \def\ary{1.2}
    \def\bry{0.4}
    \def\cry{-0.4}
    \def\dry{-1.2}

    \def\aly{0.8}
    \def\bly{0}
    \def\cly{-0.8}

    \def\sx{0}

    \def\xsh{9}

    \draw[thin, color=gray] (\sx, 0) .. controls (\sx+0.7, \ary) .. (\sx+1.5, \ary);
    \draw[thin, color=gray] (\sx, 0) .. controls (\sx+0.7, 0) .. (\sx+1.5, 0);
    \draw[thin, color=gray] (\sx, 0) .. controls (\sx+0.7, -\ary) .. (\sx+1.5, -\ary);

    \draw[thin, color=gray] (\sx-1.5, \aly) .. controls (\sx-0.7, \aly-0.1) .. (\sx, 0);
    \draw[thin, color=gray] (\sx-1.5, \cly) .. controls (\sx-0.7, \cly+0.1) .. (\sx, 0);

    \Vertex[x=\sx,y=0,L={E_5},Lpos=90,Ldist=0.1cm]{v22}

    \node at (\sx, \ary+0.8) {$\e^{2,3}$};

    \def\ix{2.8}
    \draw[thin, color=gray] (\ix, 0) -- (\ix+2.6, 0);
    \Vertex[x=\ix+1.3, y=0, NoLabel]{iv}

    \node at (\ix+1.3,\ary+0.8) {$\ii$};
    \draw[thin, color=gray] (\xsh+\sx+1.2, \ary) -- (\xsh+\sx+2.5, \ary);
    \draw[thin, color=gray] (\xsh+\sx+1.2, 0) -- (\xsh+\sx+2.5, 0);
    \draw[thin, color=gray] (\xsh+\sx+1.2, \dry) -- (\xsh+\sx+2.5, \dry);

    \draw[thin, color=gray] (\xsh+\sx-2.5, \aly) -- (\xsh+\sx-1.2, \aly);
    \draw[thin, color=gray] (\xsh+\sx-2.5, \cly) -- (\xsh+\sx-1.2, \cly);

    \tikzset{VertexStyle/.style = {shape=rectangle, fill=black, minimum size=6pt, inner sep=1pt, draw}}
    \Vertex[x=\xsh+\sx,y=0,L={F},Lpos=90]{v2}
    \tikzset{VertexStyle/.style = {shape=circle, fill=black, minimum size=5pt, inner sep=1pt, draw}}

    \Vertex[x=\xsh+\sx+1.2,y=\ary,NoLabel,Lpos=300,Ldist=-0.1cm]{f1}
    \Vertex[x=\xsh+\sx+1.2,y=0,NoLabel,Lpos=300,Ldist=-0.1cm]{f2}
    \Vertex[x=\xsh+\sx+1.2,y=\dry,NoLabel,Lpos=300,Ldist=-0.1cm]{f3}

    \Vertex[x=\xsh+\sx-1.2,y=\aly,NoLabel,Lpos=300,Ldist=-0.1cm]{l1}
    \Vertex[x=\xsh+\sx-1.2,y=\cly,NoLabel,Lpos=300,Ldist=-0.1cm]{l2}
    \Edge(v2)(f1)
    \Edge(v2)(f2)
    \Edge(v2)(f3)
    \DEdge{v2}{l1}
    \Edge(v2)(l2)

    \node at (\sx + \xsh, \ary+0.8) {$\f^{2,3}$};

    \def\fnx{15.5}

    \draw[thin, color=gray] (\fnx+\sx-2.5, \ary) -- (\fnx+\sx-1.2, \ary);
    \draw[thin, color=gray] (\fnx+\sx-2.5, 0) -- (\fnx+\sx-1.2, 0);
    \draw[thin, color=gray] (\fnx+\sx-2.5, \dry) -- (\fnx+\sx-1.2, \dry);

    \draw[thin, color=gray] (\fnx+\sx-2.5, \aly-0.2) -- (\fnx+\sx-1.2, \aly-0.2);
    \draw[thin, color=gray] (\fnx+\sx-2.5, \cly+0.2) -- (\fnx+\sx-1.2, \cly+0.2);

    \tikzset{VertexStyle/.style = {shape=rectangle, fill=black, minimum size=6pt, inner sep=1pt, draw}}
    \Vertex[x=\fnx+\sx,y=0,L={F},Lpos=90]{v3}
    \tikzset{VertexStyle/.style = {shape=circle, fill=black, minimum size=5pt, inner sep=1pt, draw}}

    \Vertex[x=\fnx+\sx-1.2,y=\ary,NoLabel,Lpos=300,Ldist=-0.1cm]{ff1}
    \Vertex[x=\fnx+\sx-1.2,y=0,NoLabel,Lpos=300,Ldist=-0.1cm]{ff2}
    \Vertex[x=\fnx+\sx-1.2,y=\dry,NoLabel,Lpos=300,Ldist=-0.1cm]{ff3}

    \Vertex[x=\fnx+\sx-1.2,y=\aly-0.2,NoLabel,Lpos=300,Ldist=-0.1cm]{ll1}
    \Vertex[x=\fnx+\sx-1.2,y=\cly+0.2,NoLabel,Lpos=300,Ldist=-0.1cm]{ll2}
    \DEdge{v3}{ff1}
    \Edge(v3)(ff2)
    \Edge(v3)(ff3)
    \Edge(v3)(ll1)
    \Edge(v3)(ll2)

    \node at (\sx + \fnx-1, \ary+0.8) {$\f^{5,0}$};

    \def\frx{\fnx+\sx+4}
    \draw[thin, color=gray] (\frx+1.2, \ary) -- (\frx+2.5, \ary);
    \draw[thin, color=gray] (\frx+1.2, 0) -- (\frx+2.5, 0);
    \draw[thin, color=gray] (\frx+1.2, \dry) -- (\frx+2.5, \dry);

    \draw[thin, color=gray] (\frx-2.5, \aly) -- (\frx-1.2, \aly);
    \draw[thin, color=gray] (\frx-2.5, \cly) -- (\frx-1.2, \cly);

    \tikzset{VertexStyle/.style = {shape=rectangle, fill=black, minimum size=6pt, inner sep=1pt, draw}}
    \Vertex[x=\frx,y=0,NoLabel]{v5}
    \node[label={[rotate=180]$F$}] at (\frx-0.1,-0.3) {};
    \tikzset{VertexStyle/.style = {shape=circle, fill=black, minimum size=5pt, inner sep=1pt, draw}}

    \Vertex[x=\frx+1.2,y=\ary,NoLabel,Lpos=300,Ldist=-0.1cm]{fff1}
    \Vertex[x=\frx+1.2,y=0,NoLabel,Lpos=300,Ldist=-0.1cm]{fff2}
    \Vertex[x=\frx+1.2,y=\dry,NoLabel,Lpos=300,Ldist=-0.1cm]{fff3}

    \Vertex[x=\frx-1.2,y=\aly,NoLabel,Lpos=300,Ldist=-0.1cm]{lll1}
    \Vertex[x=\frx-1.2,y=\cly,NoLabel,Lpos=300,Ldist=-0.1cm]{lll2}
    \Edge(v5)(fff1)
    \Edge(v5)(fff2)
    \DEdge{v5}{fff3}
    \Edge(v5)(lll1)
    \Edge(v5)(lll2)

    \node at (\frx, \ary+0.8) {$(\f^{(2)})^{2,3}$};
\end{tikzpicture}
    \caption{Examples of the fundamental gadgets $\e^{m,d}$, $\f^{m,d}$, $(\f^{(r)})^{m,d}$. The diamond indicates the
    first input to asymmetric $F$.}
    \label{fig:eandf}
\end{figure}

See \autoref{fig:eandf}. As with all gadgets, $(\f^{(r)})^{m,d}$'s inputs proceed counterclockwise for
the purpose of organizing its signature matrix, and the first input to $(\f^{(r)})^{m,d}$
is the top output dangling edge, attached to
vertex $u_1$. When $r=0$, this agrees with the defined input order of the central vertex of $(\f^{(0)})^{m,d} = \f^{m,d}$. As always, we may ignore
vertices assigned $E_2$ when computing the Holant value,
so $M(\f^{m,d}) = F^{m,d}$. When central vertex's orientation is rotated by $r$ edges, we have
$M\left((\f^{(r)})^{m,d}\right) = (F^{(r)})^{m,d}$.
$F^{(r)}$ is generally not the same
constraint function as $F$, so we will need to address this
subtlety in the proofs of \autoref{thm:generatepn} and
\autoref{thm:quantumholant}.

Since $E_{m+d}$ is completely symmetric in the order of its inputs, this subtlety does not apply to $\e^{m,d}$. We
have $M(\e^{m,d}) = E^{m,d}$ regardless of the orientation
of the central vertex.

The central vertex in the gadget $(\f^{m,d})^{\dagger} \in \pn(d,m)$ can be viewed either as hosting
constraint function $\overline{F}$ oriented clockwise, or $F^{\dagger}$ oriented counterclockwise.
See \autoref{fig:dagger}. However, recall that, for the purpose of organizing the signature matrix,
we always consider the \emph{gadget's} inputs on the dangling edges counterclockwise, which is reversed
relative to the clockwise
inputs to the \emph{central vertex} in the former case. So $M((\f^{m,d})^{\dagger}) = (F^{m,d})^\dagger
= (F^{\dagger})^{m,d}$ by \eqref{eq:transposeequiv}.

\begin{figure}[ht!]
    \center
    \begin{tikzpicture}[scale=.6]
    \tikzstyle{every node}=[font=\small]
    \GraphInit[vstyle=Classic]
    \SetUpEdge[style=-]
    \SetVertexMath

    \def\ary{1.2}
    \def\bry{0.4}
    \def\cry{-0.4}
    \def\dry{-1.2}

    \def\aly{0.8}
    \def\bly{0}
    \def\cly{-0.8}

    \def\sx{0}

    \def\xsh{9}

    \draw[thin, color=gray] (\xsh+\sx-2.5, \ary) -- (\xsh+\sx-1.2, \ary);
    \draw[thin, color=gray] (\xsh+\sx-2.5, 0) -- (\xsh+\sx-1.2, 0);
    \draw[thin, color=gray] (\xsh+\sx-2.5, \dry) -- (\xsh+\sx-1.2, \dry);

    \draw[thin, color=gray] (\xsh+\sx+1.2, \aly) -- (\xsh+\sx+2.5, \aly);
    \draw[thin, color=gray] (\xsh+\sx+1.2, \cly) -- (\xsh+\sx+2.5, \cly);

    \tikzset{VertexStyle/.style = {shape=rectangle, fill=black, minimum size=6pt, inner sep=1pt, draw}}
    \Vertex[x=\xsh+\sx,y=0,L={\reflectbox{$\overline{F}$}},Lpos=90]{v2}
    \tikzset{VertexStyle/.style = {shape=circle, fill=black, minimum size=5pt, inner sep=1pt, draw}}

    \Vertex[x=\xsh+\sx-1.2,y=\ary,NoLabel,Lpos=300,Ldist=-0.1cm]{f1}
    \Vertex[x=\xsh+\sx-1.2,y=0,NoLabel,Lpos=300,Ldist=-0.1cm]{f2}
    \Vertex[x=\xsh+\sx-1.2,y=\dry,NoLabel,Lpos=300,Ldist=-0.1cm]{f3}

    \Vertex[x=\xsh+\sx+1.2,y=\aly,NoLabel,Lpos=300,Ldist=-0.1cm]{l1}
    \Vertex[x=\xsh+\sx+1.2,y=\cly,NoLabel,Lpos=300,Ldist=-0.1cm]{l2}
    \Edge(v2)(f1)
    \Edge(v2)(f2)
    \Edge(v2)(f3)
    \DEdge{v2}{l1}
    \Edge(v2)(l2)

    \node at (\sx + \xsh, \ary+0.8) {$(\f^{2,3})^\dagger$};

    \def\frx{\sx+9+\xsh}
    \draw[thin, color=gray] (\frx-2.5, \ary) -- (\frx-1.2, \ary);
    \draw[thin, color=gray] (\frx-2.5, 0) -- (\frx-1.2, 0);
    \draw[thin, color=gray] (\frx-2.5, \dry) -- (\frx-1.2, \dry);

    \draw[thin, color=gray] (\frx+1.2, \aly) -- (\frx+2.5, \aly);
    \draw[thin, color=gray] (\frx+1.2, \cly) -- (\frx+2.5, \cly);

    \tikzset{VertexStyle/.style = {shape=rectangle, fill=black, minimum size=6pt, inner sep=1pt, draw}}
    \Vertex[x=\frx,y=0,L={F^\dagger},Lpos=90]{v5}
    \tikzset{VertexStyle/.style = {shape=circle, fill=black, minimum size=5pt, inner sep=1pt, draw}}

    \Vertex[x=\frx-1.2,y=\ary,NoLabel,Lpos=300,Ldist=-0.1cm]{fff1}
    \Vertex[x=\frx-1.2,y=0,NoLabel,Lpos=300,Ldist=-0.1cm]{fff2}
    \Vertex[x=\frx-1.2,y=\dry,NoLabel,Lpos=300,Ldist=-0.1cm]{fff3}

    \Vertex[x=\frx+1.2,y=\aly,NoLabel,Lpos=300,Ldist=-0.1cm]{lll1}
    \Vertex[x=\frx+1.2,y=\cly,NoLabel,Lpos=300,Ldist=-0.1cm]{lll2}
    \Edge(v5)(fff3)
    \Edge(v5)(fff2)
    \DEdge{v5}{fff1}
    \Edge(v5)(lll1)
    \Edge(v5)(lll2)

    \node at (\frx, \ary+0.8) {$(\f^{\dagger})^{2,3}$};

    \node at (\sx+4.5+\xsh, 0) {$=$};
\end{tikzpicture}
    \caption{Equivalence between the gadgets $(\f^{2,3})^\dagger$ and $(\f^\dagger)^{2,3}$ 
        (the latter referring to the gadget created from the signature $F^{\dagger}$). The central
        vertices are oriented clockwise and counterclockwise, respectively.}
    \label{fig:dagger}
\end{figure}
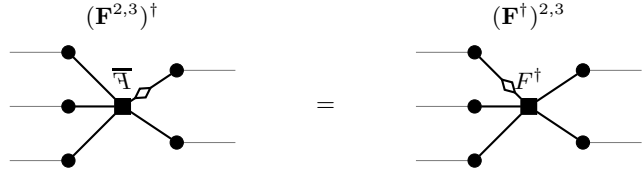

For a set of gadgets $\gf$ or matrices $\mathcal{M}$ and a set of gadget or matrix operations $\mathcal{O}$, let $\langle \gf \rangle_{\mathcal{O}}$ or $\langle \mathcal{M} \rangle_{\mathcal{O}}$ denote the closure of $\gf$ or $\mathcal{M}$ under $\mathcal{O}$, respectively.
$\e^{m,d}$ is the gadget corresponding to the bi-labeled graph called $\mathbf{M}^{m,d}$ in \cite{planar}.
Since gadget operations are equivalent to bi-labeled graph operations, we have the following.
\begin{lemma}[{\cite[{Lemma 3.18}]{planar}}]
    \label{lem:mclosure}
    For all $m,d \geq 0$, $\e^{m,d} \in \tcwdn{\e^{1,0}, \e^{1,2}}$.
\end{lemma}

We next show that we can apply arbitrary rotations to
$\f^{m,d}$ and pivot any number of its dangling edges between input and output.
\begin{lemma}
    \label{lem:rotategadget}
    For any $n$-ary constraint function $F$,
    $(\f^{(r)})^{m_2,d_2} \in \tcwdn{\e^{1,0}, \e^{1,2}, \f^{m_1,d_1}}$ 
    for all $r,m_1,d_1,m_2,d_2 \geq 0$ with $m_1+d_1 = n$, $m_2+d_2 = n$.
\end{lemma}
\begin{proof}
    By \autoref{lem:mclosure}, $\e^{2,0}, \e^{0,2}, \ii \in \tcwdn{\e^{1,0}, \e^{1,2}}$. 
    If $d_1 > d_2$ and $m_2 > m_1$, we use $d_1-d_2 = m_2-m_1$
    horizontally nested copies of $\e^{2,0}$ to pivot the
    bottom $m_2-m_1$ input dangling edges to become outputs. Indeed,
    \begin{align*}
        \f^{m_2,d_2} = (\f^{m_1,d_1} \otimes \ii^{\otimes d_1-d_2})
        &\circ (\ii^{\otimes d_1-1} \otimes \e^{2,0} \otimes \ii^{\otimes d_1-d_2-1})\\
        &\circ \ldots
        \circ (\ii^{\otimes d_1-(d_1-d_2-1)} \otimes \e^{2,0} \otimes \ii^{\otimes 1})
        \circ (\ii^{\otimes d_2} \otimes \e^{2,0}).
    \end{align*}
    See \autoref{fig:reversewire} for an illustration. Similarly, if $m_1 > m_2$ and $d_2 > d_1$, 
    we apply horizontally reflected reasoning, using $\e^{2,0}$, to pivot the
    bottom $d_2-d_1$ output dangling edges to become inputs.
    \begin{align*}
        \f^{m_2,d_2} = (\ii^{\otimes m_2} \otimes \e^{0,2})
        &\circ (\ii^{\otimes d_1-(m_1-m_2-1)} \otimes \e^{0,2} \otimes\ii^{\otimes 1}) \\
        &\circ \ldots
        \circ (\ii^{\otimes m_1-1} \otimes \e^{0,2} \otimes \ii^{\otimes m_1-m_2-1})
        \circ (\f^{m_1,d_1} \otimes \ii^{\otimes m_1-m_2}).
    \end{align*}
    Applying the vertical reflection of the reasoning for
    the above two cases, we see that we may also pivot
    the top $m$ outputs to become inputs, and the top $m$
    inputs to be outputs, for any $m$. This moves the
    center vertex's first input from the edge incident to the top output
    dangling edge, inducing $(r)$-type rotation. 
    Combining these four operations -- rotating inputs to outputs or outputs to inputs on the top or bottom -- we obtain arbitrary rotations of the gadget's arms, giving any $(r)$, and, given any such
    rotation, arbitrary partitions of the arms into
    contiguous output and input dangling edges, giving
    any $m_2$ and $d_2$.
\end{proof}

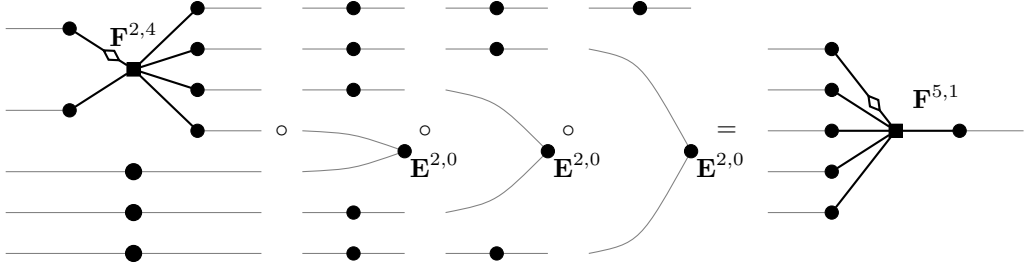
\begin{figure}[ht!]
    \center
    \begin{tikzpicture}[scale=.68]
    \GraphInit[vstyle=Classic]
    \SetUpEdge[style=-]
    \SetVertexMath

    \def\wirelen{1}
    \def\fwirelen{2.5}
    \def\wiregap{0.8}
    \def\xa{0}
    \def\xb{\xa + 2*\fwirelen + \wiregap}
    \def\xc{\xb + 2*\wirelen + \wiregap}
    \def\xd{\xc + 2*\wirelen + \wiregap}

    \def\ygap{0.8}
    \def\ya{0}
    \def\yb{\ya + \ygap}
    \def\yc{\yb + \ygap}
    \def\yd{\yc + \ygap}
    \def\ye{\yd + \ygap}
    \def\yf{\ye + \ygap}
    \def\yg{\yf + \ygap}

    \def\twirelen{2 * \wirelen}
    \def\evertexy{\yc + \ygap/2}

    \draw[thin, color=gray] (\xa, \ya) -- ({\xa + (2*\fwirelen)}, \ya);
    \draw[thin, color=gray] (\xa, \yb) -- ({\xa + (2*\fwirelen)}, \yb);
    \draw[thin, color=gray] (\xa, \yc) -- ({\xa + (2*\fwirelen)}, \yc);
    \Vertex[x=\xa+\fwirelen,y=\ya,NoLabel]{xaa}
    \Vertex[x=\xa+\fwirelen,y=\yb,NoLabel]{xab}
    \Vertex[x=\xa+\fwirelen,y=\yc,NoLabel]{xac}
    \draw[thin, color=gray] ({\xa+\fwirelen/2},{\yf+\ygap/2}) -- ({\xa},{\yf+\ygap/2});
    \draw[thin, color=gray] ({\xa+\fwirelen/2},{\yd+\ygap/2}) -- ({\xa},{\yd+\ygap/2});
    \draw[thin, color=gray] ({\xa+3*\fwirelen/2},\yd) -- ({\xa+\fwirelen*2},\yd);
    \draw[thin, color=gray] ({\xa+3*\fwirelen/2},\ye) -- ({\xa+\fwirelen*2},\ye);
    \draw[thin, color=gray] ({\xa+3*\fwirelen/2},\yf) -- ({\xa+\fwirelen*2},\yf);
    \draw[thin, color=gray] ({\xa+3*\fwirelen/2},\yg) -- ({\xa+\fwirelen*2},\yg);
    \tikzset{VertexStyle/.style = {shape=rectangle, fill=black, minimum size=5pt, inner sep=1pt, draw}}
    \Vertex[x=\xa+\fwirelen,y={\ye+\ygap/2},L={\f^{2,4}},Lpos=90,Ldist=0.1cm]{e}
    \tikzset{VertexStyle/.style = {shape=circle, fill=black, minimum size=5pt, inner sep=1pt, draw}}
    \Vertex[x={\xa+\fwirelen/2},y={\yf+\ygap/2},NoLabel]{el1}
    \Vertex[x={\xa+\fwirelen/2},y={\yd+\ygap/2},NoLabel]{el2}
    \Vertex[x={\xa+3*\fwirelen/2},y=\yd,NoLabel]{er1}
    \Vertex[x={\xa+3*\fwirelen/2},y=\ye,NoLabel]{er2}
    \Vertex[x={\xa+3*\fwirelen/2},y=\yf,NoLabel]{er3}
    \Vertex[x={\xa+3*\fwirelen/2},y=\yg,NoLabel]{er4}
    \DEdge{e}{el1}
    \Edge(e)(el2)
    \Edge(e)(er1)
    \Edge(e)(er2)
    \Edge(e)(er3)
    \Edge(e)(er4)

    \draw[thin, color=gray] (\xb, \ya) -- ({\xb + (2*\wirelen)}, \ya);
    \draw[thin, color=gray] (\xb, \yb) -- ({\xb + (2*\wirelen)}, \yb);
    \draw[thin, color=gray] (\xb, \ye) -- ({\xb + (2*\wirelen)}, \ye);
    \draw[thin, color=gray] (\xb, \yf) -- ({\xb + (2*\wirelen)}, \yf);
    \draw[thin, color=gray] (\xb, \yg) -- ({\xb + (2*\wirelen)}, \yg);
    \draw[thin, color=gray] (\xb, \yc) .. controls (\xb+\wirelen,\yc+0.05) .. ({\xb + (2*\wirelen)}, \evertexy);
    \draw[thin, color=gray] (\xb, \yd) .. controls (\xb+\wirelen,\yd-0.05) .. ({\xb + (2*\wirelen)}, \evertexy);
    \Vertex[x=\xb+\wirelen,y=\ya,NoLabel]{xba}
    \Vertex[x=\xb+\wirelen,y=\yb,NoLabel]{xbb}
    \Vertex[x=\xb+\twirelen,y=\evertexy,L={\e^{2,0}},Lpos=300,Ldist=-0.2cm]{xbm}
    \Vertex[x=\xb+\wirelen,y=\ye,NoLabel]{xbe}
    \Vertex[x=\xb+\wirelen,y=\yf,NoLabel]{xbf}
    \Vertex[x=\xb+\wirelen,y=\yg,NoLabel]{xbg}

    \draw[thin, color=gray] (\xc, \ya) -- ({\xc + (2*\wirelen)}, \ya);
    \draw[thin, color=gray] (\xc, \yf) -- ({\xc + (2*\wirelen)}, \yf);
    \draw[thin, color=gray] (\xc, \yg) -- ({\xc + (2*\wirelen)}, \yg);
    \draw[thin, color=gray] (\xc, \yb) .. controls (\xc+\wirelen,\yb+0.2) .. ({\xc + (2*\wirelen)}, \evertexy);
    \draw[thin, color=gray] (\xc, \ye) .. controls (\xc+\wirelen,\ye-0.2) .. ({\xc + (2*\wirelen)}, \evertexy);
    \Vertex[x=\xc+\wirelen,y=\ya,NoLabel]{xca}
    \Vertex[x=\xc+\twirelen,y=\evertexy,L={\e^{2,0}},Lpos=300,Ldist=-0.2cm]{xcm}
    \Vertex[x=\xc+\wirelen,y=\yf,NoLabel]{xcf}
    \Vertex[x=\xc+\wirelen,y=\yg,NoLabel]{xcg}

    \draw[thin, color=gray] (\xd, \yg) -- ({\xd + (2*\wirelen)}, \yg);
    \draw[thin, color=gray] (\xd, \ya) .. controls (\xd+\wirelen,\ya+0.2) .. ({\xd + (2*\wirelen)}, \evertexy);
    \draw[thin, color=gray] (\xd, \yf) .. controls (\xd+\wirelen,\yf-0.2) .. ({\xd + (2*\wirelen)}, \evertexy);
    \Vertex[x={\xd+\twirelen},y=\evertexy,L={\e^{2,0}},Lpos=300,Ldist=-0.2cm]{xdm}
    \Vertex[x=\xd+\wirelen,y=\yg,NoLabel]{xdg}

    \def\xr{\xd+3.5}
    \draw[thin, color=gray] ({\xr+\fwirelen/2},\yb) -- (\xr,\yb);
    \draw[thin, color=gray] ({\xr+\fwirelen/2},\yc) -- (\xr,\yc);
    \draw[thin, color=gray] ({\xr+\fwirelen/2},\yd) -- (\xr,\yd);
    \draw[thin, color=gray] ({\xr+\fwirelen/2},\ye) -- (\xr,\ye);
    \draw[thin, color=gray] ({\xr+\fwirelen/2},\yf) -- (\xr,\yf);
    \draw[thin, color=gray] ({\xr+3*\fwirelen/2},\yd) -- (\xr+2*\fwirelen,\yd);
    \Vertex[x=\xr+\fwirelen/2,y=\yb,NoLabel]{xra}
    \Vertex[x=\xr+\fwirelen/2,y=\yc,NoLabel]{xrb}
    \Vertex[x=\xr+\fwirelen/2,y=\yd,NoLabel]{xrc}
    \Vertex[x=\xr+\fwirelen/2,y=\ye,NoLabel]{xrd}
    \Vertex[x=\xr+\fwirelen/2,y=\yf,NoLabel]{xre}
    \tikzset{VertexStyle/.style = {shape=rectangle, fill=black, minimum size=5pt, inner sep=1pt, draw}}
    \Vertex[x=\xr+\fwirelen,y=\yd,L={\f^{5,1}},Lpos=60,Ldist=0.1cm]{xr}
    \tikzset{VertexStyle/.style = {shape=circle, fill=black, minimum size=5pt, inner sep=1pt, draw}}
    \Vertex[x=\xr+3*\fwirelen/2,y=\yd,NoLabel]{xrr}
    \Edge(xr)(xra)
    \Edge(xr)(xrb)
    \Edge(xr)(xrc)
    \Edge(xr)(xrd)
    \DEdge{xr}{xre}
    \Edge(xr)(xrr)

    \node at (\xb - \wiregap/2, \yd) {$\circ$};
    \node at (\xc - \wiregap/2, \yd) {$\circ$};
    \node at (\xd - \wiregap/2, \yd) {$\circ$};
    \node at (\xr - 0.8, \yd) {$=$};

\end{tikzpicture}
    \caption{$(\f^{2,4} \otimes \ii^{\otimes 3})
    \circ (\ii^{\otimes 3} \otimes \e^{2,0} \otimes \ii^{\otimes 2}) \circ (\ii^{\otimes 2} \otimes \e^{2,0} \otimes \ii)
    \circ (\ii \otimes \e^{2,0}) = \f^{5,1}$.}
    \label{fig:reversewire}
\end{figure}

Next, we show that we can construct all gadgets in $\pn$
using only the elementary gadgets $\e^{1,0}, \e^{1,2}$, and $\{\f^{n,0} \mid F \in \fc\}$. This is a generalization
of \cite[{Theorem 6.7}]{planar}, which states that the class $\mathcal{P}$ of planar bi-labeled graphs
is generated by $\e^{1,0}, \e^{1,2}$, and the $(1,1)$-bi-labeled graph $\mathbf{A}$ with underlying graph
$K_2$, and the output wire attached to one vertex and the input wire to the other. The gadget corresponding
to $\mathbf{A}$ is $(\mathbf{A_X})^{1,1}$, as it has an extra vertex subdividing the only edge. Since we can convert
$(\mathbf{A_X})^{2,0}$ to $(\mathbf{A_X})^{1,1}$ and vice-versa using \autoref{lem:rotategadget}, \cite[{Theorem 6.7}]{planar} 
is analogous to the $n=2$ case of the following \autoref{thm:generatepn}.
The proof of \autoref{thm:generatepn} will closely follow the proof of this $n=2$ case,
so we need the following result, a restatement of \cite[{Corollary 6.5}]{planar}.
\begin{lemma}
    \label{lem:consecutive}
    Let $\k$ be a planar gadget with underlying multigraph $K$. There is some $v \in K$ such that the 
    dangling ends of all dangling edges incident to $v$ occur consecutively in the cyclical 
    ordering of dangling edges around the outside of $\k$.
\end{lemma}
For constraint function $F$, let $n_F = \arity(F)$.
\begin{theorem}
    \label{thm:generatepn}
    For any conjugate closed constraint function set $\fc$,
    $\pn = \tcwdn{\e^{1,0}, \e^{1,2}, \{\f^{n_F,0} \mid F \in \fc\}}$.
\end{theorem}
\begin{proof}
    It is easy to see that $\e^{1,0}, \e^{1,2}, \f^{n_F,0} \in \pn$ for every $F \in \fc$. \cite[{Lemmas 5.12-5.14}]{planar}
    show that $\otimes, \circ, {\dagger}$ preserve planarity for bi-labeled graphs, hence for gadgets.
    $\otimes$ and ${\dagger}$ clearly preserve bipartiteness and the property that dangling edges are only
    incident to equality vertices. $\circ$ also preserves these properties because all $\plholant(\fc \mid \eq)$ gadgets
    have dangling edges incident only to equality vertices,
    so $\circ$ only connects adjacent equality vertices, which
    are then merged. 
    ${\dagger}$ additionally replaces every signature $F \in \fc$ with $\overline{F}$, but $\overline{F} \in \fc$ as well
    because $\fc$ is CC.
    Thus $\pn \supseteq \tcwdn{\e^{1,0}, \e^{1,2}, \{\f^{n_F,0} \mid F \in \fc\}}$.

    To show the reverse inclusion, the idea is to decompose an arbitrary $\k \in \pn$ into
    copies of $\e^{m,d}$ and appropriate $\f^{m,d}$. Let $v$ satisfy the 
    condition in \autoref{lem:consecutive}. Roughly, if $v$ is an equality vertex, we extract 
    $\e^{m,d}$, and if $v$ is a constraint vertex assigned $F$, we extract $\f^{m,d}$. The remaining gadget $\k'$
    is also in $\pn$, so the result follows inductively.
    
    A vertex $v$ assigned an $F \in \fc$ may have arbitrary orientation in $\k$ -- its first input could be on any incident edge, and its remaining inputs could proceed clockwise or counterclockwise from this edge.
    In the counterclockwise case, we extract a gadget
    $(\f^{(r)})^{m,d}$ for some $r,m,d$. By \autoref{lem:rotategadget}, $(\f^{(r)})^{m,d} \in \tcwdn{\e^{1,0}, \e^{1,2}, \{\f^{n_F,0} \mid F \in \fc\}}$. However, \autoref{lem:rotategadget} does not address the
    clockwise/counterclockwise reflection issue. So if $v$'s
    inputs are oriented clockwise, we apply ${\dagger}$ to the
    extracted gadget with $v$ as the center vertex. Recall
    from \autoref{def:gadgetops} that this flips $v$'s
    inputs to a counterclockwise orientation, so we obtain
    a gadget $(\overline{\f}^{(r)})^{m,d}$ for some $r$.
    $\overline{F} \in \fc$, as $\fc$ is CC, so
    $(\overline{\f}^{(r)})^{m,d} \in \tcwdn{\e^{1,0}, \e^{1,2}, \{\f^{n,0} \mid F \in \fc\}}$ 
    as well. The rest of the proof only uses properties of a gadget's underlying graph, so
    we may assume WLOG that every constraint vertex is extracted as the center vertex of a gadget
    $\f^{m,d}$ for some $F \in \fc$ and $m,d \geq 0$.

    For the purpose of finding the vertex $v$ satisfying the condition in
    \autoref{lem:consecutive}, we consider the gadget $\widehat{\k}$ created from $\k$ by removing every
    vertex assigned $E_2$ with one incident dangling edge,
    and reassigning the dangling edge to the removed vertex's neighbor.
    This is necessary to extract $\f^{m,d}$s (by
    bipartiteness, any vertex assigned $E_2$ lies on an arm of an $\f^{m,d}$ gadget).
    Removing vertices in this way preserves planarity, so \autoref{lem:consecutive} applies to $\widehat{\k}$, giving
    a $v$ that is either a constraint vertex, or assigned $E_k$ for $k > 2$.

    $\s^{m,d}$ is a bi-labeled graph defined in \cite{planar} that has the same underlying graph
    as an $\f^{m,d}$, but with no $E_2$ vertices on the output wires.
    $\s^{m,d}_L$ and $\s^{m,d}_R$ have an extra input dangling edge attached to the central constraint
    vertex above and below all other input dangling edges, respectively.
    \cite[{Lemma 6.6}]{planar} states that, for $\k \in \mathcal{P}$, there exists a 
    $\k' \in \mathcal{P}$ with one fewer vertex and $n,m,d,r,t \geq 0$ satisfying one of the following:
    \begin{align*}
        &\k = (\ii^{\otimes r} \otimes \s^{m,d} \otimes \ii^{\otimes t}) \circ \k' 
        \numberthis \label{eq:left} \\
        &\k = \k' \circ (\ii^{\otimes r} \otimes \s^{m,d} \otimes \ii^{\otimes t})^{\dagger}
        \numberthis \label{eq:right} \\
        &\k = (\s^{m,d}_L \otimes \ii^{\otimes t}) \circ (\e^{1,r} \otimes \k')
        \numberthis \label{eq:top} \\
        &\k = (\ii^{\otimes t} \otimes \s^{m,d}_R) \circ (\k' \otimes \e^{1,r})
        \numberthis \label{eq:bottom} \\
        &\k = \s^{0,n} \circ \k'.
        \numberthis \label{eq:nodangling}
    \end{align*}
    (there are four other cases involving vertices with loops, but any $\k \in \pn$ is bipartite,
    hence loopless).
    $\s^{m,d}$ is constructed using $\mathbf{A}$, to which we do not have access for $n > 2$.
    We will replace it with $\e^{m,d}$ or an appropriate $\f^{m,d}$. In case \eqref{eq:left}, all of
    $v$'s consecutive dangling edges are outputs.
    Suppose $v$ has $m$ such output dangling edges and $d$ edges to other vertices in $\widehat{\k}$. 
    Also let there be $r$ output dangling edges above those incident to $v$, and $t$ below.
    We have two subcases. First, suppose $v$ is assigned $E_{m+d}$. 
    To obtain \eqref{eq:left}, we
    create $\k'$ from $\k$ by deleting $v$ and replacing
    the edges to $v$'s neighbors with output dangling edges attached to those neighbors in the proper
    order. See \autoref{fig:evss}, where $\k$ is in the upper left, with $v$ assigned $E_5$, and the
    decomposition in \eqref{eq:left} is in the bottom right. This has two problems from the gadget
    perspective. First, we cannot construct $\s^{m,d}$, and second, it leaves dangling edges incident
    to constraint vertices. We solve both problems as follows. Instead of extracting the vertices
    that end up on the input arms of $\s^{m,d}$, leave them in $\k'$ as $E_2$'s. So we extract
    $\e^{m,d}$ in place of $\s^{m,d}$, and the new output dangling edges in $\k'$ are left attached to
    the $E_2$'s. See the upper right of \autoref{fig:evss}.

    \begin{figure}[ht!]
        \center
        \begin{tikzpicture}[scale=.64]
    \tikzstyle{every node}=[font=\small]
    \GraphInit[vstyle=Classic]
    \SetUpEdge[style=-]
    \SetVertexMath

    \def\ary{1.2}
    \def\bry{0.4}
    \def\cry{-0.4}
    \def\dry{-1.2}

    \def\aly{0.8}
    \def\bly{0}
    \def\cly{-0.8}

    \def\sx{2}
    \def\kpx{10}

    \def\akvy{2.2}
    \def\bkvy{1.7}
    \def\ckvy{0.6}
    \def\ekvy{-1.1}
    \def\fkvy{-1.8}

    \filldraw[color=black!70, fill=blue!8] (\kpx,0) circle (2.5);

    \draw[thin, color=gray] (\kpx-3.5, \akvy) -- (\kpx, \akvy);
    \draw[thin, color=gray] (\kpx-3.5, \bkvy) -- (\kpx-0.7, \bkvy);
    \draw[thin, color=gray] (\kpx-3.5, \fkvy) -- (\kpx, \fkvy);

    \draw[thin, color=gray] (\kpx-3.5, \aly) .. controls (\kpx-2.2, \aly-0.1) .. (\kpx-1.5, 0);
    \draw[thin, color=gray] (\kpx-3.5, \cly) .. controls (\kpx-2.2, \cly+0.1) .. (\kpx-1.5, 0);

    \Vertex[x=\kpx-1.5,y=0.03,L={E_5},Lpos=270]{e23}

    \tikzset{VertexStyle/.style = {shape=rectangle, fill=black, minimum size=5pt, inner sep=1pt, draw}}
    \Vertex[x=\kpx+1,y=\ckvy,L={F_i},Ldist=-0.1cm]{u43}
    \Vertex[x=\kpx+1,y=\ekvy,L={F_j},Ldist=-0.1cm]{u63}
    \tikzset{VertexStyle/.style = {shape=circle, fill=black, minimum size=5pt, inner sep=1pt, draw}}

    \Vertex[x=\kpx,y=\akvy,L={E_*},Lpos=320,Ldist=-0.2cm]{e43}
    \Vertex[x=\kpx-0.7,y=\bkvy,L={E_*},Lpos=320,Ldist=-0.2cm]{e53}
    \Vertex[x=\kpx,y=\fkvy,L={E_*},Ldist=-0.1cm]{e63}

    \Edge[style={bend right}](u43)(e23)
    \Edge[style={bend left}](u43)(e23)
    \Edge(u63)(e23)


    \def\xsh{15}
    \node at (\xsh-1,0) {$=$};
    \node at (\xsh+5.7,0) {$\circ$};

    \draw[thin, color=gray] (\xsh+\sx, 0) .. controls (\xsh+\sx+1.2, \ary) .. (\xsh+\sx+3, \ary);
    \draw[thin, color=gray] (\xsh+\sx, 0) .. controls (\xsh+\sx+1.2, 0) .. (\xsh+\sx+3, 0);
    \draw[thin, color=gray] (\xsh+\sx, 0) .. controls (\xsh+\sx+1.2, -\ary) .. (\xsh+\sx+3, -\ary);

    \draw[thin, color=gray] (\xsh+\sx-1.5, \aly) .. controls (\xsh+\sx-0.7, \aly-0.1) .. (\xsh+\sx, 0);
    \draw[thin, color=gray] (\xsh+\sx-1.5, \cly) .. controls (\xsh+\sx-0.7, \cly+0.1) .. (\xsh+\sx, 0);

    \Vertex[x=\xsh+\sx,y=0,L={\e^{2,3}},Lpos=90,Ldist=0.1cm]{v22}

    \draw[thin, color=gray] (\xsh+\sx-1.5, \akvy) -- (\xsh+\sx+3, \akvy);
    \draw[thin, color=gray] (\xsh+\sx-1.5, \bkvy) -- (\xsh+\sx+3, \bkvy);
    \draw[thin, color=gray] (\xsh+\sx-1.5, \fkvy) -- (\xsh+\sx+3, \fkvy);
    \Vertex[x=\xsh+\sx,y=\akvy,L={\ii^{\otimes 2}},Lpos=90]{m12}
    \Vertex[x=\xsh+\sx,y=\bkvy, NoLabel]{m22}
    \Vertex[x=\xsh+\sx,y=\fkvy,L={\ii},Lpos=270]{m32}


    \filldraw[color=black!70, fill=blue!8] (\xsh+\kpx,0) circle (2.5);
    \draw[thin, color=gray] (\xsh+\kpx+1, \ckvy) .. controls (\xsh+\kpx-2, \ary) .. (\xsh+\kpx-3.5, \ary);
    \draw[thin, color=gray] (\xsh+\kpx+1, \ckvy) .. controls (\xsh+\kpx-1.5, 0) .. (\xsh+\kpx-3.5, 0);
    \draw[thin, color=gray] (\xsh+\kpx+1, \ekvy) .. controls (\xsh+\kpx-2, \dry) .. (\xsh+\kpx-3.5, \dry);

    \draw[thin, color=gray] (\xsh+\kpx-3.5, \akvy) -- (\xsh+\kpx, \akvy);
    \draw[thin, color=gray] (\xsh+\kpx-3.5, \bkvy) -- (\xsh+\kpx-0.7, \bkvy);
    \draw[thin, color=gray] (\xsh+\kpx-3.5, \fkvy) -- (\xsh+\kpx, \fkvy);

    \Vertex[x=\xsh+\kpx-1.2,y=\ckvy+0.4,L={E_2},Lpos=350,Ldist=-0.2cm]{e1}
    \Vertex[x=\xsh+\kpx-1.5,y=0.03,L={E_2},Lpos=355,Ldist=-0.2cm]{e2}
    \Vertex[x=\xsh+\kpx-1,y=\ekvy-0.06,L={E_2},Lpos=30,Ldist=-0.2cm]{e3}

    \tikzset{VertexStyle/.style = {shape=rectangle, fill=black, minimum size=5pt, inner sep=1pt, draw}}
    \Vertex[x=\xsh+\kpx+1,y=\ckvy,L={F_i},Ldist=-0.1cm]{u42}
    \Vertex[x=\xsh+\kpx+1,y=\ekvy,L={F_j},Ldist=-0.1cm]{u62}
    \tikzset{VertexStyle/.style = {shape=circle, fill=black, minimum size=5pt, inner sep=1pt, draw}}

    \Vertex[x=\xsh+\kpx,y=\akvy,L={E_*},Lpos=320,Ldist=-0.2cm]{e42}
    \Vertex[x=\xsh+\kpx-0.7,y=\bkvy,L={E_*},Lpos=320,Ldist=-0.2cm]{e52}
    \Vertex[x=\xsh+\kpx,y=\fkvy,L={E_*},Ldist=-0.1cm]{e62}

    \Edge(u42)(e1)
    \Edge(u42)(e2)
    \Edge(u62)(e3)
\end{tikzpicture}
        \caption{Extracting $\e^{m,d}$ as in \eqref{eq:newleft1}
        with $m = 2$, $d = 3$, $r=2$, $t=1$.}
        \label{fig:evss}
    \end{figure}
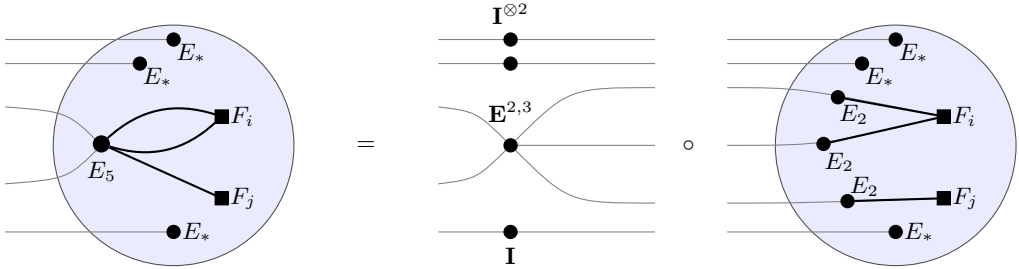

    The $E_2$'s left in $\k'$ are relevant to the other subcase, where $v$ is a constraint vertex,
    assigned an $(m+d)$-ary signature $F \in \fc$. In this case, $v$ is adjacent to equality
    vertices, so when we apply \eqref{eq:left}, the new dangling edges are correctly incident to 
    equalities. This solves the second problem in the first subcase. The first problem is also
    solved: recall that we ignored the $E_2$ vertices on constraint vertex $v$'s dangling edges 
    (these are the $E_2$ vertices left from the first subcase). We now remember their
    existence, and see that they are on the output dangling edges of the extracted $\s^{m,d}$, which
    turns the $\s^{m,d}$ into $\f^{m,d}$. Since $m+d = n$, \autoref{lem:rotategadget} says we can
    construct $\f^{m,d}$. See \autoref{fig:extractf}.
    
    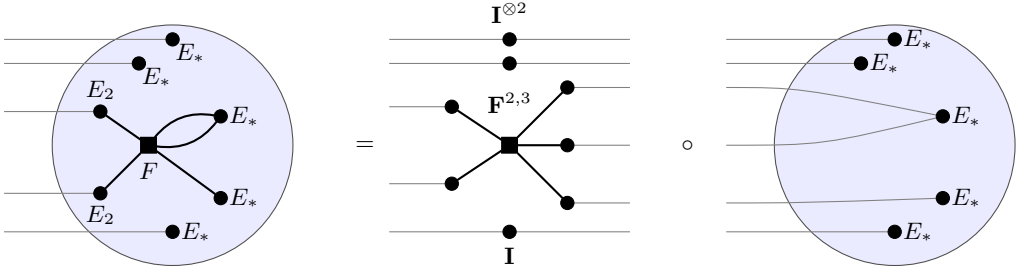
\begin{figure}[ht!]
        \center
        \begin{tikzpicture}[scale=.64]
    \tikzstyle{every node}=[font=\small]
    \GraphInit[vstyle=Classic]
    \SetUpEdge[style=-]
    \SetVertexMath

    \def\ary{1.2}
    \def\bry{0.4}
    \def\cry{-0.4}
    \def\dry{-1.2}

    \def\aly{0.8}
    \def\bly{0}
    \def\cly{-0.8}

    \def\sx{2}
    \def\kpx{10}

    \def\akvy{2.2}
    \def\bkvy{1.7}
    \def\ckvy{0.6}
    \def\ekvy{-1.1}
    \def\fkvy{-1.8}

    \filldraw[color=black!70, fill=blue!8] (\kpx,0) circle (2.5);

    \draw[thin, color=gray] (\kpx-3.5, \akvy) -- (\kpx, \akvy);
    \draw[thin, color=gray] (\kpx-3.5, \bkvy) -- (\kpx-0.7, \bkvy);
    \draw[thin, color=gray] (\kpx-3.5, \fkvy) -- (\kpx, \fkvy);

    \draw[thin, color=gray] (\kpx-3.5, 0.7) -- (\kpx-1.5, 0.7);
    \draw[thin, color=gray] (\kpx-3.5, -1) -- (\kpx-1.5, -1);

    \tikzset{VertexStyle/.style = {shape=rectangle, fill=black, minimum size=6pt, inner sep=1pt, draw}}
    \Vertex[x=\kpx-0.5,y=0,L={F},Lpos=270]{e23}
    \tikzset{VertexStyle/.style = {shape=circle, fill=black, minimum size=5pt, inner sep=1pt, draw}}

    \Vertex[x=\kpx-1.5,y=0.7,L={E_2},Lpos=90,Ldist=-0.05cm]{e21}
    \Vertex[x=\kpx-1.5,y=-1,L={E_2},Lpos=270,Ldist=-0.05cm]{e22}

    \Vertex[x=\kpx+1,y=\ckvy,L={E_*},Ldist=-0.1cm]{u43}
    \Vertex[x=\kpx+1,y=\ekvy,L={E_*},Ldist=-0.1cm]{u63}

    \Vertex[x=\kpx,y=\akvy,L={E_*},Lpos=320,Ldist=-0.2cm]{e43}
    \Vertex[x=\kpx-0.7,y=\bkvy,L={E_*},Lpos=320,Ldist=-0.2cm]{e53}
    \Vertex[x=\kpx,y=\fkvy,L={E_*},Ldist=-0.1cm]{e63}

    \Edge[style={bend right}](u43)(e23)
    \Edge[style={bend left}](u43)(e23)
    \Edge(u63)(e23)
    \Edge(e23)(e21)
    \Edge(e23)(e22)

    \def\xsh{15}
    \def\ysh{0}
    \node at (\xsh-1,\ysh) {$=$};
    \node at (\xsh+5.7,\ysh) {$\circ$};

    \draw[thin, color=gray] (\xsh+\sx+1.2, \ary+\ysh) -- (\xsh+\sx+2.5, \ary+\ysh);
    \draw[thin, color=gray] (\xsh+\sx+1.2, 0+\ysh) -- (\xsh+\sx+2.5, 0+\ysh);
    \draw[thin, color=gray] (\xsh+\sx+1.2, \dry+\ysh) -- (\xsh+\sx+2.5, \dry+\ysh);

    \draw[thin, color=gray] (\xsh+\sx-2.5, \aly+\ysh) -- (\xsh+\sx-1.2, \aly+\ysh);
    \draw[thin, color=gray] (\xsh+\sx-2.5, \cly+\ysh) -- (\xsh+\sx-1.2, \cly+\ysh);

    \tikzset{VertexStyle/.style = {shape=rectangle, fill=black, minimum size=6pt, inner sep=1pt, draw}}
    \Vertex[x=\xsh+\sx,y=0+\ysh,L={\f^{2,3}},Lpos=90,Ldist=0.2cm]{v2}
    \tikzset{VertexStyle/.style = {shape=circle, fill=black, minimum size=5pt, inner sep=1pt, draw}}

    \Vertex[x=\xsh+\sx+1.2,y=\ary+\ysh,NoLabel,Lpos=300,Ldist=-0.1cm]{f1}
    \Vertex[x=\xsh+\sx+1.2,y=0+\ysh,NoLabel,Lpos=300,Ldist=-0.1cm]{f2}
    \Vertex[x=\xsh+\sx+1.2,y=\dry+\ysh,NoLabel,Lpos=300,Ldist=-0.1cm]{f3}

    \Vertex[x=\xsh+\sx-1.2,y=\aly+\ysh,NoLabel,Lpos=300,Ldist=-0.1cm]{l1}
    \Vertex[x=\xsh+\sx-1.2,y=\cly+\ysh,NoLabel,Lpos=300,Ldist=-0.1cm]{l2}
    \Edge(v2)(f1)
    \Edge(v2)(f2)
    \Edge(v2)(f3)
    \Edge(v2)(l1)
    \Edge(v2)(l2)

    \draw[thin, color=gray] (\xsh+\sx-2.5, \akvy+\ysh) -- (\xsh+\sx+2.5, \akvy+\ysh);
    \draw[thin, color=gray] (\xsh+\sx-2.5, \bkvy+\ysh) -- (\xsh+\sx+2.5, \bkvy+\ysh);
    \draw[thin, color=gray] (\xsh+\sx-2.5, \fkvy+\ysh) -- (\xsh+\sx+2.5, \fkvy+\ysh);
    \Vertex[x=\xsh+\sx,y=\akvy+\ysh,L={\ii^{\otimes 2}},Lpos=90]{m1}
    \Vertex[x=\xsh+\sx,y=\bkvy+\ysh, NoLabel]{m2}
    \Vertex[x=\xsh+\sx,y=\fkvy+\ysh,L={\ii},Lpos=270]{m3}


    \filldraw[color=black!70, fill=blue!8] (\xsh+\kpx,0+\ysh) circle (2.5);
    \draw[thin, color=gray] (\xsh+\kpx+1, \ckvy+\ysh) .. controls (\xsh+\kpx-2, \ary+\ysh) .. (\xsh+\kpx-3.5, \ary+\ysh);
    \draw[thin, color=gray] (\xsh+\kpx+1, \ckvy+\ysh) .. controls (\xsh+\kpx-1.5, 0+\ysh) .. (\xsh+\kpx-3.5, 0+\ysh);
    \draw[thin, color=gray] (\xsh+\kpx+1, \ekvy+\ysh) .. controls (\xsh+\kpx-2, \dry+\ysh) .. (\xsh+\kpx-3.5, \dry+\ysh);

    \draw[thin, color=gray] (\xsh+\kpx-3.5, \akvy+\ysh) -- (\xsh+\kpx, \akvy+\ysh);
    \draw[thin, color=gray] (\xsh+\kpx-3.5, \bkvy+\ysh) -- (\xsh+\kpx-0.7, \bkvy+\ysh);
    \draw[thin, color=gray] (\xsh+\kpx-3.5, \fkvy+\ysh) -- (\xsh+\kpx, \fkvy+\ysh);

    \Vertex[x=\xsh+\kpx+1,y=\ckvy+\ysh,L={E_*},Ldist=-0.1cm]{u4}
    \Vertex[x=\xsh+\kpx+1,y=\ekvy+\ysh,L={E_*},Ldist=-0.1cm]{u6}

    \Vertex[x=\xsh+\kpx,y=\akvy+\ysh,L={E_*},Ldist=-0.1cm]{e4}
    \Vertex[x=\xsh+\kpx-0.7,y=\bkvy+\ysh,L={E_*},Ldist=-0.1cm]{e5}
    \Vertex[x=\xsh+\kpx,y=\fkvy+\ysh,L={E_*},Ldist=-0.1cm]{e6}

\end{tikzpicture}
        \caption{Extracting $\f^{m,d}$ as in \eqref{eq:newleft2}
        with $m = 2$, $d = 3$, $r=2$, $t=1$.}
        \label{fig:extractf}
    \end{figure}

    In summary, we replace \eqref{eq:left} with the following two decompositions, 
    where $v$ is a signature and constraint vertex, respectively.
    \begin{align*}
        &\k = (\ii^{\otimes r} \otimes \e^{m,d} \otimes \ii^{\otimes t}) \circ \k' 
        \numberthis \label{eq:newleft1} \\
        &\k = (\ii^{\otimes r} \otimes \f^{m,d} \otimes \ii^{\otimes t}) \circ \k' 
        \text{ for some $(m+d)$-ary $F \in \fc$}
        \numberthis \label{eq:newleft2}
    \end{align*}

    In case \eqref{eq:right}, all of $v$'s consecutive incident dangling edges are inputs. The
    horizontal reflection of all the reasoning in case \eqref{eq:left} applies.
    Rather than apply ${\dagger}$ to $\f^{m,d}$, we can always obtain $\f^{d,m}$ by
    \autoref{lem:rotategadget} instead, and of course $\ii$ is symmetric and we can use $\e^{d,m}$ 
    instead of $\e^{m,d}$. So we have
    \begin{align*}
        &\k = \k' \circ (\ii^{\otimes r} \otimes \e^{d,m} \otimes \ii^{\otimes t})
        \numberthis \label{eq:newright1} \\
        &\k = \k' \circ (\ii^{\otimes r} \otimes \f^{d,m} \otimes \ii^{\otimes t})
        \text{ for some $(m+d)$-ary $F \in \fc$}
        \numberthis \label{eq:newright2}
    \end{align*}
    as the reflections of \eqref{eq:newleft1} and \eqref{eq:newleft2}, respectively.

    In case \eqref{eq:top}, $v$'s consecutive incident dangling edges consist of the top $m$ output
    and top $r$ input danglng edges, and $v$ has $d$ edges to vertices in $\k'$.
    The extra top input dangling edge incident to 
    $\s_L^{m,d}$'s center vertex gets attached to the $\e^{1,r}$ to ensure $\k$ has $r$ input dangling edges.
    The reasoning is the same as \eqref{eq:left} for the $m$ input dangling edges and $d$ 
    edges into $\k'$. If $v$ is an equality vertex, we simply use $\e^{m,d+1}$ instead of
    $\s_L^{m,d}$ and use the extra input dangling edge for the same purpose as the extra dangling edge attached to $\s_L^{m,d}$'s center vertex.
    If $v$ is a constraint vertex assigned $(m+d+r)$-ary $F \in \fc$, we already have $E_2$ vertices on
    $v$'s top $r$ input arms (as well as the rest of $v$'s arms), so we can use
    $\f^{m, r+d}$ and attach the input
    dangling edges directly via $\ii^{\otimes r}$ instead of via $\e^{1,r}$. So we have
    \begin{align*}
        &\k = (\e^{m,d+1} \otimes \ii^{\otimes t}) \circ (\e^{1,r} \otimes \k')
        \numberthis \label{eq:newtop1} \\
        &\k = (\f^{m,r+d} \otimes \ii^{\otimes t}) \circ (\ii^{\otimes r} \otimes \k')
        \text{ for some $(m+d+r)$-ary $F \in \fc$}
        \numberthis \label{eq:newtop2}
    \end{align*}
    if $v$ is an equality or constraint vertex, respectively.
    See \autoref{fig:bothsides}.
    
    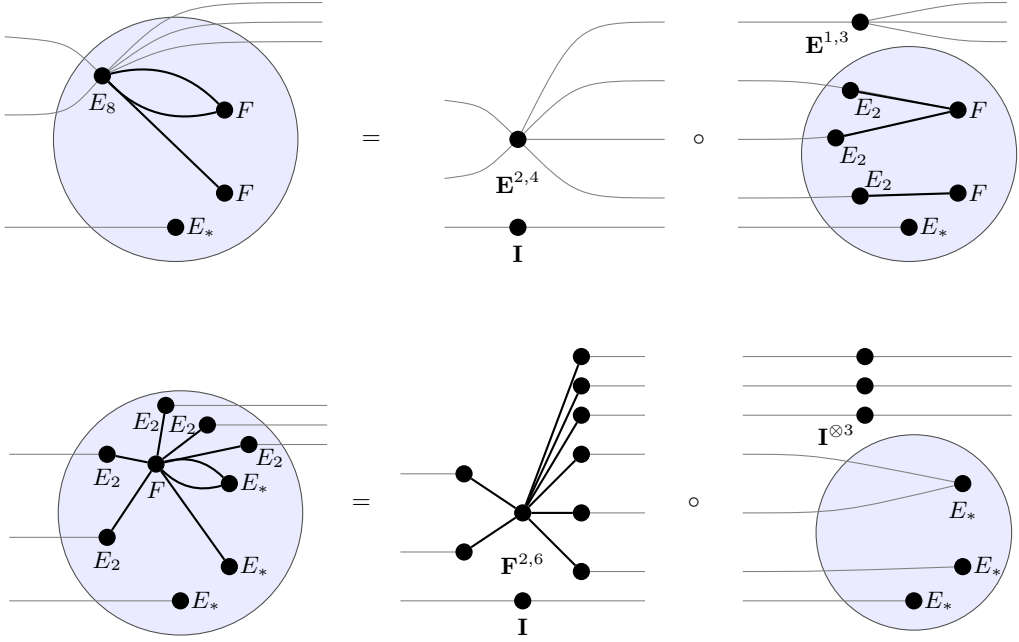
\begin{figure}[ht!]
        \center
        \begin{tikzpicture}[scale=.65]
    \tikzstyle{every node}=[font=\small]
    \GraphInit[vstyle=Classic]
    \SetUpEdge[style=-]
    \SetVertexMath

    \def\ary{1.2}
    \def\bry{0.4}
    \def\cry{-0.4}
    \def\dry{-1.2}

    \def\aly{0.8}
    \def\bly{0}
    \def\cly{-0.8}

    \def\sx{2}
    \def\kpx{10}

    \def\akvy{2.2}
    \def\bkvy{1.7}
    \def\ckvy{0.6}
    \def\ekvy{-1.1}
    \def\fkvy{-1.8}

    \def\xsh{15}

    \filldraw[color=black!70, fill=blue!8] (\kpx,0) circle (2.5);

    \draw[thin, color=gray] (\kpx-3.5, \fkvy) -- (\kpx, \fkvy);

    \draw[thin, color=gray] (\kpx-3.5, \aly+1.3) .. controls (\kpx-2.2, \aly+1.2) .. (\kpx-1.5, 1.3);
    \draw[thin, color=gray] (\kpx-3.5, \cly+1.3) .. controls (\kpx-2.2, \cly+1.3) .. (\kpx-1.5, 1.3);

    \draw[thin, color=gray] (\kpx-1.5, 1.3) .. controls (\kpx, 2) .. (\xsh-2, 2);
    \draw[thin, color=gray] (\kpx-1.5, 1.3) .. controls (\kpx, 2.4) .. (\xsh-2, 2.4);
    \draw[thin, color=gray] (\kpx-1.5, 1.3) .. controls (\kpx, 2.8) .. (\xsh-2, 2.8);

    \Vertex[x=\kpx-1.5,y=1.3,L={E_8},Lpos=270]{e23}

    \Vertex[x=\kpx+1,y=\ckvy,L=F,Ldist=-0.1cm]{u43}
    \Vertex[x=\kpx+1,y=\ekvy,L=F,Ldist=-0.1cm]{u63}

    \Vertex[x=\kpx,y=\fkvy,L={E_*},Ldist=-0.1cm]{e63}

    \Edge[style={bend right}](u43)(e23)
    \Edge[style={bend left}](u43)(e23)
    \Edge(u63)(e23)


    \node at (\xsh-1,0) {$=$};
    \node at (\xsh+5.7,0) {$\circ$};

    \draw[thin, color=gray] (\xsh+\sx, 0) .. controls (\xsh+\sx+1.2, \ary+1.2) .. (\xsh+\sx+3, \ary+1.2);
    \draw[thin, color=gray] (\xsh+\sx, 0) .. controls (\xsh+\sx+1.2, \ary) .. (\xsh+\sx+3, \ary);
    \draw[thin, color=gray] (\xsh+\sx, 0) .. controls (\xsh+\sx+1.2, 0) .. (\xsh+\sx+3, 0);
    \draw[thin, color=gray] (\xsh+\sx, 0) .. controls (\xsh+\sx+1.2, -\ary) .. (\xsh+\sx+3, -\ary);

    \draw[thin, color=gray] (\xsh+\sx-1.5, \aly) .. controls (\xsh+\sx-0.7, \aly-0.1) .. (\xsh+\sx, 0);
    \draw[thin, color=gray] (\xsh+\sx-1.5, \cly) .. controls (\xsh+\sx-0.7, \cly+0.1) .. (\xsh+\sx, 0);

    \Vertex[x=\xsh+\sx,y=0,L={\e^{2,4}},Lpos=270,Ldist=0.2cm]{v22}

    \draw[thin, color=gray] (\xsh+\sx-1.5, \fkvy) -- (\xsh+\sx+3, \fkvy);
    \Vertex[x=\xsh+\sx,y=\fkvy,L={\ii},Lpos=270]{m32}

    \filldraw[color=black!70, fill=blue!8] (\xsh+\kpx,-0.3) circle (2.2);
    \draw[thin, color=gray] (\xsh+\kpx+1, \ckvy) .. controls (\xsh+\kpx-2, \ary) .. (\xsh+\kpx-3.5, \ary);
    \draw[thin, color=gray] (\xsh+\kpx+1, \ckvy) .. controls (\xsh+\kpx-1.5, 0) .. (\xsh+\kpx-3.5, 0);
    \draw[thin, color=gray] (\xsh+\kpx+1, \ekvy) .. controls (\xsh+\kpx-2, \dry) .. (\xsh+\kpx-3.5, \dry);

    \draw[thin, color=gray] (\xsh+\kpx-3.5, \fkvy) -- (\xsh+\kpx, \fkvy);

    \draw[thin, color=gray] (\xsh+\kpx-3.5, \ary+1.2) -- (\xsh+\kpx-1, \ary+1.2);
    \draw[thin, color=gray]  (\xsh+\kpx-1, \ary+1.2) .. controls (\xsh+\kpx+1, \ary+1.6) .. (\xsh+\kpx+2, \ary+1.6);
    \draw[thin, color=gray]  (\xsh+\kpx-1, \ary+1.2) .. controls (\xsh+\kpx+1, \ary+1.2) .. (\xsh+\kpx+2, \ary+1.2);
    \draw[thin, color=gray]  (\xsh+\kpx-1, \ary+1.2) .. controls (\xsh+\kpx+1, \ary+0.8) .. (\xsh+\kpx+2, \ary+0.8);
    \Vertex[x=\xsh+\kpx-1,y=\ary+1.2,L={\e^{1,3}},Lpos=230,Ldist=-0.1cm]{ex}

    \Vertex[x=\xsh+\kpx-1.2,y=\ckvy+0.4,L={E_2},Lpos=350,Ldist=-0.2cm]{e1}
    \Vertex[x=\xsh+\kpx-1.5,y=0.03,L={E_2},Lpos=355,Ldist=-0.2cm]{e2}
    \Vertex[x=\xsh+\kpx-1,y=\ekvy-0.06,L={E_2},Lpos=30,Ldist=-0.2cm]{e3}

    \Vertex[x=\xsh+\kpx+1,y=\ckvy,L=F,Ldist=-0.1cm]{u42}
    \Vertex[x=\xsh+\kpx+1,y=\ekvy,L=F,Ldist=-0.1cm]{u62}

    \Vertex[x=\xsh+\kpx,y=\fkvy,L={E_*},Ldist=-0.1cm]{e62}

    \Edge(u42)(e1)
    \Edge(u42)(e2)
    \Edge(u62)(e3)
\end{tikzpicture}
        
        \vspace{1cm}
        
        \begin{tikzpicture}[scale=.65]
    \tikzstyle{every node}=[font=\small]
    \GraphInit[vstyle=Classic]
    \SetUpEdge[style=-]
    \SetVertexMath

    \def\ary{1.2}
    \def\bry{0.4}
    \def\cry{-0.4}
    \def\dry{-1.2}

    \def\aly{0.8}
    \def\bly{0}
    \def\cly{-0.8}

    \def\sx{2}
    \def\kpx{10}

    \def\akvy{2.2}
    \def\bkvy{1.7}
    \def\ckvy{0.6}
    \def\ekvy{-1.1}
    \def\fkvy{-1.8}

    \def\xsh{15}
    \def\ysh{0}
    \node at (\xsh-1.3,\ysh+0.2) {$=$};
    \node at (\xsh+5.5,\ysh+0.2) {$\circ$};

    \filldraw[color=black!70, fill=blue!8] (\kpx,0) circle (2.5);

    \draw[thin, color=gray] (\kpx-3.5, \fkvy) -- (\kpx, \fkvy);

    \draw[thin, color=gray] (\kpx-3.5, 1.2) -- (\kpx-1.5, 1.2);
    \draw[thin, color=gray] (\kpx-3.5, -0.5) -- (\kpx-1.5, -0.5);

    \draw[thin, color=gray] (\kpx-0.3, 2.2) -- (\xsh-2, 2.2);
    \draw[thin, color=gray] (\kpx+0.55, 1.8) -- (\xsh-2, 1.8);
    \draw[thin, color=gray] (\kpx+1.4, 1.4) -- (\xsh-2, 1.4);

    \Vertex[x=\kpx-0.5,y=1,L={F},Lpos=270]{e23}

    \Vertex[x=\kpx-1.5,y=1.2,L={E_2},Lpos=270,Ldist=-0.05cm]{e21}
    \Vertex[x=\kpx-1.5,y=-0.5,L={E_2},Lpos=270,Ldist=-0.05cm]{e22}

    \Vertex[x=\kpx-0.3,y=2.2,L={E_2},Lpos=210,Ldist=-0.18cm]{e24}
    \Vertex[x=\kpx+0.55,y=1.8,L={E_2},Lpos=180,Ldist=-0.1cm]{e25}
    \Vertex[x=\kpx+1.4,y=1.4,L={E_2},Lpos=300,Ldist=-0.2cm]{e26}

    \Vertex[x=\kpx+1,y=\ckvy,L={E_*},Ldist=-0.1cm]{u43}
    \Vertex[x=\kpx+1,y=\ekvy,L={E_*},Ldist=-0.1cm]{u63}

    \Vertex[x=\kpx,y=\fkvy,L={E_*},Ldist=-0.1cm]{e63}

    \Edge[style={bend right}](u43)(e23)
    \Edge[style={bend left}](u43)(e23)
    \Edge(u63)(e23)
    \Edge(e23)(e21)
    \Edge(e23)(e22)
    \Edge(e23)(e24)
    \Edge(e23)(e25)
    \Edge(e23)(e26)

    \draw[thin, color=gray] (\xsh+\sx+1.2, \ary+\ysh) -- (\xsh+\sx+2.5, \ary+\ysh);
    \draw[thin, color=gray] (\xsh+\sx+1.2, 0+\ysh) -- (\xsh+\sx+2.5, 0+\ysh);
    \draw[thin, color=gray] (\xsh+\sx+1.2, \dry+\ysh) -- (\xsh+\sx+2.5, \dry+\ysh);

    \draw[thin, color=gray] (\xsh+\sx-2.5, \aly+\ysh) -- (\xsh+\sx-1.2, \aly+\ysh);
    \draw[thin, color=gray] (\xsh+\sx-2.5, \cly+\ysh) -- (\xsh+\sx-1.2, \cly+\ysh);

    \draw[thin, color=gray] (\xsh+\sx+1.2, \ary+2) -- (\xsh+\sx+2.5, \ary+2);
    \draw[thin, color=gray] (\xsh+\sx+1.2, \ary+1.4) -- (\xsh+\sx+2.5, \ary+1.4);
    \draw[thin, color=gray] (\xsh+\sx+1.2, \ary+0.8) -- (\xsh+\sx+2.5, \ary+0.8);

    \Vertex[x=\xsh+\sx,y=0+\ysh,L={\f^{2,6}},Lpos=270,Ldist=0.3cm]{v2}
    \Vertex[x=\xsh+\sx+1.2,y=\ary+\ysh,NoLabel,Lpos=300,Ldist=-0.1cm]{f1}
    \Vertex[x=\xsh+\sx+1.2,y=0+\ysh,NoLabel,Lpos=300,Ldist=-0.1cm]{f2}
    \Vertex[x=\xsh+\sx+1.2,y=\dry+\ysh,NoLabel,Lpos=300,Ldist=-0.1cm]{f3}

    \Vertex[x=\xsh+\sx-1.2,y=\aly+\ysh,NoLabel,Lpos=300,Ldist=-0.1cm]{l1}
    \Vertex[x=\xsh+\sx-1.2,y=\cly+\ysh,NoLabel,Lpos=300,Ldist=-0.1cm]{l2}

    \Vertex[x=\xsh+\sx+1.2,y=\ary+2,NoLabel]{f4}
    \Vertex[x=\xsh+\sx+1.2,y=\ary+1.4,NoLabel]{f5}
    \Vertex[x=\xsh+\sx+1.2,y=\ary+0.8,NoLabel]{f6}

    \Edge(v2)(f1)
    \Edge(v2)(f2)
    \Edge(v2)(f3)
    \Edge(v2)(l1)
    \Edge(v2)(l2)
    \Edge(v2)(f4)
    \Edge(v2)(f5)
    \Edge(v2)(f6)

    \draw[thin, color=gray] (\xsh+\sx-2.5, \fkvy+\ysh) -- (\xsh+\sx+2.5, \fkvy+\ysh);
    \Vertex[x=\xsh+\sx,y=\fkvy+\ysh,L={\ii},Lpos=270]{m3}

    \draw[thin, color=gray] (\xsh+\kpx-3.5, \ary+2) -- (\xsh+\kpx+2, \ary+2);
    \draw[thin, color=gray] (\xsh+\kpx-3.5, \ary+1.4) -- (\xsh+\kpx+2, \ary+1.4);
    \draw[thin, color=gray] (\xsh+\kpx-3.5, \ary+0.8) -- (\xsh+\kpx+2, \ary+0.8);
    \Vertex[x=\xsh+\kpx-1,y=\ary+2,NoLabel]{ex1}
    \Vertex[x=\xsh+\kpx-1,y=\ary+1.4,NoLabel]{ex2}
    \Vertex[x=\xsh+\kpx-1,y=\ary+0.8,L={\ii^{\otimes 3}},Lpos=210,Ldist=-0.1cm]{ex3}

    \filldraw[color=black!70, fill=blue!8] (\xsh+\kpx,\ysh-0.4) circle (2);
    \draw[thin, color=gray] (\xsh+\kpx+1, \ckvy+\ysh) .. controls (\xsh+\kpx-2, \ary+\ysh) .. (\xsh+\kpx-3.5, \ary+\ysh);
    \draw[thin, color=gray] (\xsh+\kpx+1, \ckvy+\ysh) .. controls (\xsh+\kpx-1.5, 0+\ysh) .. (\xsh+\kpx-3.5, 0+\ysh);
    \draw[thin, color=gray] (\xsh+\kpx+1, \ekvy+\ysh) .. controls (\xsh+\kpx-2, \dry+\ysh) .. (\xsh+\kpx-3.5, \dry+\ysh);

    \draw[thin, color=gray] (\xsh+\kpx-3.5, \fkvy+\ysh) -- (\xsh+\kpx, \fkvy+\ysh);

    \Vertex[x=\xsh+\kpx+1,y=\ckvy+\ysh,L={E_*},Lpos=270]{u4}
    \Vertex[x=\xsh+\kpx+1,y=\ekvy+\ysh,L={E_*},Ldist=-0.1cm]{u6}

    \Vertex[x=\xsh+\kpx,y=\fkvy+\ysh,L={E_*},Ldist=-0.1cm]{e6}

\end{tikzpicture}
        \caption{In illustration of \eqref{eq:newtop1} (top) and \eqref{eq:newtop2} (bottom) for $m = 2, d=3, r = 3, t=1$.}
        \label{fig:bothsides}
    \end{figure}

    Case \eqref{eq:bottom} is the vertical reflection of \eqref{eq:top} ($v$'s dangling edges consist of
    the bottom output and input dangling edges), so similar reasoning applies.
    We have
    \begin{align*}
        &\k = (\ii^{\otimes t} \otimes \e^{m,d+1}) \circ (\k' \otimes \e^{1,r})
        \numberthis \label{eq:newbottom1} \\
        &\k = (\ii^{\otimes t }\otimes \f^{m,r+d}) \circ (\k' \otimes \ii^{\otimes r})
        \text{ for some $(m+d+r)$-ary $F \in \fc$}
        \numberthis \label{eq:newbottom2}
    \end{align*}
    if $v$ is an equality or constraint vertex, respectively.

    Together, \eqref{eq:newleft1}-\eqref{eq:newbottom2} cover all cases with a $v$ with a nonzero 
    quantity of consecutive incident dangling edges. 
    Case \eqref{eq:nodangling} occurs when  $\k$ has no dangling edges. In this case, choose an arbitrary 
    constraint vertex $v$ on the outer face (such a vertex must exist unless $\k = \e^{m,d}$, in which case
    we're already done) assigned $n$-ary $F \in \fc$ and construct $\k'$ from $\k$ by deleting $v$ and replacing the edges between
    $v$ and its neighbors with output dangling edges attached to the appropriate neighbors. Then
    we have 
    \begin{equation}
        \label{eq:newnodangling}
        \k = \f^{0,n} \circ \k'
        \text{ for some $n$-ary $F \in \fc$}.
    \end{equation}

    The $\k'$ in \eqref{eq:newleft2}, \eqref{eq:newright2}, \eqref{eq:newtop2}, \eqref{eq:newbottom2},
    and \eqref{eq:newnodangling}
    (the constraint vertex cases) is exactly the same as the $\k'$ in \eqref{eq:left}, \eqref{eq:right}, \eqref{eq:top},
    \eqref{eq:bottom}, and \eqref{eq:nodangling}, respectively
    (\eqref{eq:newtop2} and \eqref{eq:newbottom2} have structural
    differences, but they involve the new input dangling edges and have nothing to do with $\k'$).
    Hence we can apply the proofs in \cite{planar} that $\k'$ is still planar. Since
    we extracted a constraint vertex, which is only adjacent to equality vertices,
    $\k'$ has all dangling edges attached to equality vertices, and clearly is still bipartite.
    Thus $\k' \in \pn$.

    The $\k'$ in \eqref{eq:newleft1}, \eqref{eq:newright1}, \eqref{eq:newtop1}, and \eqref{eq:newbottom1},
    (the equality vertex cases) is the same as the $\k'$ in \eqref{eq:left}, \eqref{eq:right}, \eqref{eq:top}, and
    \eqref{eq:bottom}, respectively, up to the addition of $E_2$ vertices on new dangling edges.
    Edge subdivision preserves planarity, however, so the planarity proofs in \cite{planar} again
    still apply. By construction, all new edges dangling from $\k'$ are incident to equality vertices, and
    $\k'$ is still bipartite, so $\k' \in \pn$.

    In both cases, by \autoref{lem:consecutive}, $\widehat{\k'}$ contains a vertex $v$ whose dangling edges
    are consecutive, so set $\k$ to be $\k'$ and repeat.
    At each step, we always make progress towards decreasing the number of constraint vertices in $\k$
    by one.
    In cases \eqref{eq:newleft2}, \eqref{eq:newright2}, \eqref{eq:newtop2}, and \eqref{eq:newbottom2},
    and \eqref{eq:newnodangling}, $\k'$ has one fewer constraint vertex
    than $\k$. 
    In cases \eqref{eq:newleft1}, \eqref{eq:newright1}, \eqref{eq:newtop1}, and \eqref{eq:newbottom1},
    $\k'$ has one fewer vertex assigned $E_k$, $k > 2$ than $\k$ (we replaced this vertex with some
    number of vertices assigned $E_2$). If we keep extracting equality vertices, then eventually
    $\k$ will have no more vertices assigned $E_k$ for $k > 2$. Since we consider $\widehat{\k}$, ignoring vertices assigned $E_2$ with attached dangling edges, when choosing the next vertex $v$ to extract,
    $v$ must be a constraint vertex.

    Therefore by induction on the number of constraint vertices and equality vertices with arity $>2$,
    we assume $\k' \in \tcwdn{\e^{1,0}, \e^{1,2}, \{\f^{n_F,0} \mid F \in \fc\}}$, giving $\k \in \tcwdn{\e^{1,0}, \e^{1,2}, \{\f^{n_F,0} \mid F \in \fc\}}$.
    If $\k$ has only one such vertex, it is either $\e^{m,d}$ or some $\f^{m,d}$ by bipartiteness ($E_2$
    vertices can't appear in the first case, and can only appear adjacent to the constraint vertex, hence as part of $\f^{m,d}$, in the
    second case) and we have seen that these are both in
    $\tcwdn{\e^{1,0}, \e^{1,2}, \{\f^{n_F,0} \mid F \in \fc\}}$. Thus $\pn \subseteq \tcwdn{\e^{1,0}, \e^{1,2}, \{\f^{n_F,0} \mid F \in \fc\}}$.
\end{proof}

\section{The Quantum Holant Theorem}
\subsection{Gadgets and quantum permutation matrices}
\label{sec:invariance}
The quantum Holant theorem, proved in \autoref{sec:holant}, 
stems from viewing the quantum permutation matrix $\u$ itself as a signature
in a Holant signature grid, indicated by a triangle vertex $\blacktriangle$. 
An immediate corollary of this theorem is one half of our main result \autoref{thm:result}:
Planar \#CSP instances with quantum isomorphic signature sets have the same partition function value.
Due to the asymmetry of our constraint functions, the quantum Holant theorem depends on the
following three lemmas, which show that quantum isomorphism is invariant under constraint function
flattening, rotation, and reflection, respectively. See \autoref{fig:invariance}. For $F = G$,
the three lemmas follow directly from \autoref{lem:rotategadget} and results in 
\autoref{sec:quantumpermutationgroups}. See \autoref{rem:frobenius}.

\begin{figure}[ht!]
    \center
    \input{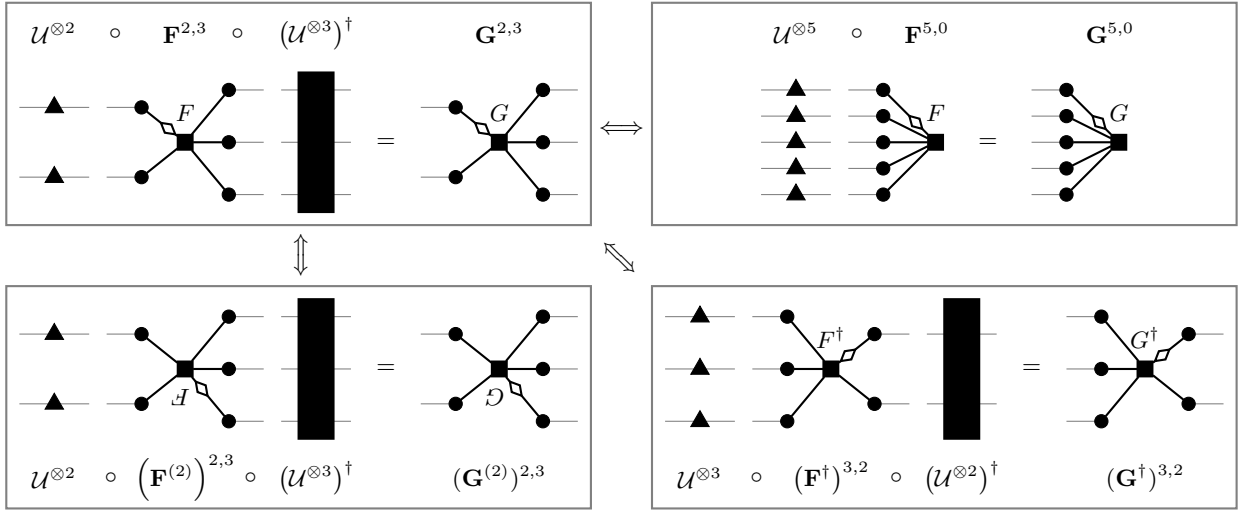}
    \caption{
    Illustrations of, clockwise from top right, Lemmas \ref{lem:tensorcftog}, 
    \ref{lem:conjugatecommute}, and \ref{lem:rotateinput}.
        $\left(\u^{\otimes n}\right)^{\dagger}$
        is drawn as a black box, since, due to noncommutativity, $\left(\u^{\otimes n}\right)^{\dagger}
        \neq \left(\u^{\dagger}\right)^{\otimes n}$ in general.}
    \label{fig:invariance}
\end{figure}

\begin{lemma}
    \label{lem:tensorcftog}
    Let $F,G \in \c^{[q]^n}$ and $\u$ be a $q \times q$ quantum permutation matrix.
    Then for all $m_1,d_1,m_2,d_2 \geq 0$ such that $m_1+d_1=m_2+d_2 = n$, we have
    \[
        \u^{\otimes m_1} F^{m_1,d_1} = G^{m_1,d_1}\u^{\otimes d_1}
        \iff
        \u^{\otimes m_2} F^{m_2,d_2} = G^{m_2,d_2}\u^{\otimes d_2}
    \]
    where we let $\u^{\otimes 0} = \one$.
\end{lemma}
\begin{proof}
    By \eqref{eq:tensorinverse}, we have 
    \[
        \u^{\otimes m} F^{m,d} = G^{m,d}\u^{\otimes d} \iff
        \u^{\otimes m} F^{m,d} (\u^{\otimes d})^{{\dagger}} = G^{m,d}.
    \]
    We will show $\u^{\otimes m} F^{m,d} (\u^{\otimes d})^{{\dagger}} = G^{m,d}$ is equivalent to
    $\u^{\otimes m+d}f = g$ (the $m=n$, $d=0$ case) for every
    $m,d \geq 0$, $m+d=n$.
Recall that $F^{m,d}(k_1, \ldots, k_m, \ell_1, \ldots, \ell_d) = f_{k_1\ldots k_m \, \ell_d \ldots \ell_1}$ (notice the reversal of second half of the
indices by definition), and note that 
\[( (\u^{\otimes d})^{\dagger} )_{\ell_1 \ldots \ell_d, j_1  \ldots j_d}
= ( (\u^{\otimes d})_{j_1  \ldots j_d, \ell_1 \ldots \ell_d} )^{\dagger}
=( u_{j_1, \ell_1}  \cdots u_{j_d,\ell_d} )^* 
= u_{j_d,\ell_d} \cdots u_{j_1, \ell_1}.
\]
For $md > 0$ we have
    \begin{align*}
        (\u^{\otimes m} F^{m,d} (\u^{\otimes d})^{{\dagger}})_{i_1^m,j_1^d}
        &= \sum_{k_1^m} \sum_{\ell_1^d} (\u^{\otimes m})_{i_1^m, k_1^m}
        F^{m,d}(k_1^m, \ell_1^d)
        ((\u^{\otimes d})^{\dagger} )_{\ell_1^d, j_1^d}\\
        &= \sum_{k_1^m} \sum_{\ell_1^d} u_{i_1,k_1}  \ldots u_{i_m,k_m} 
        u_{j_d,\ell_d}  \ldots u_{j_1,\ell_1}  f_{k_1^m, \ell_d^1} \\
       &= (\u^{\otimes m+d} f)_{i_1^mj_d^1},
    \end{align*}
    matching $G^{m,d}_{i_1^m,j_1^d} = g_{i_1^m j_d^1}$.
    If $m=0$ and $d=n$ then $F^{m,d}$ is a length-$q^d$ row vector and we have
    \begin{align*}
        (F^{0,d} (\u^{\otimes d})^{{\dagger}})_{j_1^d}
        &= \sum_{\ell_1, \ldots,  \ell_d} F^{0,d}(\mbox{-},\ell_1, \ldots, \ell_d) u_{j_d,\ell_d}  \cdots u_{j_1,\ell_1} \\
        &= \sum_{\ell_1^d} u_{j_d,\ell_d} \cdots  u_{j_1,\ell_1}  f_{\ell_d^1}  \\
        &= (\u^{\otimes d} f)_{j_d^1},
    \end{align*}
matching $G^{0,d}_{\mbox{-},j_1^d} = g_{j_d^1} = g_{j_d \ldots j_1}$.
\end{proof}

\begin{lemma}
    \label{lem:rotateinput}
    Let $F, G \in \c^{[q]^n}$ and $\u$ be a $q \times q$ quantum permutation matrix.
    For any $r \in [n-1]$ and $m,d \geq 0$, $m+d= n$,
    \[
        \u^{\otimes m} \left(F^{(r)}\right)^{m,d} = \left(G^{(r)}\right)^{m,d}
        \u^{\otimes d}
        \iff 
        \u^{\otimes m} F^{m,d} = G^{m,d}
        \u^{\otimes d}.
    \]
\end{lemma}
\begin{proof}
Iterating $r$ times the transformation from $F$ to $F^{(1)}$ we obtain $F^{(r)}$.
So it suffices to prove the case for $r=1$.
    By \autoref{lem:tensorcftog}, it is sufficient to show
    $\u^{\otimes n} f^{(1)} = g^{(1)} \iff \u^{\otimes n} f = g$, where $f^{(1)} = \left(F^{(1)}\right)^{n,0}$
    and $g^{(1)} = \left(G^{(1)}\right)^{n,0}$. Assume $\u^{\otimes n} f = g$,
    so for all $\vx \in [q]^n$, $\sum_{\vy} u_{x_1y_1}
    \ldots u_{x_{n-1}y_{n-1}} u_{x_ny_n}f_{\vy} = g_{\vx}$.
    For fixed $a$ and $x_1^{n-1}$, multiply both sides on the right
    by $u_{x_n a}$ and sum the resulting equation over all values
    of $x_n$. Since $\sum_{x} u_{xy_n} u_{xa} = \delta_{y_na}$, we get
    \begin{equation}
        \sum_{y_1^{n-1}} u_{x_1y_1} \ldots u_{x_{n-1}y_{n-1}}
        f_{y_1^{n-1}a}  
        = \sum_{x_n} u_{x_na} g_{\vx}.
        \label{eq:rotproof1}
    \end{equation}
    Now, for any  $b$,  multiply both sides of \eqref{eq:rotproof1} by $u_{ba}$ 
    on the left, and sum over all values of
    $a$. Since $\sum_{a} u_{ba} u_{xa} = \delta_{bx}$, we get
    \[
        (\u^{\otimes n} f^{(1)})_{b x_1^{n-1}} =
        \sum_{a,y_1^{n-1}} u_{ba} u_{x_1y_1} \ldots u_{x_{n-1}y_{n-1}} f_{y_1^{n-1}a}
        = g_{x_1^{n-1}b} = g^{(1)}_{b x_1^{n-1}}.
    \]
    This holds for all $b$ and $x_1^{n-1}$, so we have
    $\u^{\otimes n} f^{(1)} = g^{(1)}$. Now the $(\Rightarrow)$
    direction  (for $r=1$ case)  
    follows by repeating the above $n-1$ times. 
\end{proof}
We would like an analogue of \autoref{lem:rotateinput} for the reflection $F^{\top}$. 
However,  due to $\u$'s noncommutativity (in particular, the inverse
of $\u^{\otimes n}$ is $(\u^{\otimes n})^{\dagger}$, not $(\u^{\otimes n})^{\top}$, for $n > 1$), reflection introduces a conjugation.
\begin{lemma}
    \label{lem:conjugatecommute}
    Let $F, G \in \c^{[q]^n}$ and $\u$ be a $q \times q$ quantum permutation matrix.
    For any $m,d \geq 0$, $m+d = n$,
    \[
        \u^{\otimes m} (F^{\dagger})^{m,d}
        = (G^{\dagger})^{m,d} \u^{\otimes d}
        \iff
        \u^{\otimes m} F^{m,d}
        = G^{m,d} \u^{\otimes d}.
    \]
\end{lemma}
\begin{proof}
    By \eqref{eq:transposeequiv}, we have
    $(F^{\dagger})^{m,d} = (F^{d,m})^{\dagger}$ and $(G^{\dagger})^{m,d} = (G^{d,m})^{\dagger}$. $F$ and $G$ are scalar-valued, so their
    entries commute with $\u$'s. Hence taking the conjugate transpose of both sides of the first equation gives 
    \[
        F^{d,m}(\u^{\otimes m})^{\dagger} = (\u^{\otimes d})^{\dagger} G^{d,m}.
    \] 
    The result follows from multiplying on the right by $\u^{\otimes m}$ and on
    the left by $\u^{\otimes d}$.
\end{proof}

\subsection{Quantum holographic transformations}
\label{sec:holanttheorem}
We now prove the first half of our main result: If $\fc \cong_{qc} \gc$, then $Z(K) = Z(K_{\fc\to\gc})$ for
every planar \#CSP$(\fc)$ instance $K$. The proof will be by a technique inspired by 
\emph{holographic transformation}.
Holographic transformations, introduced by Valiant in \cite{valiant}, transform a Holant
signature grid into another grid on the same underlying graph with different signatures, resulting in the same Holant value.
Fix domain size $q$. For a set $\mathcal{F}$ of signatures and an invertible 
$T \in \c^{q \times q}$, write $T\mathcal{F} = \{T^{\otimes k}f \mid F \in \mathcal{F} \text{ has arity $k$}\}$.
Define $\mathcal{F}T$ similarly.
Valiant's Holant Theorem in \cite{valiant} states that, for any signature grid $\Omega$ and
sets of signatures $\mathcal{F}, \mathcal{G}$,
\[
    \holant_{\Omega}(\mathcal{F} \mid \mathcal{G}) = 
    \holant_{\Omega'}(T\mathcal{F} \mid \mathcal{G}T^{-1}),
\]
where $\Omega'$ is constructed from $\Omega$ by replacing every signature in 
$\mathcal{F}$ or $\mathcal{G}$ with
the corresponding transformed signature in $T\mathcal{F}$ or $\mathcal{G}T^{-1}$, respectively.

The signature grids $\Omega_K$ and $\Omega_{K_{\fc\to\gc}}$ satisfying 
$$Z(K) = \plholant_{\Omega_K}(\fc \mid \eq) \mbox{~~~~and~~~~}
Z(K_{\fc\to\gc}) = \plholant_{\Omega_{K_{\fc\to\gc}}}(\gc \mid \eq)$$ 
are the same, up to every signature $F \in \fc$ being replaced by the corresponding signature $G \in \gc$.
Assuming $\fc \cong_{qc} \gc$, there is a quantum permutation matrix $\u$ satisfying $\u^{\otimes n}f = g$ for every $F \in \fc$ and corresponding $G \in \gc$. This suggests
that we perform a \emph{quantum} holographic transformation using
$\u$. The following calculation supports
this idea:
for any $q \times q$ quantum permutation matrix $\u$, and any $k \geq 1$, $\vx \in [q]^k$,
\[
    (\u^{\otimes k} E^{k,0})_{\vx}
    = \sum_{\vy} u_{x_1y_1} \ldots u_{x_ky_k} E^{k,0}_{\vy}
    = \sum_y u_{x_1y} \ldots u_{x_ky}
    = \begin{cases}
        \sum_y u_{x_1y} = \one & \mbox{if } x_1 = \ldots = x_k \\
        ~~~~~~{\bf 0} & \text{otherwise}
    \end{cases}\]
which is 
    $E^{k,0}_{\vx}$ (or more technically $E^{k,0}_{\vx} \one$).
Thus by \autoref{lem:tensorcftog},
\begin{equation}
    \label{eq:ue}
    \u^{\otimes m} E^{m,d} = E^{m,d} \u^{\otimes d}
    \text{ for every $m,d \geq 0$.}
\end{equation}
In particular, since $\u^{-1} = \u^{\dagger}$ is also a quantum permutation matrix, $(\eq)\u^{-1} = \eq$. Now the Holant
theorem with $T$ set to $\u$ seems to give
\begin{equation}
    \label{eq:incorrectholant}
    Z(K) = \holant_{\Omega_K}(\fc \mid \eq) = \holant_{\Omega_{K_{\fc\to\gc}}}(\gc \mid \eq) = Z(K_{\fc\to\gc})
\end{equation}
for any, not necessarily planar, \#CSP$(\fc)$ instance $K$. However, this cannot be true. If $\fc = \{F\}$ and $\gc = \{G\}$, where 
$F$ and $G$ are symmetric, binary, 0-1 valued, \eqref{eq:incorrectholant} implies that the graphs with adjacency matrices $F$ and $G$
admit the same number of homomorphisms from any graph, giving $F \cong G$, a classical result of Lov\'asz~\cite{lovasz_operations}. In other words, any quantum isomorphic graphs are classically isomorphic.
But this is known to be false -- see e.g. \cite{asterias}.
This discrepancy is due to the proof of the Holant theorem breaking down when the  
entries of the holographic transformation matrix $T$ are noncommutative.
The proof of the Holant theorem goes roughly as follows (see \cite{cai_chen_2017}): 
Since $TT^{-1} = I$, we can subdivide
each edge of $\Omega$ into three edges by introducing two arity-2 vertices, assigned $T$ and $T^{-1}$,
without changing the Holant value. Since $\Omega$ is bipartite, partitioned according to
whether a vertex is
assigned a signature from $\mathcal{F}$ or $\mathcal{G}$, we can subdivide every edge in such a way
that every $u$ assigned $F \in \mathcal{F}$ or $v$ assigned $G \in \mathcal{G}$ is now adjacent to only vertices 
assigned $T$ or $T^{-1}$, respectively. Then we can associate $u$ and its surrounding $T$ vertices
into a single `gadget' with signature $T^{\otimes \text{deg}(u)}f$, and similarly associate $v$ and
its surrounding $T^{-1}$ vertices into $g^{T}(T^{-1})^{\otimes \text{deg}(v)}$, giving the result.
Back to  \eqref{eq:incorrectholant}, even though the signatures and the ultimate Holant values  are scalars
(technically scalar multiples of $\one$, such values  commute multiplicatively), the second equality in
\eqref{eq:incorrectholant} does not hold because it has a hidden dependence on the Holant value of
the signature grid in the proof, 
expressed as sum of products in the $C^*$-algebra,
in which $\u$ and $\u^{\dagger}$ with noncommutative entries appear as
signatures. Without further specification, the Holant value for such a signature grid is not even well-defined: the product
$\prod_{v \in V} F_v (\sigma \mid_{E(v)})$ 
in (\ref{eqn:def-holant}) does not specify an ordering of $V$, but when $F_v$ can
take noncommutative values, different orderings give different products.

When $\Omega_K$ is planar, however, the decomposition procedure for $\Omega_K$ in the proof of 
\autoref{thm:generatepn}
produces a sequence of gadgets whose signature matrices multiply to the Holant value, in a sense
defining an ordering of $\Omega_K$'s vertices. We will use $\u$ as a `quantum holographic transformation'
by inserting $\u^{\otimes k}$ and its inverse $(\u^{\otimes k})^{\dagger}$ between every pair of these
gadgets, converting every $F \in \fc$ to the corresponding $G \in \gc$ and preserving $\eq$.

\begin{theorem}[Quantum Holant Theorem]
    \label{thm:quantumholant}
    Let $\u$ be a $q \times q$ quantum permutation matrix, and
    let ${\cal F}$ and $\u\fc$ be compatible CC sets of domain-$q$ complex-valued signatures. Then for every
    $\plholant({\cal F})$ signature grid $\Omega$,
    \[
        \plholant_{\Omega}({\cal F}) = 
        \plholant_{\Omega'}(\u{\cal F}),
    \]
    where $\Omega'$ is constructed from $\Omega$ by replacing
    every signature in ${\cal F}$ with the corresponding
    signature in $\u{\cal F}$.
\end{theorem}
\begin{proof}
    Since Holant is multiplicative over connected components of a signature grid
    (and scalar multiples of $\one$ commute multiplicatively), we may assume 
    $\Omega$ is connected. We treat $\Omega \in \pn(0,0)$ as a gadget
    with no danging edges.
    This proof uses a similar decomposition in the proof of \autoref{thm:generatepn}, but we may ignore some subtleties in that
    proof. The vertices in our $\Omega$ are
    not exclusively assigned signatures in $F$ or $\eq$, but
    the proof of \autoref{thm:generatepn} does not use any
    properties of these signatures. Here, we ignore
    vertices assigned $E_2$, since they do not affect the Holant
    value. Therefore, since the only difference between the
    $\e$ and $\f$-cases of extraction in the proof of
    \autoref{thm:generatepn} is the placement of vertices assigned $E_2$, we do not distinguish these cases. 
    
    Recall from the proof of \autoref{thm:generatepn} that,
    since constraint vertices can have arbitrary orientation
    in the signature grid,
    every extracted $\f^{m,d}$ gadget is really
    $(\f^{(r)})^{m,d}$ for some $r$, or the corresponding
    clockwise version. In the first case, \autoref{lem:rotateinput} states that, if $G$ is the
    signature satisfying $\u^{\otimes n}f = g$,
    $\u^{\otimes m}F^{m,d} = G^{m,d}\u^{\otimes d}$ is
    equivalent to $\u^{\otimes m} (F^{(r)})^{m,d}
    = (G^{(r)})^{m,d} \u^{\otimes d}$. Since $G$ inherits $F$'s orientation when we replace $F$ with $G$
    in the signature grid, we may assume below that every
    $\f$-type gadget is $\f^{m,d}$, with no rotation.
    If the extracted constraint vertex is oriented clockwise,
    the extracted gadget has signature
    matrix
    $((F^{(r)})^{d,m})^{\top}$ for some $r$, since transposing the extracted gadget reverses its
    central vertex's
    orientation in the signature grid from clockwise to counterclockwise and transposes its signature matrix.
    By \eqref{eq:transposeequiv},
    $\left(\left(F^{\left(r\right)}\right)^{d,m}\right)^{\top} = \left(\left(\overline{F}^{\left(r\right)}\right)^{\dagger}\right)^{m,d}$. Since $\fc$ and
    $\gc = \u\fc$ are compatible, we have
    $\u^{\otimes m}F^{m,d} = G^{m,d}\u^{\otimes d} \iff \u^{\otimes m} \overline{F}^{m,d}
    = \overline{G}^{m,d} \u^{\otimes d}$. Now by
    \autoref{lem:rotateinput} and \autoref{lem:conjugatecommute},
    \[
        \u^{\otimes m}F^{m,d} = G^{m,d}\u^{\otimes d} \iff
        \u^{\otimes m}\left(\left(\overline{F}^{(r)}\right)^{\dagger}\right)^{m,d}
        = \left(\left(\overline{G}^{(r)}\right)^{\dagger}\right)^{m,d}\u^{\otimes d}.
    \]
    Thus we may again assume the extracted gadget is
    $\f^{m,d}$ for the calculation below.
    
    Consider a decomposition of $\Omega$ in $p$ steps,
    choosing vertex $v_i$, assigned signature $F_i$, on step $i$ as the vertex with consecutive dangling edges.
    Since $\Omega$ initially has no dangling edges, we apply \eqref{eq:newnodangling} after choosing
    $v_1$, extracting a gadget
    with signature matrix $F_1^{0,n}$
    and producing only output dangling edges.
    Thus after choosing $v_1$, we apply 
    \eqref{eq:newleft1}/\eqref{eq:newleft2},  
    This rule leaves no input dangling
    edges in $\k'$, so after choosing every $v_i$ we always apply \eqref{eq:newleft1}/\eqref{eq:newleft2}.
    The extracted gadgets form a chain of compositions.
    Upon extracting all $p$ vertices from $\Omega'$, we find that there exist integers
    $m_2^p$, $r_1^{p-1}$, $t_1^{p-1}$, and $d_1^{p-1}$ satisfying
    \begin{enumerate}
        \item $m_i > 0$ for $2 \leq i \leq p$ and $r_i,t_i,d_i \geq 0$ for $1 \leq i \leq p-1$.
        \item $r_1 = t_1 = 0$
        \item $r_i + d_i + t_i = r_{i+1} + m_{i+1} + t_{i+1} > 0$ for $1 \leq i \leq p-2$
        \item $m_{p} = r_{p-1} + d_{p-1} + t_{p-1}$.
    \end{enumerate}
    such that the extracted gadgets have signature matrices
    \begin{align*}
        &V_1^f := F_1^{0,d_1} \\
        & V_i^f := E_2^{\otimes r_i} \otimes F_i^{m_i,d_i}  
            \otimes E_2^{\otimes t_i}
        \quad \text{for } 2 \leq i \leq p-1, \\
        &V_p^f := F_p^{m_p,0}.
    \end{align*}
    $\Omega$ is a gadget with no dangling edges, so its
    signature matrix is a $q^0 \times q^0 = 1 \times 1$ matrix, which we view as a scalar storing $\plholant_{\Omega}$.
    We have
    \begin{equation}
        \label{eq:holantdecompf}
        \plholant_{\Omega} = V_1^f \left(\prod_{i=2}^{p-1} V_i^f\right) V_p^f,
    \end{equation}

    Define $V^g_1$, $V^g_i$, and $V^g_p$ 
    by replacing replacing $F_i$ in $V^f_i$, $1 \leq i \leq p$,
    by $G_i$ satisfying $g_i = \u^{\otimes n}f_i$. We similarly have
    \begin{equation}
        \label{eq:holantdecompg}
        \plholant_{\Omega'} = V_1^g \left(\prod_{i=2}^{p-1} V_i^g\right) V_p^g.
    \end{equation}

    By assumption, \autoref{lem:tensorcftog},
    \eqref{eq:tensorinverse}, and \eqref{eq:tensordistribute}, we have, for any $2 \le i \le p-1$, (for brevity we suppress the
subscript $i$ below)
  \begin{eqnarray*}
    & &   \u^{\otimes r+m+t}
        \left(E_2^{\otimes r}
        \otimes F^{m,d} 
        \otimes E_2^{\otimes t}\right) \\
        &= &\left(\u^{\otimes r} E_2^{\otimes r}\right)
        \otimes \left(\u^{\otimes m} F^{m,d} \right)
        \otimes \left(\u^{\otimes t} E_2^{\otimes t}\right)\\
        &= & \left(E_2^{\otimes r} \u^{\otimes r}\right)
        \otimes \left( G^{m,d} \u^{\otimes d}\right)
        \otimes \left(E_2^{\otimes t} \u^{\otimes t}\right)\\
        &= & \left(E_2^{\otimes r}
            \otimes G^{m,d} 
        \otimes E_2^{\otimes t}\right)
        \u^{\otimes r+d+t}.
    \end{eqnarray*}
    Multiplying on the right by $(\u^{\otimes r+d+t})^{\dagger}$ gives
    \begin{equation}
        \label{eq:ftogcsp}
        \u^{\otimes r+m+t}
        \left(E_2^{\otimes r}
        \otimes F^{m,d} 
        \otimes E_2^{\otimes t}\right) (\u^{\otimes r+d+t})^{\dagger}
        = \left(E_2^{\otimes r}
        \otimes G^{m,d} 
        \otimes E_2^{\otimes t}\right).
    \end{equation}
    
    $(\u^{\otimes k})^{\dagger} \neq (\u^{\dagger})^{\otimes k}$ in general due to the noncommutativity of $\u$'s entries
    (one can view this as another reason why the Holant theorem fails in this case), 
    so we cannot distribute $(\u^{\otimes k})^{\dagger}$ over
    another tensor product like we do with $\u^{\otimes k}$. This is why we show
    $(\u^{\otimes k})^{\dagger}$ as a blackbox signature in \autoref{fig:insertu}, rather than as a separate signature
    on each wire as with $\u^{\otimes k}$. The above calculation, however, shows this doesn't matter;
    it's enough to have only the $\u^{\otimes r+m+t}$ distribute over each gadget.

    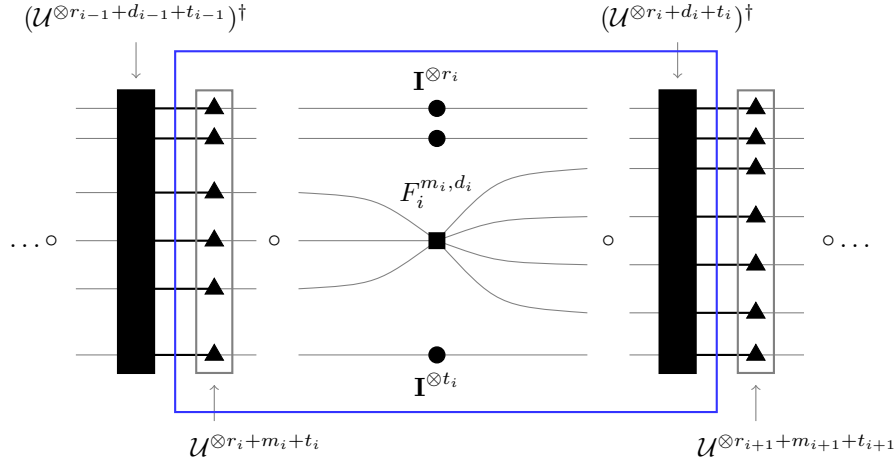
\begin{figure}[ht!]
        \center
        \begin{tikzpicture}[scale=0.8]
    \GraphInit[vstyle=Classic]
    \SetUpEdge[style=-]
    \SetVertexMath

    \def\ary{1.2}
    \def\bry{0.4}
    \def\cry{-0.4}
    \def\dry{-1.2}

    \def\aly{0.8}
    \def\bly{0}
    \def\cly{-0.8}

    \def\sx{8}

    \def\akvy{2.2}
    \def\bkvy{1.7}
    \def\fkvy{-1.9}

    \def\agx{\sx+3.5}
    \def\ulx{\sx-5}
    \def\urx{\sx+4}

    \def\ewirel{\ulx+2.7}
    \def\ewirer{\urx-1.5}

    \draw[thin, color=gray] (\sx, 0) .. controls (\sx + 1, \ary/1.1) .. (\ewirer, \ary);
    \draw[thin, color=gray] (\sx, 0) .. controls (\sx + 1, \bry/1.1) .. (\ewirer, \bry);
    \draw[thin, color=gray] (\sx, 0) .. controls (\sx + 1, \cry/1.1) .. (\ewirer, \cry);
    \draw[thin, color=gray] (\sx, 0) .. controls (\sx + 1, \dry/1.1) .. (\ewirer, \dry);

    \draw[thin, color=gray] (\ewirel, \aly) .. controls (\sx-1, \aly/1.1) .. (\sx, 0);
    \draw[thin, color=gray] (\ewirel, \bly) -- (\sx, 0);
    \draw[thin, color=gray] (\ewirel, \cly) .. controls (\sx-1, \cly/1.1) .. (\sx, 0);

    \node at (\ulx+2.3,0) {$\circ$};
    \node at (\urx-1.15,0) {$\circ$};
    \node at (\ulx-1.7,0) {$\ldots\circ$};
    \node at (\urx+2.8,0) {$\circ\ldots$};

    \draw[thin, color=gray] (\ewirel, \akvy) -- (\ewirer, \akvy);
    \draw[thin, color=gray] (\ewirel, \bkvy) -- (\ewirer, \bkvy);
    \draw[thin, color=gray] (\ewirel, \fkvy) -- (\ewirer, \fkvy);
    \Vertex[x=\sx,y=\akvy,L={\ii^{\otimes r_i}},Lpos=90]{m1}
    \Vertex[x=\sx,y=\bkvy,NoLabel]{m2}
    \Vertex[x=\sx,y=\fkvy,L={\ii^{\otimes t_i}},Lpos=270]{m3}

    \tikzset{VertexStyle/.style = {shape=rectangle, fill=black, minimum size=6pt, inner sep=1pt, draw}}
    \Vertex[x=\sx,y=0,L={F_i^{m_i,d_i}},Lpos=90,Ldist=0.2cm]{v2}

    \draw[thin, color=gray] (\ulx-1, \akvy) -- (\ulx+2, \akvy);
    \draw[thin, color=gray] (\ulx-1, \bkvy) -- (\ulx+2, \bkvy);
    \draw[thin, color=gray] (\ulx-1, \aly) -- (\ulx+2, \aly);
    \draw[thin, color=gray] (\ulx-1, \bly) -- (\ulx+2, \bly);
    \draw[thin, color=gray] (\ulx-1, \cly) -- (\ulx+2, \cly);
    \draw[thin, color=gray] (\ulx-1, \fkvy) -- (\ulx+2, \fkvy);

    \tikzset{VertexStyle/.style = {shape=rectangle, fill=black, minimum size=1pt, inner sep=1pt}}
    \Vertex[x=\ulx,y=\akvy,NoLabel]{ul1}
    \Vertex[x=\ulx,y=\bkvy,NoLabel]{ul2}
    \Vertex[x=\ulx,y=\aly,NoLabel]{ul3}
    \Vertex[x=\ulx,y=\bly,NoLabel]{ul4}
    \Vertex[x=\ulx,y=\cly,NoLabel]{ul5}
    \Vertex[x=\ulx,y=\fkvy,NoLabel]{ul6}

    \Vertex[x=\ulx+1.3,y=\akvy,NoLabel]{ulr1}
    \Vertex[x=\ulx+1.3,y=\bkvy,NoLabel]{ulr2}
    \Vertex[x=\ulx+1.3,y=\aly,NoLabel]{ulr3}
    \Vertex[x=\ulx+1.3,y=\bly,NoLabel]{ulr4}
    \Vertex[x=\ulx+1.3,y=\cly,NoLabel]{ulr5}
    \Vertex[x=\ulx+1.3,y=\fkvy,NoLabel]{ulr6}
    \foreach \unum in {1,2,...,6} {
        \Edge(ul\unum)(ulr\unum)
    };

    \node[draw, fill=black, regular polygon, regular polygon sides=3, minimum size = 8pt, inner sep = 1pt] at (\ulx+1.3,\akvy) {};
    \node[draw, fill=black, regular polygon, regular polygon sides=3, minimum size = 8pt, inner sep = 1pt] at (\ulx+1.3,\bkvy) {};
    \node[draw, fill=black, regular polygon, regular polygon sides=3, minimum size = 8pt, inner sep = 1pt] at (\ulx+1.3,\aly) {};
    \node[draw, fill=black, regular polygon, regular polygon sides=3, minimum size = 8pt, inner sep = 1pt] at (\ulx+1.3,\bly) {};
    \node[draw, fill=black, regular polygon, regular polygon sides=3, minimum size = 8pt, inner sep = 1pt] at (\ulx+1.3,\cly) {};
    \node[draw, fill=black, regular polygon, regular polygon sides=3, minimum size = 8pt, inner sep = 1pt] at (\ulx+1.3,\fkvy) {};

    \draw[color=black,thick,fill=black] (\ulx-0.3, \fkvy-0.3) rectangle (\ulx+0.3, \akvy+0.3);
    \draw[color=gray,thick] (\ulx+1, \fkvy-0.3) rectangle (\ulx+1.6, \akvy+0.3);
    \node at (\ulx,\akvy+1.5) {$(\u^{\otimes r_{i-1}+d_{i-1}+t_{i-1}})^{\dagger}$};
    \node at (\ulx+2,\fkvy-1.5) {$\u^{\otimes r_i+m_i+t_i}$};
    \draw[color=gray,->] (\ulx, \akvy+1.1) -- (\ulx, \akvy+0.5);
    \draw[color=gray,->] (\ulx+1.3, \fkvy-1.1) -- (\ulx+1.3, \fkvy-0.5);

    \draw[thin, color=gray] (\urx-0.8, \akvy) -- (\urx+2.1, \akvy);
    \draw[thin, color=gray] (\urx-0.8, \bkvy) -- (\urx+2.1, \bkvy);
    \draw[thin, color=gray] (\urx-0.8, \ary) -- (\urx+2.1, \ary);
    \draw[thin, color=gray] (\urx-0.8, \bry) -- (\urx+2.1, \bry);
    \draw[thin, color=gray] (\urx-0.8, \cry) -- (\urx+2.1, \cry);
    \draw[thin, color=gray] (\urx-0.8, \dry) -- (\urx+2.1, \dry);
    \draw[thin, color=gray] (\urx-0.8, \fkvy) -- (\urx+2.1, \fkvy);

    \Vertex[x=\urx,y=\akvy,NoLabel]{ur1}
    \Vertex[x=\urx,y=\bkvy,NoLabel]{ur2}
    \Vertex[x=\urx,y=\ary,NoLabel]{ur3}
    \Vertex[x=\urx,y=\bry,NoLabel]{ur4}
    \Vertex[x=\urx,y=\cry,NoLabel]{ur5}
    \Vertex[x=\urx,y=\dry,NoLabel]{ur6}
    \Vertex[x=\urx,y=\fkvy,NoLabel]{ur7}

    \Vertex[x=\urx+1.3,y=\akvy,NoLabel]{url1}
    \Vertex[x=\urx+1.3,y=\bkvy,NoLabel]{url2}
    \Vertex[x=\urx+1.3,y=\ary,NoLabel]{url3}
    \Vertex[x=\urx+1.3,y=\bry,NoLabel]{url4}
    \Vertex[x=\urx+1.3,y=\cry,NoLabel]{url5}
    \Vertex[x=\urx+1.3,y=\dry,NoLabel]{url6}
    \Vertex[x=\urx+1.3,y=\fkvy,NoLabel]{url7}
    \foreach \unum in {1,2,...,7} {
        \Edge(ur\unum)(url\unum)
    };

    \node[draw, fill=black, regular polygon, regular polygon sides=3, minimum size = 8pt, inner sep = 1pt] at (\urx+1.3,\akvy) {};
    \node[draw, fill=black, regular polygon, regular polygon sides=3, minimum size = 8pt, inner sep = 1pt] at (\urx+1.3,\bkvy) {};
    \node[draw, fill=black, regular polygon, regular polygon sides=3, minimum size = 8pt, inner sep = 1pt] at (\urx+1.3,\ary) {};
    \node[draw, fill=black, regular polygon, regular polygon sides=3, minimum size = 8pt, inner sep = 1pt] at (\urx+1.3,\bry) {};
    \node[draw, fill=black, regular polygon, regular polygon sides=3, minimum size = 8pt, inner sep = 1pt] at (\urx+1.3,\cry) {};
    \node[draw, fill=black, regular polygon, regular polygon sides=3, minimum size = 8pt, inner sep = 1pt] at (\urx+1.3,\dry) {};
    \node[draw, fill=black, regular polygon, regular polygon sides=3, minimum size = 8pt, inner sep = 1pt] at (\urx+1.3,\fkvy) {};

    \draw[color=black,thick,fill=black] (\urx-0.3, \fkvy-0.3) rectangle (\urx+0.3, \akvy+0.3);
    \draw[color=gray,thick] (\urx+1, \fkvy-0.3) rectangle (\urx+1.6, \akvy+0.3);
    \node at (\urx,\akvy+1.5) {$(\u^{\otimes r_{i}+d_{i}+t_{i}})^{\dagger}$};
    \node at (\urx+2,\fkvy-1.5) {$\u^{\otimes r_{i+1}+m_{i+1}+t_{i+1}}$};
    \draw[color=gray,->] (\urx, \akvy+1.1) -- (\urx, \akvy+0.5);
    \draw[color=gray,->] (\urx+1.3, \fkvy-1.1) -- (\urx+1.3, \fkvy-0.5);

    \draw[color=blue!80,thick] (\ulx+0.65, \fkvy-0.95) rectangle (\urx+0.65, \akvy+0.95);
\end{tikzpicture}
        \caption{A visualization of a gadget with signature matrix $V_i^f$ after inserting the identity 
            $(\u^{\otimes k})^{\dagger} \u^{\otimes k}$ on either side, with $r_i = 2$, $m_i = 3$, $d_i = 4$, $t_i = 1$.
            The `gadget' inside the blue box has signature matrix equal to the $i$th factor
            of the product in \eqref{eq:assocgadgetproduct}. Vertices assigned $\u$ are drawn as triangles.}
        \label{fig:insertu}
    \end{figure}

    Now \eqref{eq:ftogcsp} gives
    \[
        \u^{\otimes r_i+m_i+t_i} V_i^f (\u^{\otimes r_i+d_i+t_i})^{\dagger} = V_i^g
    \]
    for $2 \leq i \leq p-1$. The $m=0$ and $d=0$ cases of \autoref{lem:tensorcftog} give
    \[
        V_1^f (\u^{\otimes n})^{\dagger} = V_1^g
        \text{ and }
        \u^{\otimes m_p} V_p^f = V_p^g,
    \]
    respectively.
    Putting it all together, insert $I = (\u^{\otimes k})^{\dagger} \u^{\otimes k}$ between each factor in
    the product in the RHS of \eqref{eq:holantdecompf}. We can visualize this as inserting the
    `gadget' $(\u^{\otimes k})^{\dagger} \u^{\otimes k}$ between every two gadgets on the RHS of
    \eqref{eq:holantdecompf}, then apply associativity of multiplication in the style of the Holant theorem. See
    \autoref{fig:insertu}.
    By \eqref{eq:holantdecompf}, \eqref{eq:tensorinverse},
    conditions (i)-(iv), the three preceding equations, and \eqref{eq:holantdecompg}, we have
    \begin{align*}
        &\plholant_{\Omega}({\cal F}) \\
        &= V_1^f \left(\prod_{i=1}^{p-1} V_i^f\right) V_p^f \\
        &= V_1^f \left[(\u^{\otimes d_1})^{\dagger} \u^{\otimes d_1}\right] 
        \left(\prod_{i=2}^{p-1} V_i^f \left[(\u^{\otimes r_i+d_i+t_i})^{\dagger} \u^{\otimes r_i+d_i+t_i}\right]\right)
        V_p^f \\
        &= V_1^f (\u^{\otimes d_1})^{\dagger}
        \left(\prod_{i=2}^{p-1} \u^{\otimes r_i+m_i+t_i} V_i^f (\u^{\otimes r_i+d_i+t_i})^{\dagger} \right)
        \u^{\otimes m_p} V_p^f \numberthis \label{eq:assocgadgetproduct}\\
        &= V_1^g \left(\prod_{i=2}^{p-1} V_i^g\right) V_p^g \\
        &= \plholant_{\Omega'}(\u{\cal F}).
    \end{align*}
\end{proof}
    
\begin{corollary}
    Let $\fc$ and $\gc$ be compatible conjugate closed sets of constraint functions. If $\fc \cong_{qc} \gc$, then $Z(K) = Z(K_{\fc\to\gc})$ for every planar \#CSP$(\fc)$ instance $K$.
    \label{lem:forward}
\end{corollary}
\begin{proof}
Let $K$ be a  planar \#CSP$(\fc)$ instance and 
let $\Omega_K \in \pn(0,0)$, $\Omega_{K_{\fc\to\gc}} \in \mathcal{P}_{\cal G}(0,0)$ be the
Holant signature grids satisfying $Z(K) = \plholant_{\Omega_K}(\fc \mid \eq)$ and
$Z(K_{\fc\to\gc}) = \plholant_{\Omega_{K_{\fc\to\gc}}}(\gc \mid \eq)$. By
\eqref{eq:ue} and \autoref{thm:quantumholant} with the quantum permutation matrix $\u$ defining $\fc \cong_{qc} \gc$, we have
\[
    Z(K) = \plholant_{\Omega_K}(\fc \mid \eq)
    = \plholant_{\Omega_{K_{\fc\to\gc}}}(\gc \mid \eq)
    = Z(K_{\fc\to\gc}).
\]
\end{proof}

\section{Planar \#CSP Equivalence Implies Quantum Isomorphism}
\label{sec:backward}
Henceforth, we assume all constraint function sets $\fc$ are
finite.
\subsection{Quantum Permutation Groups}
\label{sec:quantumpermutationgroups}
\begin{definition}[compact matrix quantum group (CMQG), fundamental representation]
    \label{def:cmqg}
    For $q \in \mathbb{N}$, a \emph{compact matrix quantum group} (CMQG) $\mathcal{Q}$ of order $q$ 
    is defined by a unital $C^*$-algebra $C(\mathcal{Q})$ and a $q \times q$ matrix 
    $\u = (u_{ij}) \in C(\mathcal{Q})^{[q]^2}$
    whose entries $u_{ij}$ generate $C(\mathcal{Q})$, satisfying
    \begin{enumerate}
        \item $\u$ and $\u^{\top}$ are invertible, and
        \item The \emph{comultiplication} map $\Delta: C(\mathcal{Q}) \to C(\mathcal{Q}) \otimes C(\mathcal{Q})$ defined by
            $\Delta(u_{ij}) = \sum_{k \in [q]} u_{ik} \otimes u_{kj}$ is a $*$-homomorphism.
    \end{enumerate}
    $\u$ is called the \emph{fundamental representation} of the CMQG.
\end{definition}
CMQGs were originally defined by Woronowicz in \cite{woronowicz_compact_1987}, where they were called
`compact matrix pseudogroups'. 
The following CMQG was introduced by Wang in \cite{wang_quantum_1998}.
\begin{definition}[$S_q^+$]
    \label{def:sqplus}
    The \emph{quantum symmetric group} $S_q^+$ is defined by the universal $C^*$-algebra $C(S_q^+)$
    generated by
    the entries of a $q \times q$ matrix $\u = (u_{ij})$ subject to the relations
    that make $\u$ a quantum permutation matrix.
\end{definition}

\begin{definition}[$\qut(X)$ \cite{banica_quantum_2005}]
    For an undirected, unweighted graph $X$, the \emph{quantum automorphism group} $\qut(X)$ of $X$
    is defined by the universal $C^*$-algebra $C(\qut(X))$ generated by the entries of $\u = (u_{ij})$
    subject to the following relations:
    \begin{itemize}
        \item $\u$ is a quantum permutation matrix, and
        \item $\u A_X = A_X\u$, where $A_X$ is the adjacency matrix of $X$.
    \end{itemize}
\end{definition}
If we add the condition that the entries of $\u$ commute, then
it is known that $C(\qut(X))$ is isomorphic to $C(\text{Aut}(X))$, the commutative algebra of continuous complex linear functionals on $\text{Aut}(X)$, the classical automorphism group of $X$.
This isomorphism maps $u_{ij}$ to the characteristic function
of the automorphisms sending vertex $i$ ot $j$.
Without the commutativity condition, $\qut(X)$ doesn't actually exist
as a group, but we still think of $C(\qut(X))$ as the algebra
of continuous complex linear functionals on it. Hence the
absence of commutativity is what makes things ``quantum."

Just as a graph $X$'s classical automorphism group is a subgroup of $S_{V(X)}$, 
if $X$ has $q$ vertices, $\qut(X)$ is a \emph{quantum subgroup} of $S_q^+$ ($\qut(X) \subseteq S_q^+$) in the sense that there is a surjective unital $*$-homomorphism from $C(S_q^+)$ to $C(\qut(X))$ mapping each entry of the
fundamental representation of $S_q^+$ to the corresponding entry of the fundamental representation
of $\qut(X)$. When the CMQG's in question are
presented as a universal $C^*$-algebra generated by
the entries of the fundemantal representation subject to certain relations (such as $S_q^+$ and $\qut(X)$), we can also say
CMQGs $\mathcal{Q}_1 \subseteq \mathcal{Q}_2$ if the relations imposed on $\mathcal{Q}_2$ are implied
by those imposed on $\mathcal{Q}_1$. 

\begin{definition}[Quantum permutation group]
    A CMQG $\mathcal{Q}$ is a \emph{quantum permutation group} if it is a quantum subgroup of $S_q^+$ for some $q$, in which
    case we call $\mathcal{Q}$ an \emph{order}-$q$ quantum permutation group.
    If the fundamental representation $\u$ of $\mathcal{Q}$ is indexed by a (finite) set $X$, we say
    $\mathcal{Q}$ \emph{acts} on $X$.
\end{definition}
The above discussion shows $\qut(X)$ is a quantum permutation group.
In this work, we generalize $\qut(X)$ to the quantum permutation group
$\qut(\fc)$ for any set $\fc$ of arbitrary-arity tensors
over $\c$.
By \autoref{lem:tensorcftog}, for $n=2$ and $F^{1,1} = G^{1,1} = A_X$, we can rewrite the second relation defining
$\qut(X)$ as $\u^{\otimes 2} A_X^{2,0} = A_X^{2,0}$. This, along with the fact that
a classical permutation matrix $P$ defines an automorphism of arity-$n$ tensor $F$ 
if and only if $P^{\otimes n}f = f$, motivates the following definition:

\begin{definition}[$\qut(\fc)$, $\qut(F)$]
    For a set $\fc$ of constraint functions with $|V(\fc)| = q$, 
    the \emph{quantum automorphism group} $\qut(\fc)$ of $\fc$
    is defined by the universal $C^*$-algebra $C(\qut(\fc))$ generated by the entries of $q \times q$ matrix $\u = (u_{ij})$
    satisfying the following conditions (recall that $f = F^{\arity(F),0}$ is
    the column vector form of $F$):
    \begin{itemize}
        \item $\u$ is a quantum permutation matrix, and
        \item $\u^{\otimes \arity(F)} f = f$ \text{ for every $F \in \fc$}.
    \end{itemize}
    For a single constraint function $F$, define $\qut(F) = \qut(\{F\})$.
\end{definition}
In the case that $F$ is symmetric, binary, and 0-1 valued, $\qut(F)$ coincides with $\qut(X)$,
where $F = A_X$.
We have yet to prove that $\qut(\fc)$ is actually a CMQG. We do that now.
\begin{proposition}
    $\qut(\fc)$ is a quantum permutation group for any set $\fc$
    of constraint functions.
\end{proposition}
\begin{proof}
    To show $\qut(\fc)$ is a CMQG, it suffices to show that the map $\Delta: C(\qut(\fc)) \to C(\qut(\fc)) \otimes C(\qut(\fc))$ given in
    \autoref{def:cmqg} is a $*$-homomorphism. To do this, we must show that the images 
    $v_{xy} = \Delta(u_{xy}) = \sum_{z \in V(\fc)}u_{xz} \otimes u_{zy}$ also satisfy the relations imposed
    on the preimages $u_{xy}$ -- namely that the matrix $\vc = (v_{xy}) \in 
    (C(\qut(\fc)) \otimes C(\qut(\fc)))^{V(\fc)^2}$ is a quantum permutation matrix, and $\vc^{\otimes \arity(F)}f = f$ for every
    $F \in \fc$.

    It is known (see e.g. \cite[{Theorem 2.4}]{banica_quantum_2005}), and one can readily see, that
    $\vc$ defined this way is a quantum permutation matrix.
    Let $F \in \fc$ have arity $n$, so $\u^{\otimes n}f = f$. For $\vx \in V(\fc)^n$,
    \begin{align*}
        (\vc^{\otimes n} f)_{\vx} 
        &= \sum_{\vy} \left(\sum_{z_1} u_{x_1z_1} \otimes u_{z_1y_1}\right)
            \ldots \left(\sum_{z_n} u_{x_nz_n} \otimes u_{z_ny_n}\right) f_{\vy} \\
        &= \sum_{\vz} \sum_{\vy} (u_{x_1z_1} \ldots u_{x_nz_n}) \otimes (u_{z_1y_1} \ldots u_{z_ny_n})
            f_{\vy} \\
        &= \sum_{\vz}  (u_{x_1z_1} \ldots u_{x_nz_n}) \otimes 
        \left(\sum_{\vy} u_{z_1y_1} \ldots u_{z_ny_n} f_{\vy} \right) \\
        &= \left(\sum_{\vz} u_{x_1z_1} \ldots u_{x_nz_n} f_{\vz}\right) \otimes \one \\
        &= f_{\vx} \one \otimes \one.
    \end{align*}
    This holds for every $F \in \fc$, so $\Delta$ is a $*$-homomorphism. Thus $\qut(\fc)$ is a CMQG.

    $\qut(\fc) \subseteq S_q^+$ because the relations defining $S_q^+$ are a subset of those defining
    $\qut(\fc)$.
\end{proof}

\begin{definition}[Orbit, Orbital \cite{lupini_nonlocal_2018}]
    For a quantum permutation group $\mathcal{Q}$ with fundamental representation $\u$
    acting on (finite) set $X$, define relations $\sim_1$ on $X$ and $\sim_2$ on $X \times X$ by
    \[
        x \sim_1 y \iff u_{xy} \neq 0, \qquad (x_1, x_2) \sim_2 (y_1, y_2) \iff u_{x_1y_1} u_{x_2y_2}
        \neq 0.
    \]
    Then $\sim_1$ and $\sim_2$ are equivalence relations, and their equivalence classes are known
    as the \emph{orbits} and \emph{orbitals}, respectively, of $\mathcal{Q}$.
\end{definition}

\begin{definition}[$C_{\mathcal{Q}}(m,d)$, $C_{\mathcal{Q}}$]
    For CMQG $\mathcal{Q}$ with $q \times q$ fundamental representation $\u$, let
    \[
         C_{\mathcal{Q}}(m,d) = \{M \in \c^{q^m \times q^d} \mid \u^{\otimes m} M
             = M \u^{\otimes d} \}
    \]
    be the space of $(m,d)$-\emph{intertwiners} of $\mathcal{Q}$, and
    $C_{\mathcal{Q}} = \bigcup_{m,d} C_{\mathcal{Q}}(m,d)$.
\end{definition}
By \autoref{lem:tensorcftog}, we have $F^{m,d} \in C_{\qut(\fc)}(m,d)$
for every $m,d \geq 0$ and $(m+d)$-ary $F \in \fc$.

The next lemma provides a key connection between the intertwiners
and orbits of $\qut(\fc)$.
\begin{lemma}[{\cite[{Lemma 3.7}]{lupini_nonlocal_2018}}]
    \label{lem:constantorbits}
    Let $\mathcal{Q}$ be an order-$q$ quantum permutation group and $v \in \c^q$. Then
    $v \in C_{\mathcal{Q}}(1, 0)$ if and only if $v$ is constant on the orbits of $\mathcal{Q}$
    (that is, if $x \sim_1 y$, then $v_x = v_y$).
\end{lemma}

It is known (see \cite{banica_liberation_2009, planar}) that $C_{\q}$ is a 
\emph{tensor category with duals} (see also \autoref{sec:category}), 
in the sense that it satisfies the following properties:
\begin{enumerate}[label=(\roman*)]
    \item For fixed $m,d$, $C_{\q}(m,d)$ is a vector space over $\c$.
    \item $M \in C_{\q}(m,d), M' \in C_{\q}(d, w) \implies M \circ M' := M M'
        \in C_{\q}(m,w)$.
    \item $M \in C_{\q}(m,d), M' \in C_{\q}(m',d') \implies M \otimes M'
        \in C_{\q}(m+m',d+d')$.
    \item $M \in C_{\q}(m,d) \implies M^{\dagger} \in C_{\q}(d, m)$.
    \item $I = E^{1,1} \in C_{\q}(1,1)$.
    \item $E^{2,0} \in C_{\q}(2,0)$.
\end{enumerate}

\begin{remark}
    \label{rem:frobenius}
    We now remark on the relationship between \autoref{lem:rotategadget} and the invariance Lemmas 
    \ref{lem:tensorcftog} and \ref{lem:rotateinput}.
    For $F=G$, 
    if $F^{m,d} \in C_{\mathcal{Q}}$ for the quantum permutation group $\mathcal{Q}$ with fundamental representation $\u$, then 
    Lemmas \ref{lem:tensorcftog}-\ref{lem:conjugatecommute} assert that
    $F^{m_2,d_2}, (F^{(r)})^{m,d}, (F^\dagger)^{m,d} \in C_{\mathcal{Q}}$ as well, respectively, for
    any $r$ and $m_2+d_2 = m+d$.
    The proof of \autoref{lem:rotategadget} shows that $\f^{m_2,d_2}, (\f^{(r)})^{m,d} \in 
    \tcwdn{\ii, \e^{2,0}, \f^{m,d}}$. Taking signature matrices gives
    $F^{m_2,d_2}, (F^{(r)})^{m,d} \in \tcwdn{I, E^{2,0}, F^{m,d}}$, so, since $C_{\cal Q}$ is a 
    tensor category with duals, $F^{m_2,d_2}, (F^{(r)})^{m,d} \in C_{\mathcal{Q}}$ by properties
    (ii)-(vi).
    Similarly, by \eqref{eq:transposeequiv}, $(F^\dagger)^{m,d} \in C_{\mathcal{Q}}$ simply follows from 
    property (iv).

    The idea that $C_{\qut(\fc)}$ is invariant under changing the dimensions of its matrices is not new, and is an instance of a wider phenomenon called ``Frobenius duality'' in \cite{chassaniol_study_2019} and ``Frobenius reciprocity'' in 
    \cite{gromada_quantum_2021, banica_quantum_2005}. 
    See also \autoref{sec:category}.
\end{remark}

The following important result, referred to as \emph{Tannaka-Krein duality},
was proved in \cite{woronowicz_tannaka-krein_1988}, and expressed in the following form in \cite{chassaniol_study_2019} and elsewhere. See also
\autoref{sec:gadgetgroup}.
\begin{theorem}
    \label{thm:tannaka}
    The mapping $\q \mapsto C_{\q}$ induces a bijection between quantum permutation groups 
    acting on $[q]$ and tensor categories with duals $C$ satisfying $C_{S_q^+} \subseteq C$.
\end{theorem}

\subsection{Intertwiners and Signature Matrices}
The following lemma provides a concise characterization of the intertwiners of $\qut(\fc)$. \cite[{Proposition 5.5}]{chassaniol_study_2019}
proves the special case where $\qut(\fc)$ is $\qut(X)$ for a symmetric unweighted graph $X$.
\begin{lemma}
    \label{lem:qutfintertwiners}
    $C_{\qut(\fc)} = \tcwd{E^{1,0}, E^{1,2}, \{f \mid F \in \fc\}}$.
\end{lemma}
\begin{proof}
    It is known (see e.g. \cite{chassaniol_study_2019})
    that for a CMQG $\q$ with fundamental representation $\u$, 
    \[
        E^{1,0} \in C_{\q}(1,0) \text{ and } E^{1,2} \in C_{\q}(1,2) \iff \u \text{ is a quantum permutation matrix}.
    \]
    Also, by definition, we have $f \in C_{\q}(n,0) \iff \u^{\otimes n}f = f$. The backwards
    direction of each $\iff$, and the fact that $C_{\qut(F)}$ is a tensor category with duals,
    show $C_{\qut(F)} \supseteq \tcwd{E^{1,0}, E^{1,2}, \{f \mid F \in \fc\}}$. The other inclusion follows from
    Tannaka-Krein duality. Namely, any tensor category with duals containing $E^{1,0}, E^{1,2}$,
    and $f$ for each $F \in \fc$ is the intertwiner space $C_{\q}$ of some CMQG $\q$ whose fundamental 
    representation $\u'$ is a quantum permutation matrix satisfying $\u'^{\otimes n}f = f$ for every $F \in \fc$. Thus $\q$ is a quantum
    subgroup of $\qut(\fc)$, so $C_{\qut(\fc)} \subseteq C_{\q}$. 
    $\tcwd{E^{1,0}, E^{1,2}, \{f \mid F \in \fc\}}$ is the intersection of all tensor categories with duals containing
    $E^{1,0}, E^{1,2}$, and $\{f \mid F \in \fc\}$, so $C_{\qut(\fc)} \subseteq \tcwd{E^{1,0}, E^{1,2}, \{f \mid F \in \fc\}}$.
\end{proof}

Note that the RHS of \autoref{lem:qutfintertwiners} is the linear span of the signature matrices of the
gadgets in the RHS of \autoref{thm:generatepn}. This observation motivates the following definition.
\begin{definition}[$\qk(m,d)$]
    A \emph{planar $(m,d)$-quantum $\fc$-gadget} is a formal (finite) linear combination of gadgets in
    $\pn(m,d)$ with coefficients in $\c$. Let $\qk(m,d)$ be the collection of all
    planar $(m,d)$-quantum $\fc$-gadgets. Extend the signature matrix function $M$ linearly to
    $\qk(m,d)$.
\end{definition}
Recall from the start of \autoref{sec:decomposition} that in the context of graph homomorphisms 
$K \to X$, gadgets in $\mathcal{P}_{\{A_X\}}$ correspond to planar (bi-)labeled graphs. In this way,
$(m,d)$-quantum $\fc$-gadgets are a generalization of the $k$-labeled \emph{quantum graphs}
of \cite{freedman_reflection, lovasz, cai-lovasz}, but without the restriction that each vertex has
at most one label. 
Now, by \autoref{lem:gadgetmatrix}, applying $M$ to quantum $\fc$-gadgets composed of gadgets on the RHS of 
\autoref{thm:generatepn} yields the RHS of
\autoref{lem:qutfintertwiners}, so we have the following key
connection between the planar gadget decomposition and the intertwiners of $\qut(\fc)$.
\begin{theorem}
    \label{thm:intertwinersigmatrix}
    Let $\fc$ be conjugate closed. Then
    $C_{\qut(\fc)}(m,d) = \{M(\mathbf{Q}) \mid \mathbf{Q} \in \qk(m,d)\}$ for every
    $m,d \in \mathbb{N}$.
\end{theorem}

The next lemma is a quantum analogue of several similar classical results for the special case of 
graph homomorphism, including
\cite[{Lemma 2.4}]{lovasz}.
\begin{lemma}
    Let $\fc$ be conjugate closed. $x,y \in V(\fc)$ are in the same orbit of $\qut(\fc)$ if and only if
    $Z^x(K) = Z^y(K)$ for all 1-labeled planar \#CSP$(\fc)$ instances $K$.
    \label{lem:sameorbit}
\end{lemma}
\begin{proof}
    View $K$ as a gadget $\k \in \pn(1,0)$ with a single dangling
    edge attached to the equality vertex corresponding to the labeled variable (there is always an 
    embedding of the underlying multigraph such that this vertex is on the outer face).
    By \eqref{eq:onelabelequiv}, for every $z \in V(\fc)$, $Z^z(K) = (M(\k))_z$. Hence the
    condition $Z^x(K) = Z^y(K)$ for all 1-labeled planar \#CSP$(\fc)$ instances $K$ is equivalent to
    $M(\mathbf{Q})_x = M(\mathbf{Q})_y$ for every $\mathbf{Q} \in \qk(1,0)$, which in turn
    is equivalent to $v_x = v_y$ for every $\mathbf{v} \in C_{\qut(\fc)}(1,0)$ by 
    \autoref{thm:intertwinersigmatrix}. Now the result follows from \autoref{lem:constantorbits}
    (for if $x,y$ are in different orbits of $\qut(\fc)$ then the characteristic vector $\mathbf{v}$ 
    of $x$'s orbit is in $C_{\qut(\fc)}$ by \autoref{lem:constantorbits} and satisfies
    $v_x = 1 \neq 0 = v_y$).
\end{proof}

\subsection{Arity Reduction and Projective Connectivity}
The final two steps are the most involved. First, we prove a higher-dimensional analogue of \cite[{Theorem 4.5}]{lupini_nonlocal_2018},
which states that for connected (undirected, unweighted) graphs $X$ and $Y$, if there exist vertices
$x \in V(X)$ and $y \in V(Y)$ in the same orbit of the quantum automorphism group of the disjoint
union of $X$ and $Y$, then $X \cong_{qc} Y$. This theorem is proved using the concept of the
\emph{quantum orbital algebra} of a graph $X$, defined by the orbitals of $\qut(X)$.
A natural generalization of the quantum orbital algebra to $n$-dimensional tensors would require defining
an analogue of orbits and orbitals acting on $n$-tuples of vertices. However, as discussed in
\cite[{Section 5}]{lupini_nonlocal_2018}, it is unclear whether such ``higher order orbits'' are
even well-defined for $n > 2$. Thus our strategy will be to reduce the dimension of our tensor to 2 while
preserving its inclusion in $\qut(\fc)$, then use the orbital algebra to extract the required 
information about $\u$.

First, we must handle separately the case where $\fc$ and $\gc$ contain only unary $(n=1)$ functions. All
relevant results above, including \autoref{lem:forward}, apply to any set of constraint functions, including those with only unary functions.
By \cite[Theorem 1]{young2022equality}, any compatible $\fc$ and $\gc$ are \emph{classically} isomorphic
if and only if $Z(K) = Z(K_{\fc\to\gc})$ for every (not necessarily planar) $K$. However,
if $\fc$ is a set of unary constraint functions, every
\#CSP$(\fc)$ instance is planar, as the underlying graph is a disjoint union of star graphs. Since
quantum isomorphism is a relaxation of classical isomorphism, we have the following.
\begin{lemma}
    \label{lem:n1}
    Let $\fc$, $\gc$ be sets of unary constraint functions.
    If $Z(K) = Z(K_{\fc\to\gc})$ for every planar \#CSP$(\fc)$ instance $K$, then $\fc \cong_{qc} \gc$.
\end{lemma}
\begin{corollary}
    Let $\fc$ and $\gc$ be sets of unary constraint functions.
    Then $\fc \cong_{qc} \gc$ if and only if $\fc \cong \gc$.
\end{corollary}

From this point on, assume $\fc$ and $\gc$ each contain at least one function with arity $n \geq 2$. We now recall some definitions which apply to all quantum permutation
groups.
\begin{definition}[Coherent algebra \cite{coherent}]
    \label{def:coherent}
    A \emph{coherent algebra} is a unital $\c$-algebra of complex matrices under usual matrix multiplication and 
    addition, closed under matrix transpose and Hadamard (entrywise) product $\bullet$, and containing the
    all-ones matrix.
\end{definition}
\cite[{Theorem 3.10}]{lupini_nonlocal_2018} states that, for any quantum permutation group $\mathcal{Q}$
acting on a set $X$,
if $O_1,\ldots,O_m \subseteq X \times X$ 
are the orbitals of $\mathcal{Q}$, then the span of the matrices $\{A^i \mid i \in [m]\}$ defined by
$A^i_{x,y} =  1$ if $(x,y) \in O_i$ and $A^i_{x,y} = 0$ otherwise
forms a coherent algebra, called the \emph{orbital algebra} of $\mathcal{Q}$. When $\mathcal{Q}$ is $\qut(X)$ for a graph $X$,
we call the orbital algebra of $\mathcal{Q}$ the
\emph{quantum orbital algebra} of $X$.
Observe that a matrix $M$ is in the orbital algebra of $\mathcal{Q}$ if and only if $M$ is constant on the orbitals of $\mathcal{Q}$ -- 
that is, if $(x_1,y_1)$ and $(x_2,y_2)$ are in the same orbital, then $M_{x_1,y_1} = M_{x_2,y_2}$.
\begin{theorem}[{\cite[{Theorem 3.12}]{lupini_nonlocal_2018}}]
    Let $\mathcal{Q}$ be a quantum permutation group acting on a set $X$, with fundamental representation
    $\u$, and let $M \in \c^{X^2}$. Then $\u M = M\u$ iff $M$ is in the orbital algebra
    of $\mathcal{Q}$.
    \label{thm:orbitalalgebra}
\end{theorem}
Ideally, we would prove an extension of \autoref{thm:orbitalalgebra} for tensors
$M \in \c^{X^n}$ for $n > 2$. However, such a result would require a higher-dimensional orbital algebra, hence higher-dimensional orbitals. As discussed above, such
objects are not known to exist. Instead, we take the dimensionality reduction approach below.

\begin{lemma}
    Let $F$ be an $n$-ary constraint function with $n > 2$
    and let $\u$ be the fundamental representation of $\qut(F)$.
    Define an arity-$(n-1)$ constraint function $F'$ by
    \[
        F'_{x_2\ldots x_n} = \sum_{x_1}F_{x_1x_2\ldots x_n}.
    \]
    Then $\u^{\otimes n-1}f' = f'$ (where $f'$ is the vectorization of $F'$).
    \label{lem:dimreduction}
\end{lemma}
\begin{proof}
    We have $\u^{\otimes n}f = f$, so for all $\mathbf{x}$,
    \begin{equation}
        \sum_{\mathbf{y}} u_{x_1y_1} \ldots u_{x_ny_n} f_{\mathbf{y}} = f_{\mathbf{x}}.
        \label{eq:entrywiseuff}
    \end{equation}
    For any $x_1, z \in [q]$, left multiply \eqref{eq:entrywiseuff}  by $u_{x_1z}$, and then sum over all  values of $x_1$
    and $z$. 
    The right hand side sum is 
    \[\sum_{x_1}\left( \sum_{z} u_{x_1 z} \right) f_{\mathbf{x}}
    = \sum_{x_1} f_{\mathbf{x}} = f'_{x_2\ldots x_n}
    \]
    The left hand side sum is
    \begin{align*}
        \sum_{\mathbf{y}, z} \left(\sum_{x_1} u_{x_1 z}u_{x_1 y_1}\right)
         u_{x_2y_2}\ldots u_{x_ny_n} f_{\mathbf{y}}
        & = \sum_{y_2^n} u_{x_2y_2}\ldots u_{x_ny_n} \left(\sum_{z} f_{zy_2^n}\right) \\
        & = \sum_{y_2^n} u_{x_2y_2}\ldots u_{x_ny_n}  f'_{y_2\ldots y_n},
    \end{align*}
    hence
    $(\u^{\otimes n-1} f')_{x_2^n} = f'_{x_2^n}$.
\end{proof}

After $n-2$ applications of  \autoref{lem:dimreduction}, the resulting binary constraint function is the adjacency matrix of a $\c$-weighted graph.
We would like this $\c$-weighted graph to be connected in the following sense.
\begin{definition}[(Dis)connected $\c$-weighted graph]
    \label{def:connectedgraph}
    A $\c$-weighted graph $X$ is \emph{disconnected} if there exist $I, J \subseteq V(X)$ such that
    $I \cap J = \varnothing$, $I \cup J = V(X)$, and $X_{uv} = 0$ unless $u,v \in I$ or $u,v \in J$.
    $X$ is \emph{connected} if it is not disconnected.
\end{definition}
This corresponds to the usual definition of connected for undirected graphs $X$.
Instead of defining connectivity for higher-dimensional tensors via a generalization of 
\autoref{def:connectedgraph}, we give a definition motivated by \autoref{lem:dimreduction}.
\begin{definition}[Projectively connected]
    For $n \geq 2$, an $n$-ary constraint function $F$ is \emph{projectively connected}
    if the $\c$-weighted graph $X$ defined by
    \[
        X_{uv} = \sum_{x_1,x_2,\ldots,x_{n-2}} F_{x_1x_2\ldots x_{n-2}uv}
    \]
    is connected.
    \label{def:projectivelyconnected}
\end{definition}

\begin{definition}[$\oplus$]
\label{def:oplus-f-g}
    Let $F \in \c^{V(F)^n}, G \in \c^{V(G)^n}$ be constraint functions of
    the same arity $n$, but possibly different domain sizes $|V(F)|$ and $|V(G)|$.
    For $n \geq 2$, the \emph{direct sum} $F \oplus G \in \c^{(V(F) \sqcup V(G))^n}$ (where $\sqcup$
    denotes set disjoint union) of $F$ and $G$ is an
    arity-$n$ constraint function defined by
    \[
        (F \oplus G)_{\mathbf{x}} = \begin{cases} 
            F_{\mathbf{x}} & \mathbf{x} \in V(F)^n \\
            G_{\mathbf{x}} & \mathbf{x} \in V(G)^n \\
            0 & \text{otherwise}
        \end{cases}
    \]
    for $\mathbf{x} \in (V(F) \sqcup V(G))^n$.
    If $F$ and $G$ have arity $n = 1$, then define $F \oplus G \in \c^{(V(F) \sqcup V(G))^2}$ as a \emph{binary} function
    \[
        (F \oplus G)_{xy} = \begin{cases}
            F_x & x = y \in V(F) \\
            G_x & x = y \in V(G) \\
            0 & \text{otherwise}.
        \end{cases}
    \]
    For constraint function sets $\fc$ and $\gc$ of size $s$, define $\fc \oplus \gc = \{F_i \oplus G_i \mid i \in [s]\}$.
\end{definition}
Observe that for $n=2$, the direct sum corresponds to the disjoint union of the real-weighted
graphs whose adjacency matrices are $F$ and $G$.

Now we prove a generalization of \cite[{Theorem 4.5}]{lupini_nonlocal_2018}. 
\begin{lemma}
    Let $\fc$ and $\gc$ be conjugate closed constraint function sets with disjoint domains
    $V(\fc)$ and $V(\gc)$, respectively (not necessarily of the same size), and further assume that
    $\fc$ and $\gc$ contain a pair of corresponding projectively connected constraint functions
    (such functions have arity $> 1$).
    If there is some $\hat{x} \in V(\fc)$, $\hat{y} \in V(\gc)$ in the same orbit of $\qut(\fc \oplus \gc)$,
    then $\fc \cong_{qc} \gc$.
    \label{lem:sameorbitiso}
\end{lemma}
\begin{proof}
    Let $\u$ be the fundamental representation of $\qut(\fc \oplus \gc)$, and let $F \in \c^{V(\fc)^n}$ 
    and $G \in \c^{V(\gc)^n}$ be the projectively connected constraint functions in $\fc$ and $\gc$,
    respectively. Define a $\c$-weighted graph $Z$ by
    \[
        Z_{uv} = \sum_{z_1,z_2,\ldots,z_{n-2} \in V(\fc) \sqcup V(\gc)} (F \oplus G)_{z_1z_2\ldots z_{n-2}uv}
    \]
    for $u,v \in V(Z) = V(\fc) \sqcup V(\gc) = V(\fc \oplus \gc)$.
    One can see that $Z$ is the disjoint union of $\c$-weighted graphs $X$ and $Y$ 
    on $V(X) = V(\fc)$ and $V(Y) = V(\gc)$, respectively, where
    \[
        X_{uv} = \sum_{z_1,z_2,\ldots,z_{n-2} \in V(\fc)} F_{z_1,z_2,\ldots, z_{n-2},u,v},\qquad
        Y_{uv} = \sum_{z_1,z_2,\ldots,z_{n-2} \in V(\gc)} G_{z_1,z_2,\ldots, z_{n-2},u,v},
    \]
    Since $F$ and $G$ are projectively connected, $X$ and $Y$ are connected. By \autoref{lem:dimreduction}, $\u^{\otimes 2}z = z$.
    Thus by \autoref{lem:tensorcftog} and \autoref{thm:orbitalalgebra}, $Z$ is in the orbital algebra of $\qut(\fc \oplus \gc)$.
    The orbital algebra, being a coherent algebra, is closed
    under $\bullet$ and $\top$, 
    so $Z^{\bullet m}$ and $(Z^{\bullet m})^{\top}$ are also in the orbital algebra for all
    $m \geq 0$. For each distinct pair $u,v \in V(Z)$, if there is some $m_{uv} \ge 1$ such that
    $Z_{uv}^{m_{uv}} = -Z_{vu}^{m_{uv}}$ and is nonzero, then there is a $p \geq 1$ such that 
    $Z_{uv} = \omega_{2m_{uv}}^p Z_{vu}$, where $\omega_r$ is the $r$th primitive root of unity. Letting $m$ be the
    LCM of all such $m_{uv}$, $Z' = Z^{\bullet 4m} + (Z^{\bullet 4m})^{\top}$ is in the orbital algebra; it  is
    the adjacency matrix
    of a symmetric $\c$-weighted graph, and satisfies 
    $Z_{uv} \neq 0 \implies Z'_{uv} \neq 0$. Hence the
    subgraphs $X'$ and $Y'$ of $Z'$ induced by $V(X), V(Y) \subseteq V(Z)$ are still connected, 
    and are now symmetric.

    Recall that $(Z'^w)_{uv}$ stores the sum of the products of the edge weights of all $w$-edge walks
    from $u$ to $v$ in $Z'$. Since $Z'$ is disconnected into $X'$ and $Y'$, we have 
    $(((Z')^{\bullet d})^w)_{xy} = (((Z')^{\bullet d})^w)_{yx} = 0$ for all $w,d \geq 1$ and $x \in V(X')=V(X), y \in V(Y')=V(Y)$. 
    Fix $x', x'' \in V(X')$. Since $X'$ is symmetric and connected, there is a $w \geq 1$-edge walk 
    $x' = x_1, x_2, \ldots, x_{w+1} = x''$ of vertices in $V(X')$ such that 
    \begin{equation}
        X'_{x'x_2}X'_{x_2x_3}\ldots X'_{x_{w-1}x_w}X'_{x_wx''} \neq 0.
        \label{eq:nonzerowalk}
    \end{equation}
    Suppose that for every $d \geq 1$,
    \begin{equation}
        (((Z')^{\bullet d})^w)_{x'x''} = \sum_{z_2,z_3,\ldots,z_{w} \in V(X)} \left(X'_{x'z_2}X'_{z_2z_3}\ldots X'_{z_{w-1}z_w}
        X'_{z_wx''}\right)^d = 0.
        \label{eq:walk}
    \end{equation}
    After collecting equal terms into equivalence classes, we obtain a Vandermonde system from
    \eqref{eq:walk} for all $d \leq D$ for an appropriate $D$. This implies every term in the sum in
    \eqref{eq:walk} is 0, contradicting \eqref{eq:nonzerowalk}. Thus there are some $w,d \geq 1$
    such that $(((Z')^{\bullet d})^w)_{x'x''} \neq 0$, but 
    $(((Z')^{\bullet d})^w)_{xy} = 0$ for every $x \in V(X'), y \in V(Y')$.
    The orbital algebra of $\qut(\fc \oplus \gc)$ contains $Z'$ and is a coherent algebra, so it contains $((Z')^{\bullet d})^w$.
    Thus, by the discussion after \autoref{def:coherent}, 
    $((Z')^{\bullet d})^w$ is constant on the orbitals
    of $\qut(\fc \oplus \gc)$. Hence for every $x, x', x'' \in V(X)$ and $y \in V(Y)$, 
    $(x', x'')$ and $(x,y)$ cannot be in the same orbital, so $u_{x'x''}u_{xy} = u_{xy}u_{x'x''} = 0$.

    Now we can mostly follow the proof of \cite[{Theorem 4.5}]{lupini_nonlocal_2018}. 
    For each $x \in V(X)$, let $p_x = \sum_{y \in V(Y)} u_{xy}$. For $x, x' \in V(X)$, we have
    \[
        p_x - p_xp_{x'} = p_x(1-p_{x'}) = \left(\sum_{y \in V(Y)} u_{xy}\right) \left(\sum_{x'' \in V(X)} u_{x'x''}\right) = 0
    \]
    by the last sentence of the previous paragraph, so $p_x = p_xp_{x'}$. Similarly,
    $p_{x'} = p_xp_{x'}$, so $p_x = p_{x'}$. If we define $p_y = \sum_{x \in V(X)} u_{xy}$, for $y \in V(Y)$, then a
    symmetric calculation shows $p_y = p_{y'}$ for any $y, y' \in V(Y)$. Finally, for $x \in V(X)$,
    $y \in V(Y)$,
    \[
        |V(X)| p_x = \sum_{x \in V(X)}p_x = \sum_{x \in V(X), y \in V(Y)} u_{xy}
        = \sum_{y \in V(Y)} p_y = |V(Y)| p_y.
    \]
    For any $x \in V(X)$, $y\in V(Y)$, we have
    \begin{align*}
        |V(X)|u_{xy} &= |V(X)|\sum_{y' \in V(Y)}u_{xy'}u_{xy} = (|V(X)|p_x) u_{xy} \\
        &= (|V(Y)|p_y) u_{xy} = |V(Y)|\sum_{x' \in V(X)}u_{x'y} u_{xy}
        = |V(Y)| u_{xy}.
        \numberthis\label{eq:q1q2}
    \end{align*}
    Specializing to $u_{\hat{x}\hat{y}} \neq 0$,
    we get  $|V(X)| = |V(Y)|$. Thus $p_x = p_y$, so $p_z$ does not depend
    on the choice of $z \in V(Z)$. Call this common element $p$.
    If $p_{\hat{x}} = 0$, then multiplying $p_{\hat{x}}$ by
    $u_{\hat{x}\hat{y}}$ would produce $u_{\hat{x}\hat{y}} = 0$,
    a contradiction. Hence
    $p_{\hat{x}} \neq 0$, so $p \neq 0$.

    Consider the $C^*$-subalgebra of $C(\qut(\fc \oplus \gc))$ generated by $u_{xy}$ for $x \in V(X), y \in V(Y)$. The first two
    equalities in \eqref{eq:q1q2} show $p u_{xy} = u_{xy}$, and
    a similar calculation shows $u_{xy} p = u_{xy}$.
    Thus $p$ is the identity in this subalgebra. 
    Consider the matrix $\u' = (u_{xy})$ for $x \in V(X) = V(\fc), y \in V(Y) = V(\gc)$ -- the top right quadrant of $\u$.
    Note that $\u'$ does not depend on the choice of
    $F$ from $\fc$ or $G$ from $\gc$.
    $\u'$ is square, as $|V(X)| = |V(Y)|$.
    Since $p$ is
    equal to each row- and column-sum of $\u'$ and $\u'$'s entries are taken from the quantum
    permutation matrix $\u$, $\u'$ is a quantum permutation matrix over the subalgebra.

    $\u$ is the fundamental representation of $\fc \oplus \gc$, so for any corresponding arity-$n$ $F \in \fc$ and $G \in \gc$ with $n \geq 2$, \autoref{lem:tensorcftog} with
    $m=n-1$ and $d=1$ gives $\u^{\otimes n-1} (F \oplus G)^{n-1,1} = (F \oplus G)^{n-1,1} \u$.
    For $x_1,\ldots,x_{n-1} \in V(\fc)$, and $y \in V(\gc)$, we have
    \begin{align*}
        ((\u')^{\otimes n-1} G^{n-1,1})_{x_1\ldots x_{n-1}, y}
        &= \sum_{z_1,\ldots,z_{n-1} \in V(\gc)} u_{x_1z_1}\ldots u_{x_{n-1}z_{n-1}} G_{z_1\ldots z_{n-1}y}\\
        &= \sum_{z_1,\ldots,z_{n-1} \in V(\fc) \sqcup V(\gc)} u_{x_1z_1}\ldots u_{x_{n-1}z_{n-1}} (F \oplus G)_{z_1\ldots z_{n-1}y}\\
        &= (\u^{\otimes n-1} (F \oplus G)^{n-1,1})_{x_1\ldots x_{n-1}, y} \\
        &= ((F \oplus G)^{n-1,1} \u)_{x_1\ldots x_{n-1}, y} \\
        &= \sum_{z \in V(\fc) \sqcup V(\gc)} (F \oplus G)_{x_1\ldots x_{n-1}z} u_{zy}\\
        &= \sum_{z \in V(\fc)} F_{x_1\ldots x_{n-1}z} u_{zy}\\
        &= (F^{n-1,1} \u')_{x_1 \ldots x_{n-1}, y}.
    \end{align*}
    Thus $(\u')^{\otimes n-1} G^{n-1,1} = F^{n-1,1} \u'$, so by \autoref{lem:tensorcftog},
    $(\u')^{\otimes n} g = f$. 
    
    If $n = 1$, recall $F \oplus G$ has arity 2, so we still
    have $\u (F \oplus G)^{1,1} = (F \oplus G)^{1,1} \u$, so
    for all $x \in V(\fc)$, $y \in V(\gc)$,
    \[
        u_{xy}g_y = \sum_{z \in V(\fc) \sqcup V(\gc)} u_{xz}
        (F \oplus G)_{zy} = \sum_{z \in V(\fc) \sqcup V(\gc)}
        (F \oplus G)_{xz} u_{zy} = f_xu_{xy}.
    \]
    Now, for $x \in V(\fc)$,
    \[
        (\u' g)_x
        = \sum_{y \in V(\gc)} u_{xy} g_y
        = \sum_{y \in V(\gc)} f_x u_{xy}  = f_x,
    \]
    so $\u' g = f$. Thus $\fc \cong_{qc} \gc$.
\end{proof}

\begin{lemma}
    \label{lem:backward}
    Let $\fc$ and $\gc$ be conjugate closed compatible sets of constraint functions.
    If $Z(K) = Z(K_{\fc\to\gc})$ for every planar \#CSP$(\fc)$ instance $K$, then $\fc \cong_{qc} \gc$.
\end{lemma}
\begin{proof}
    Let $0_F$ and $0_G$ be new domain elements.
    For each $F \in \fc$, $G \in \gc$ with arity $n \geq 2$ define constraint functions $F'$ and $G'$ on
    $V(\fc) \sqcup \{0_F\}$ and $V(\gc) \sqcup \{0_G\}$, by letting, for
    $\vx \in V(\fc) \sqcup \{0_F\}$ or $\vx \in V(\gc) \sqcup \{0_G\}$, respectively,
    \begin{equation}
        \label{eq:fprimegprime}
        F'_{\vx} = \begin{cases}
            F_{\vx} & \vx \in V(\fc)^n \\
            \gamma & \vx = (0_F, 0_F, \ldots, 0_F, c),\ c \neq 0_F \\
            0 & \text{otherwise}
        \end{cases}
        , \qquad
        G'_{\vx} = \begin{cases}
            G_{\vx} & \vx \in V(\gc)^n \\
            \gamma & \vx = (0_G, 0_G, \ldots, 0_G, c),\ c \neq 0_G \\
            0 & \text{otherwise}
        \end{cases},
    \end{equation}
    for some fixed $\gamma \in \mathbb{R}\setminus \{0\}$.
    Defining a $\c$-weighted graph $X'$ by 
    $X'_{uv} = \sum_{x_1, x_2, \ldots, x_{n-2}}F'_{x_1x_2\ldots x_{n-2}uv}$, we see that
    $X'_{0_Fv} = \gamma \neq 0$ for every $v \in V(X') \setminus \{0_F\} = V(\fc)$.
    Thus $X'$ is connected. Defining $Y'$ from $G'$ similarly, 
    we see $F'$ and $G'$ are projectively connected.
    
    For each $F \in \fc$, $G \in \gc$ with arity 1, define
    $F'$ and $G'$ on $V(\fc) \sqcup \{0_F\}$ and $V(\gc) \sqcup \{0_G\}$, respectively, by
    \[
        F'_x = \begin{cases} F_x & x \in V(\fc) \\ 0 & x = 0_F
        \end{cases}, \qquad
        G'_x = \begin{cases} G_x & x \in V(\gc) \\ 0 & x = 0_G
        \end{cases}.
    \]
    Let $\fc' = \{F' \mid F \in \fc\}$ and $\gc' = \{G' \mid G \in \gc\}$, with 
    $V(\fc') = V(\fc) \sqcup \{0_F\}$
    and $V(\gc') = V(\gc) \sqcup \{0_G\}$.
    
    By assumption, $\fc$ and $\gc$ have at least one non-unary
    constraint function, so $\fc'$ and $\gc'$ contain a pair of corresponding projectively
    connected constraint functions.
    We now show that 
    \begin{equation}
        Z^{0_F}(K) = Z^{0_G}(K)
        \label{eq:1labeledequality}
    \end{equation}
    for every planar 1-labeled \#CSP$(\fc' \oplus \gc')$ instance $K = (V, C)$.
    Let $v_0$ be the labeled variable in $V$.
    If $K$ is not connected, then the components of $K$ that do not 
    contain $v_0$ contribute the same value to the partition function regardless of the value taken
    by $v_0$. 
    Hence, to establish~\eqref{eq:1labeledequality}
    we may assume $K$ is connected.
    As such, when $v_0$ takes value $0_F \in V(\fc')$, 
    all variables of $K$ must take values in $V(\fc')$
    (otherwise the assignment contributes 0 to the partition function, as every $F' \oplus G'$ takes value 0 unless all its inputs are in $V(\fc')$ or all its inputs are in $V(\gc')$). 
    We stratify the sum over assignments $\sigma: V \to V(\fc')$ on the left side of \eqref{eq:1labeledequality}
    by the set of variables $S = \sigma^{-1}(0_F) \subseteq V$ assigned $0_F$. 
    For all $\sigma$, $v_0 \in \sigma^{-1}(0_F)$.
    Let $\chi_S = 1$ if no unary constraints $F' \oplus G'$ 
    are applied to any variable in $S$, and if any variable in $S$ appears in a constraint
    of arity $n \ge 2$, then the first $n-1$ arguments of the constraint 
    are in $S$ and the $n$th argument is not.
    By the definition of $F'$, $\chi_S = 1$ captures the property that
    an assignment $\sigma$ with $\sigma^{-1}(0_F) = S$ does not contribute 0 to the partition
    function.
    We define $\chi_S = 0$ otherwise. 
    Suppose $\chi_S = 1$. 
    The remaining `free' variables $v \in V \setminus S$ satisfy $\sigma(v) \in V(\fc)$ and either 
    appear in no constraints with any variables in $S$, or appear as
    the last argument of a constraint with all other variables in $S$, 
    in which case the constraint takes value 
    $\gamma$.
    $F' \oplus G'$ for unary constraints $F'$ and $G'$ only receives inputs in $V(\fc)$ ($\chi_S = 1$ so such a constraint does not receive input $0_F$), so by the special definition of $\oplus$ for unary signatures (\autoref{def:oplus-f-g}), $F' \oplus G'$ is
    effectively the constraint
    $(x,y) \mapsto F_x$ if $x=y$, and 0 otherwise. Thus we may
    combine the two variables to which this constraint is applied and replace
    it with $F$ applied to the merged variable, without changing
    the partition function.
    
    Let $K^{\fc}_{V \setminus S}$ be the \#CSP$(\fc)$ instance formed by eliminating all variables in 
    $S$ and replacing every function $F' \oplus G'$ with $F$. In other words, $K^{\fc}_{V \setminus S}$ is the \#CSP$(\fc)$ instance corresponding to the subgraph of $K$'s signature grid induced by
    the vertices corresponding to variables in $V \setminus S$ and the constraints applied to only
    variables in $V \setminus S$, with the above function
    substitutions.
    Let $d_S$ be the number of constraints containing
    any variable in $S$ (each of which, as $\chi_S = 1$, is of the form $F'_{v_1 \ldots v_{n-1} u}$ for
    some $v_1,\ldots,v_{n-1} \in S$ and $u \in V \setminus S$). 

    Critically, $\chi_S$ is determined by $S$, and does not
    depend on individual assignments $\sigma$ which
    has $S$ as the preimage set of $0_F$. Thus
    for a fixed $S$ with $\chi_S = 1$, the remaining variables $V \setminus S$ take all values in $V(\fc') \setminus \{0_F\} = V(\fc)$ as $\sigma$ ranges over $\{\sigma \mid \sigma^{-1}(0_F) = S\}$. Hence the sum over all assignments $\sigma$ with
    $\sigma^{-1}(0_F) = S$ of the product of the constraints containing no variables in $S$ is
    $Z(K^{\fc}_{V \setminus S})$. The $d_S$ constraints containing variables in $S$ each contribute
    $\gamma$ regardless of the value of the free variable in the $n$th argument. So
    \[
        Z^{0_F} (K) = \sum_{S \subseteq V, S \ni v_0} \chi_S 
        \gamma^{d_S} Z(K^{\fc}_{V \setminus S}).
    \]
    The above reasoning depends only on the relationship between $0_F$ and $V(\fc)$, which by construction
    is exactly the same as the relationship between $0_G$ and $V(\gc)$. So we get a similar expression
    for $Z^{0_G} (K)$. $K^{\fc}_{V\setminus S}$ 
    and $K^{\gc}_{V\setminus S}$ are
    planar \#CSP$(\fc)$ and \#CSP$(\gc)$ instances, respectively, and by construction $K^{\gc}_{V\setminus S} = (K^{\fc}_{V\setminus S})_{\fc\to\gc}$. Thus by assumption we have
    \[
        Z^{0_F} (K) = 
        \sum_{S \subseteq V, S \ni v_0} \chi_S \gamma^{d_S} Z(K^{\fc}_{V \setminus S}) =
        \sum_{S \subseteq V, S \ni v_0} \chi_S \gamma^{d_S} Z(K^{\gc}_{V \setminus S}) =
        Z^{0_G} (K),
    \]
    proving \eqref{eq:1labeledequality}. Now \autoref{lem:sameorbit} asserts that $0_F$ and $0_G$
    are in the same orbit of $\qut(\fc' \oplus \gc')$, so by \autoref{lem:sameorbitiso}, $\fc' \cong_{qc} \gc'$.

    The final step is to extract a quantum isomorphism between $\fc$ and $\gc$ from the quantum permutation matrix
    $\u = (u_{uv})_{u \in V(\gc'), v \in V(\fc')}$ satisfying $\u^{\otimes n}f' = g'$ (vector forms of $F'$ and $G'$)
    for every $n$-ary $F \in \fc$ and corresponding $G \in \gc$. Define the matrix 
    $\widehat{\u} = (u_{uv})_{u \in V(\gc), v \in V(\fc)}$ (in other words, we eliminate row $0_G$ and
    column $0_F$ from $\u$). We must show that $\widehat{\u}$ is a quantum permutation matrix and that
    $\widehat{\u}^{\otimes n}f = g$. First, if $n \geq 2$, for $\vx \in V(\gc)^n$, we have
    \begin{align*}
        g_{\vx} &= g'_{\vx} = \sum_{\vy \in V(\fc')^n} u_{x_1y_1}\ldots u_{x_ny_n} f'_{\vy} \\
                 &= \sum_{\vy \in V(\fc)^n} u_{x_1y_1}\ldots u_{x_ny_n} f_{\vy}
        + \gamma \cdot u_{x_1 0_F} \ldots u_{x_{n-1} 0_F} \sum_{y_n \in V(\fc)} u_{x_n y_n} \\
                 &= \left(\widehat{\u}^{\otimes n} f\right)_{\vx}
        + \gamma \cdot u_{x_1 0_F} \ldots u_{x_{n-1} 0_F} (\one - u_{x_n 0_F}) \numberthis\label{eq:extraterm}.
    \end{align*}
    If $\vx$ does not satisfy $x_1 = x_2 = \ldots = x_{n-1} \neq x_n$, then 
    the second term of the RHS of \eqref{eq:extraterm} is 0, giving the desired
    $g_{\vx} = (\widehat{\u}^{\otimes n} f)_{\vx}$. If $x_1 = x_2 = \ldots = x_{n-1} \neq x_n$,
    then \eqref{eq:extraterm} becomes
    \begin{equation}
        g_{\vx} = (\widehat{\u}^{\otimes n} f)_{\vx} + \gamma \cdot u_{x_1 0_F}.
        \label{eq:takenorm}
    \end{equation}
    Since each $u_{xy}$ is a (self-adjoint) projector, we have $\|u_{x_10_F}\| \in \{0,1\}$ and
    $\|(\widehat{\u}^{\otimes n} f)_{\vx}\| \leq \sum_{\vy \in V(\fc)^n} |f_{\vy}|$. So if we choose
    \[
        \gamma > \sum_{\vy \in V(\fc)^n} |f_{\vy}| + \max_{\vx \in V(\gc)^n} |g_{\vx}|,
    \]
    then \eqref{eq:takenorm} is impossible unless $u_{x_10_F} = 0$, again giving 
    $g_{\vx} = (\widehat{\u}^{\otimes n} f)_{\vx}$. By applying the above to all $\vx \in V(\gc)^n$,
    we obtain $g = \widehat{\u}^{\otimes n} f$, as well as $u_{x_10_F} = 0$ for all $x_1 \in V(\gc)$.
     The latter implies $u_{0_G0_F} = \one$, so $u_{0_G y_1} = 0$ for all $y_1 \in V(\fc)$ as well. Thus the
    rows and columns of $\widehat{\u}$ still sum to $\one$, so $\widehat{\u}$ is a quantum permutation matrix.
    
    If $n = 1$, for any $x \in V(\gc)$,
    \[
        g_x = g'_x = \sum_{y \in V(\fc')} u_{xy} f'_y
        = \sum_{y \in V(\fc)} u_{xy} f_y,
    \]
    so $\widehat{\u} f = g$. Therefore $\fc \cong_{qc} \gc$.
\end{proof}
The proof of \autoref{lem:backward} was inspired by Lov\'asz's proof of \cite[{Corollary 2.6}]{lovasz}, which is roughly
the classical version of \autoref{thm:result}
restricted to positive-real-weighted graphs.
For unweighted graphs, one can use the complement to assume WLOG that both graphs are connected. This is the approach used
in \cite{planar}.
Since there is no concept of complement for weighted graphs, the proof of 
\cite[{Corollary 2.6}]{lovasz} \emph{makes} both graphs connected by adding new vertices, analogous to $0_F$ and $0_G$,
adjacent to every existing vertex. Then it recovers an isomorphism of the original graphs.
    
Now \autoref{thm:result} follows 
from \autoref{lem:forward}, \autoref{lem:n1}, and \autoref{lem:backward}.

\subsection{Gadget Quantum Groups}
\label{sec:gadgetgroup}
In this section, we introduce a broad generalization of the
graph-theoretic quantum groups defined in \cite{planar}.
To motivate graph-theoretic quantum groups, \cite{planar}
defines a family $\mathfrak{F}$ of bi-labeled graphs to be a
\emph{graph category} if $\mathfrak{F}$ contains $\mathbf{I}$ and
$\mathbf{M^{2,0}}$ (the bi-labeled graphs corresponding to the
gadgets $\mathbf{I} = \e^{1,1}$ and $\e^{2,0}$) and is closed
under bi-labeled graph $\otimes, \circ$, and ${\dagger}$. 
Analogously, for any set $\fc$ of signatures and family $\mathfrak{F}$ of gadgets in the context of $\holant(\fc)$,
say $\mathfrak{F}$ is a \emph{$\fc$-gadget category} if it
contains $\e^{1,1}$ and $\e^{2,0}$ and is closed under
gadget $\otimes, \circ$, and ${\dagger}$. Our \autoref{thm:generatepn}
shows that $\pn$ is a gadget category for any $\fc$.

As with bi-labeled graphs and their homomorphism matrices,
the equivalence between gadget and
signature matrix operations then implies that, for any
$\fc$-gadget category $\mathfrak{F}$, $C_{\mathfrak{F}} = \spn\{M(\k) \mid \k \in \mathfrak{F}\}$ is a tensor category with duals.
Recall the statement \autoref{rem:frobenius} that we may replace $\f$ in \autoref{lem:rotategadget} with any $\k \in \pn$. More generally,
any $\fc$-gadget category $\mathfrak{F}$ contains $\e^{1,1}$
and $\e^{2,0}$, the gadgets used in the proof of \autoref{lem:rotategadget}. The vertex in both $\e^{1,1}$ and $\e^{2,0}$ is assigned signature $E_2$, which we may ignore
(treat as an edge) when computing a signature matrix. Thus we
may apply \autoref{lem:rotategadget} to any $\k \in \mathfrak{F}$
in place of $\f$ to show that for any arity-$n$ tensor
$H$, viewed as the signature function of some gadget in $\mathfrak{F}$, we have $H^{m_1,d_1} \in C_{\mathfrak{F}}
\iff H^{m_2,d_2} \in C_{\mathfrak{F}}$ for every $m_1,d_1,m_2,d_2 \geq 0$, $m_1+d_1 = m_2+d_2 = n$.

A more general version of the Tannaka-Krein duality in \autoref{thm:tannaka}, as expressed
in \cite{banica_liberation_2009, planar} and elsewhere, states
that the map sending a CMQG to its intertwiner space induces a
bijection between sub-CMQGs of $O_q^+$ (defined in \cite{wang_free_1995} as the CMQG such that $C(O_q^+)$ is the
universal $C^*$-algebra generated by the entries of an orthogonal
matrix whose entries are projectors -- $S_q^+ \subset O_q^+$)
and all tensor categories with duals.
Thus for any set $\fc$ of signature functions and $\fc$-gadget
category $\mathfrak{F}$, $C_{\mathfrak{F}}$ is the intertwiner
space of some CMQG $\mathcal{Q} \subseteq O_q^+$. We call
such a CMQG $\mathcal{Q}$ a \emph{gadget quantum group}, inspired by the similarly-defined graph-theoretic quantum groups in \cite{planar}. When $\fc = \{A_X\} \cup \eq$ for a binary symmetric 0-1 valued $A_X$ and the defining
properties of $\mathfrak{F}$ depend only on the signature grid's
underlying graph, the CMQG whose intertwiner space is $C_{\mathfrak{F}}$ is a graph-theoretic quantum group
parameterized by the graph $X$ whose adjacency matrix is $A_X$.

\cite{planar} introduced graph-theoretic quantum groups as a
generalization of the \emph{easy quantum groups} defined in
\cite{banica_liberation_2009} by showing that the easy
quantum groups are precisely the graph-theoretic quantum groups
arising from graph categories of bi-labeled graphs with no
internal edges. Such bi-labeled graphs correspond to
$\holant(A_X \mid \eq)$ gadgets with no vertices assigned $A_X$.
In our framework, then, easy quantum groups are the gadget
quantum groups arising from $\fc$-gadget categories for
$\fc \subseteq \eq$. \cite{planar} notes that graph-theoretic
quantum groups, like easy quantum groups, have a combinatorial
structure (in the form of their intertwiner spaces, which by Tannaka-Krein duality uniquely determine the quantum group)
that facilitates their study, but are a richer class than easy
quantum groups, since, for any graph category, every graph $X$ gives rise to a new graph-theoretic quantum group. Gadget quantum groups have a
similar combinatorial structure to graph-theoretic quantum groups, but are an even richer class, since any finite set $\fc$
of signatures, rather than a single signature $A_X$, gives rise
to a new quantum group. 

\appendix
\section{Appendix: An Alternate Approach to Connectivity}
\label{sec:connectivity}
Recall that the proof of \autoref{lem:sameorbitiso} makes use of the orbitals of
$\qut(\fc \oplus \gc)$. This construction, as mentioned earlier, does not extend to dimensions higher
than 2. This is why we, via projective connectivity, effectively
require that $\fc$ and $\gc$ contain a \emph{binary} connected constraint function in the hypothesis
of \autoref{lem:sameorbitiso}. To satisfy this hypothesis, we ensure that the modified constraint
functions $F'$ and $G'$ in the proof of \autoref{lem:backward} are projectively connected.
In this appendix, we present a different construction, due to Roberson \cite{roberson}, which, 
rather than modify the existing constraint functions, adds new binary connected constraint 
functions to $\fc$ and $\gc$, while preserving quantum isomorphism. This removes the need for projective
connectivity entirely, simplifies the proof of \autoref{lem:sameorbitiso}, and could
simplify the proof of \autoref{lem:backward}, since $F'$ and $G'$ no longer have to
be projectively connected (though we still need $Z^{0_F}(K) = Z^{0_G}(K)$ for all planar 1-labeled $K$).
The alternate construction also makes use of two lemmas which should be of independent interest. 

First, we extend \autoref{def:compatible} to gadgets.
\begin{definition}[$\k_{\fc\to\gc}$]
    For compatible constraint function sets $\fc$ and $\gc$ and $\k \in \pn$, let $\k_{\fc\to\gc}
    \in {\cal P}_{\gc}$ be the \emph{corresponding} gadget formed by replacing each constraint signature $F_i \in \fc$
    with the corresponding $G_i \in \gc$. Extend this mapping linearly to $\qk$.
\end{definition}
The first lemma shows that, viewing intertwiners themselves as constraint functions,
we may add `corresponding' pairs of intertwiners (the signature matrices of corresponding quantum $\fc$ and
$\gc$-gadgets -- recall \autoref{thm:intertwinersigmatrix}) to $\fc$ and $\gc$, while preserving
equivalence on planar \#CSP instances.
\begin{lemma}[\cite{roberson}]
    \label{lem:rob1}
    Let $\fc$ and $\gc$ be compatible CC constraint function sets. Let $M_F \in C_{\qut(\fc)}(m,d)$ and
    $M_G \in C_{\qut(\gc)}(m,d)$ such that $M_F = M({\bf Q})$ and 
    $M_G = M({\bf Q}_{\fc \to \gc})$ for some quantum $\fc$-gadget ${\bf Q} \in \qk(m,d)$.
    Let $F$ and $G$ be the constraint functions satisfying $F^{m,d} = M_F$ and $G^{m,d} = M_G$, and
    let $\fc' = \fc\cup\{F, \overline{F}\}$ and $\gc' = \gc\cup\{G, \overline{G}\}$. Then
    $Z(K) = Z(K_{\fc\to\gc})$ for all planar \#CSP$(\fc)$ instances $K$ if and only if 
    $Z(K) = Z(K_{\fc'\to\gc'})$ for all planar \#CSP$(\fc')$ instances $K$.
\end{lemma}
\begin{proof}
    The backward direction is immediate. Let $K$ be a planar \#CSP$(\fc')$ instance, and let
    $\Omega_K$ be the corresponding $\plholant(\fc'\mid \eq)$ instance. Create a `quantum
    signature grid' $\widehat{\Omega}_K \in \mathfrak{Q}^{\cal P}_{\fc}(0,0)$ as follows.
    Replacing every instance of a vertex $v$ assigned
    $F$ in $\Omega_K$ by the equivalent quantum gadget $\vq \in \qk$, matching the cyclically-ordered
    dangling edges of each gadget to the cyclically-ordered incident edges of $v$ 
    (and contracting the edges between the adjacent equality vertices to preserve bipartiteness).
    Then similarly replace every vertex assigned $\overline{F}$ by the gadget $\overline{\vq}$ formed
    by conjugating every signature in $\vq$ (we have $\overline{\vq} \in \qk$ because $\fc$ is CC).
    Similarly create $\widehat{\Omega}_{K_{\fc'\to\gc'}} \in \mathfrak{Q}^{\cal P}_{\gc}(0,0)$ by 
    replacing every instance of a vertex assigned $G$ or $\overline{G}$ in $\Omega_{K_{\fc'\to\gc'}}$ 
    with $\vq_{\fc\to\gc}$ or $\overline{\vq_{\fc\to\gc}} \in \mathfrak{Q}^{\cal P}_{\gc}$, respectively.
    Then the index-$\alpha$ summand of $\widehat{\Omega}_K$ or 
    $\widehat{\Omega}_{K_{\fc'\to\gc'}}$ is a planar \#CSP$(\fc)$ instance 
    $K^{\alpha}_{\fc}$ or planar \#CSP$(\gc)$ instance $K^{\alpha}_{\gc}$, respectively, and furthermore
    $K^\alpha_{\gc} = (K^{\alpha}_{\fc})_{\fc\to\gc}$. Thus
    \begin{align*}
        Z(K) &= \plholant_{{\Omega}_K}(\fc' \mid \eq) \\
             &= \plholant_{\widehat{\Omega}_K}(\fc \mid \eq) \\
             &= \plholant_{\widehat{\Omega}_{K_{\fc'\to\gc'}}}(\gc \mid \eq)\\
             &= \plholant_{\Omega_{K_{\fc'\to\gc'}}}(\gc' \mid \eq)\\
             &= Z(K_{\fc'\to\gc'}).
    \end{align*}
\end{proof}

An alternate proof would also need the following lemma, which is equivalent to \autoref{lem:rob1} once
\autoref{thm:result} is proved.
\begin{lemma}[\cite{roberson}]
    \label{lem:rob2}
    Suppose $\fc$, $\gc$, $F$, $G$, $\fc'$, and $\gc'$ satisfy the hypotheses of \autoref{lem:rob1}.
    Then $\fc \cong_{qc} \gc \iff \fc' \cong_{qc} \gc'$.
\end{lemma}
\begin{proof}
    The backward direction is immediate. Let $\u$ be the quantum permutation matrix defining the quantum
    isomorphism between $\fc$ and $\gc$. It suffices to show 
    $\u^{\otimes m+d}f = g$, and $\u^{\otimes m+d} \overline{f} = \overline{g}$, 
    or equivalently, by \autoref{lem:tensorcftog}, 
    $\u^{\otimes m} M_F (\u^{\otimes d})^{\dagger} = M_G$ and 
    $\u^{\otimes m} \overline{M_F} (\u^{\otimes d})^{\dagger} = \overline{M_G}$.
    These identities follows from the proof of
    the quantum Holant theorem. Indeed, while the statement in
    \autoref{thm:quantumholant} only applies to signature grids (gadgets in $\qk(0,0)$), the proof
    may be easily modified to apply to $\vq \in \qk(m,d)$ as follows. After inserting 
    $(\u^{\otimes k})^{\dagger} \u^{\otimes k}$ between each pair of gadgets in the 
    \autoref{thm:generatepn} decomposition of (each summand of) $\vq$ and
    reassociating to convert every $\fc$ signature to the corresponding $\gc$ signature and fix
    the internal equality vertices, we must effectively apply an additional $\u$ or $\u^{\dagger}$ to each original dangling 
    edge of $\vq$ to fully transform each equality vertex with a dangling edge back to itself via
    $\u^{\otimes a} E^{a,b} (\u^{\otimes b})^{\dagger} = E^{a,b}$. Thus 
    $\u^{\otimes m} \vq (\u^{\otimes d})^{\dagger} = \vq_{\fc\to\gc}$, so 
    $\u^{\otimes m+d} f = g$. $\u^{\otimes m+d} \overline{f} = \overline{g}$ follows
    by the same logic applied to $\overline{\q}$, since $\fc$ and $\gc$ are compatible.
\end{proof}
Together, Lemmas \ref{lem:rob1} and \ref{lem:rob2} show that, to prove \autoref{thm:result}, we may
assume that $\fc$ and $\gc$ contain (the constraint functions created from) any intertwiners $M_F$
and $M_G$ which are the signature matrices of corresponding quantum $\fc$ and $\gc$-gadgets. 
In particular, we trivially have $(\e^{1,0} \circ \e^{0,1})_{\fc\to\gc} = (\e^{1,0} \circ \e^{0,1})$,
so we may assume $\fc$ and $\gc$ both contain $M(\e^{1,0} \circ \e^{0,1}) = J$, the all-1s matrix.
$J$ is a connected constraint function, so we may immediately apply \autoref{lem:sameorbitiso}. 
Moreover, $J \oplus J$ already satisfies the properties $((Z')^{\bullet d})^w$ was carefully constructed 
to have in the proof of \autoref{lem:sameorbitiso}: 
$(J \oplus J)_{x'x''} \neq 0$ and $(J \oplus J)_{xy} = 0$ for all $x,x',x'' \in V(\fc)$ and $y \in
V(\gc)$, and similarly for $\fc$ and $\gc$ swapped.

\section{The Complex-Weighted Graph Isomorphism Game}
\label{sec:game}
Quantum graph isomorphism was first defined in \cite{asterias} in
terms of a \emph{nonlocal game} called the \emph{graph isomorphism game}. Briefly, a two-player nonlocal game is played between a verifier and
two cooperating players, Alice and Bob. The verifier gives Alice and Bob inputs
$x_A$ and $x_B$ from finite sets $X_A$ and $X_B$, respectively.
After receiving their inputs, Alice and Bob respond with outputs
$y_A$ and $y_B$ from finite sets $Y_A$ and $Y_B$, respectively.
The verifier indicates whether Alice and Bob have won (1) or lost (0) by a predicate $V: X_A \times X_B \times Y_A \times Y_B \to \{0,1\}$.
We are concerned with whether Alice and Bob have a \emph{perfect strategy} - a strategy that wins the game with probability 1.
Alice and Bob can agree on a predetermined strategy, but cannot communicate after receiving their inputs. However, if allowed a quantum strategy, they can make measurements on a shared quantum
state after receiving their inputs. In particular, we consider
\emph{quantum commuting strategies} (the namesake of the `qc' in
$\cong_{qc}$). In this framework, Alice and
Bob share a state in a possibly infinite-dimensional Hilbert space $H$ and Alice's measurement operators -- positive bounded linear
operators on $H$ -- must commute with Bob's measurement operators.

\paragraph*{The $\c$-weighted graph isomorphism game}
Let $F \in \c^{V(F)^2}$, $G \in \c^{V(G)^2}$ be $\c$-weighted 
graphs. We define the nonlocal \emph{$(F,G)$-isomorphism game}
as follows. Let $X_A = X_B = Y_A = Y_B = V(F) \cup V(G)$.
Let $f_A = \{x_A, y_A\} \cap V(F)$ and $g_A = \{x_A,y_A\} \cap V(G)$, and define $f_B$ and $g_B$ similarly for Bob.
$f_A$ and $g_A$ are well-defined (assuming the players win) given 
condition (i) below.
The players win if and only if the following three conditions are satisfied:
\begin{enumerate}[label=(\roman*)]
    \item $x_A \in V(F) \iff y_A \in V(G)$ and $x_B \in V(F) \iff y_B \in V(G)$;
 \item $f_A = f_B \iff g_A = g_B$;   \item $F_{f_A f_B} = G_{g_A g_B}$.
\end{enumerate}
In other words, $V(x_A,y_A,x_B,y_B) = 1$ if $x_A,y_A,x_B,y_B$
satisfy (i)-(iii), and $V(x_A,y_A,x_B,y_B) = 0$ otherwise.
Say $F \cong_{qcg} G$ if Alice and Bob have a perfect quantum
commuting strategy for the $(F,G)$-isomorphism game. We will show $F \cong_{qcg} G \iff F \cong_{qc} G$.
We thank David Roberson \cite{roberson} for
several improvements to results and proofs in this section.

If $F$ and $G$ are symmetric and 0-1 valued, then the $(F,G)$-isomorphism game is the graph isomorphism
game defined in \cite{asterias} for the graphs whose adjacency
matrices are $F$ and $G$. Many of the results in \cite{asterias}
extend to the $\c$-weighted graph isomorphism game. First, since
Alice and Bob must define a common bijection $V(F) \to V(G)$ by conditions (i) and (ii), $F \cong G$ if and only if Alice and Bob have
a perfect \emph{classical} strategy for the $(F,G)$-isomorphism game.
Second, by (i) and (ii) and since any quantum strategy is
\emph{non-signalling}, $F \cong_{qcg} G \implies |V(F)| = |V(G)|$
~\cite[Lemma 4.1]{asterias}. Third, since the $\c$-weighted graph
isomorphism game is, like the graph game, a \emph{synchronous game}
(the players share the same input set, the same output set, and
$V(y,y', x,x) = 0$ for $y \neq y'$), we have the following extension of \cite[Theorem 5.14]{asterias}:
\begin{lemma}
    \label{lem:qcgclass}
    Let $F \in \c^{V(F)^2}$, $G \in \c^{V(G)^2}$. Then $F \cong_{qcg} G$ if and only if there is a $C^*$ algebra
    $\mathcal{A}$ which admits a faithful tracial state, and projectors
    $u_{fg} \in {\cal A}$ for $f \in V(F)$, $g \in V(G)$ such that
    \begin{enumerate}
        \item $\u = (u_{fg})_{f \in V(F), g \in V(G)}$ is a quantum permutation matrix.
        \item $u_{f_1g_1} u_{f_2g_2} = 0$ if $F_{f_1f_2} \neq G_{g_1g_2}$
    \end{enumerate}
\end{lemma}
A \emph{state} is a linear functional $s: \mathcal{A} \to \c$
satisfying $s(\one) = 1$ and and $s(a^*a) \geq 0$ for all $a \in \mathcal{A}$, and a state is \emph{tracial} if $s(ab) = s(ba)$ for
all $a,b \in \mathcal{A}$, and is \emph{faithful} if $s(a^*a) = 0 \iff a = 0$. 

$F \cong_{qcg} G$ and $F \cong_{qc} G$ both imply
$|V(F)| = |V(G)|$, so let $V(F) = V(G) = [q]$. For graph isomorphism, it
is known~\cite[Theorem 2.5]{lupini_nonlocal_2018} (see also
\cite[Lemma 5.8]{asterias}) that condition 2 is equivalent to
$F \u = \u G$. We have the following extension for $\c$-weighted graphs $F$ and $G$:
\begin{lemma}
    \label{lem:guuh}
    Let $F \in \c^{[q]^2}$, $G \in \c^{[q]^2}$ and let $\u = (u_{fg})_{f,g \in [q]}$ be a quantum 
    permutation matrix. Then $F \u = \u G$ if and only if
    $F_{f_1f_2} \neq G_{g_1g_2} \implies u_{f_1g_1} u_{f_2g_2} = 0$.
\end{lemma}
\begin{proof}
    The proof is a simple modification of the proof of \cite[Lemma 3.9]{lupini_nonlocal_2018}, which
    is the case $F = G$.
    Suppose $F_{f_1f_2} \neq G_{g_1g_2} \implies u_{f_1g_1} u_{f_2g_2} = 0$. Then
    for any $\vg \in [q]^2$,
    \[
        \left((\u^{{\dagger}})^{\otimes 2} f\right)_{\vg}
        = \sum_{\vf \in [q]^2} u_{f_1g_1}u_{f_2g_2} F_{\vf}
        = \sum_{\vf:\ F_{\vf} = G_{\vg}} u_{f_1g_1} u_{f_2g_2} G_{\vg}
        = G_{\vg}.
    \]
    Now by \autoref{lem:tensorcftog}, $F \u = \u G$.
    Conversely, suppose $(\u^{\dagger})^{\otimes 2}f = g$, or equivalently $\sum_{\vf} u_{f_1g_1}
    u_{f_2g_2} F_{\vf} = G_{\vg}$ for every $\vg$. Multiplying on the left by $u_{f' g_1}$ and on the
    right by $u_{f'' g_2}$ for any $f'$ and $f''$ 
    yields $F_{f'f''} u_{f' g_1} u_{f'' g_2} = G_{g_1g_2} u_{f' g_1} 
    u_{f''g_2}$. Hence, if $u_{f' g_1} u_{f'' g_2} \neq 0$, then $F_{f'f''} = G_{g_1g_2}$.
\end{proof}

Next, we extend \cite[Theorem 4.4]{lupini_nonlocal_2018} from graphs to
$\c$-weighted graphs. 
\begin{theorem}
    Let $F \in \c^{V(F)^2}$, $G \in \c^{V(G)^2}$. Then
    $F \cong_{qcg} G \iff F \cong_{qc} G$.
    \label{thm:qcg}
\end{theorem}
\begin{proof}
    By \autoref{lem:guuh} and \autoref{lem:qcgclass}
    (applied to $\u^{\dagger}$, also a quantum permutation matrix),
    $F \cong_{qcg} G \implies F \cong_{qc} G$. 
    However, our definition of $\cong_{qc}$ did not assume that the 
    $C^*$-algebra over which $\u$ is defined admits a faithful tracial state. To show
    $F \cong_{qc} G \implies F \cong_{qcg} G$, it
    suffices to construct a faithful tracial state on
    the $C^*$-algebra $\mathcal{A}$ over which the quantum permutation matrix
    witnessing $F \cong_{qc} G$ is defined.
    Aside from a few small differences, the construction follows the proof of 
    \cite[Theorem 4.4]{lupini_nonlocal_2018} (this theorem restricted to unweighted graphs), 
    which doubles as the
    proof of \cite[Theorem 4.5]{lupini_nonlocal_2018}, the graph case of our \autoref{lem:sameorbitiso}.
    The proof in \cite{lupini_nonlocal_2018} assumes $F$ and $G$ are connected by taking the complement, 
    which is not possible for $\c$-weighted graphs. Instead, as in the alternate proof of 
    \autoref{lem:sameorbitiso} in \autoref{sec:connectivity}, we may assume
    $\fc = \{F\}$ and $\gc = \{G\}$ also contain the trivially connected all-1s signature of
    $\e^{1,0} \circ \e^{0,1}$ (the assumption in \autoref{lem:rob2} that $\fc$ and $\gc$ are CC is not
    necessary to add this signature). Now the proof of \cite[Theorem 4.4]{lupini_nonlocal_2018} goes
    through on these connected signature sets, where we use \autoref{lem:guuh} in place
    of the corresponding statement for unweighted graphs.
\end{proof}
Combining \autoref{thm:qcg} with \autoref{thm:result}, we have the following.
\begin{corollary}
    Let $F, G \in \r^{[q]^2}$.
    There is a perfect quantum commuting strategy for the
    $(F,G)$-isomorphism game if and only if
    $Z(K) = Z_{F \to G}(K)$ for every planar \#CSP$(F)$ instance $K$ 
    (equivalently, iff the $\r$-weighted graphs with adjacency matrices $F$ and $G$ admit the
    same number of real-weighted homomorphisms from all planar directed graphs). Furthermore, for
    $F,G \in \c^{[q]^2}$, if there is a perfect quantum commuting stategy for the $(F,G)$-isomorphism
    game, then $Z(K) = Z_{F \to G}(K)$ for every planar \#CSP$(F)$ instance $K$.
\end{corollary}
\begin{proof}
    It remains to show the second statement. In general, \autoref{thm:result} does not apply to sets
    of constraint functions which are not conjugate closed. However, since clockwise and counterclockwise
    orientations are equivalent for binary signatures, we can at least rescue \autoref{lem:forward} in
    this case. Namely, the $\subset$ inclusion of \autoref{thm:generatepn} still holds, since we can
    transpose any $\f$ gadget, switching between the two orientations, simply by pivoting its
    dangling edges: $\f^\top = (\f^{(r)})^{d,m}$ for some $r \in \{0,1\}$ and $d+m=2$.
    Then \autoref{thm:quantumholant} also applies because \autoref{lem:tensorcftog} and 
    \autoref{lem:rotateinput} give $\u^{\otimes d} \f^\top (\u^{\otimes m})^\dagger = \mathbf{G}^\top$.
\end{proof}

\paragraph*{The tensor isomorphism game.}
Given the results this work, it is natural to attempt a generalization of the results of this 
section to $n > 2$ -- a ``tensor isomorphism game'' with $n$ players. However, any such extension must
overcome the following observation by Roberson \cite{roberson}, that \autoref{lem:guuh} does not in 
general hold for $n > 2$. Recall the quantum symmetric group $S_q^+$ (i.e. $\qut(\varnothing)$) from
\autoref{def:sqplus}. By \autoref{lem:qutfintertwiners}, $C_{S_q^+} = \tcwd{E^{1,0}, E^{1,2}}$
(the span of the ``non-crossing partitions''). Thus
$f = (I \otimes E^{1,0} \otimes I) \circ E^{2,0} \in C_{S_q^+}$ 
(recall \autoref{lem:mclosure}), with $F_{abc} = \delta_{ac}$. For example, for $q \geq 3$, 
$F_{121} = 1 \neq 0 = F_{123}$. An $n=3$ version of \autoref{lem:guuh} would imply
$u_{11} u_{22} u_{13} = 0$, where $\u$ is the fundamental representation of $S_q^+$. However, this
is false for $q \geq 4$, since in this case $S_q^+$ has \emph{free orbitals} 
(see \cite[Theorem 2.5]{mccarthy2023tracing}).

\section{Planar Gadget Categories}
\label{sec:category}
In this section, we briefly discuss the relationship between Holant gadgets
and diagrammatic representations of monoidal categories, in particular pivotal dagger categories.
We hope to clarify the terminology \emph{tensor category with duals} from \autoref{sec:backward}.
For convenience, we specialize to gadgets in the context of $\plholant(\fc \mid \eq)$, but everything
in this section holds for gadgets in the context of any Holant problem.

Define a category $\gc_{\fc}$ as follows. The objects of $\gc_{\fc}$ are natural numbers, and
$\gc_{\fc}(k,\ell) = \pn(k,\ell)$ for every $k,\ell \in \mathbb{N}$ -- that is, a morphism from
$k$ to $\ell$ is a planar gadget with $k$ output and $\ell$ input dangling edges. Arrow composition
is gadget composition $\circ$, and the identity morphism is $1_k = \ii^{\otimes k}$. 
The gadget tensor product $\otimes$ 
makes $\gc_{\fc}$ a \emph{monoidal} category (a concept similar to that of a \emph{tensor category}).
For objects $k,\ell \in \mathbb{N}$, we define $k \otimes \ell = k + \ell$. Then if
$\k_1 \in \pn(k_1,\ell_1)$ and $\k_2 \in \pn(k_2,\ell_2)$ are morphisms $k_1 \to \ell_1$ and
$k_2 \to \ell_2$, $\k_1 \otimes \k_2 \in \pn(k_1 + k_2,\ell_1 + \ell_2)$ is a morphism
$k_1 \otimes k_2 \to \ell_1 \otimes \ell_2$.

The definition of $\gc_{\fc}$ follows that of the category $\mathcal{D}$ 
in \cite{temperley}, whose morphisms $k \to \ell$ are planar diagrams corresponding to elements of
a \emph{Temperley-Lieb algebra}, composed similarly to gadgets by merging the diagram `edges'. More
generally, monoidal categories are often depicted using such diagrams of edges, or \emph{strings},
composed in a similar fashion, with various properties (e.g. pivotal categories, dagger categories)
expressed graphically, as below \cite{selinger_survey_2011}. 

Recall the gadget conjugate transpose operation $\dagger$ from \autoref{def:gadgetops}.
Define $k^{\dagger} = k$ for
objects $k$, so we have $\k^{\dagger}: \ell^{\dagger} \to k^{\dagger}$ if $\k: k \to \ell$.
$\k^{\dagger}$ is commonly called the \emph{adjoint} of $\k$.
Diagramatically, $\dagger$ is simply a horizontal reflection, so $\dagger$ is involutive,
$(\k_1 \circ \k_2)^{\dagger} = \k_2^{\dagger} \circ \k_1^{\dagger}$, and $1_k^{\dagger} = 1_k$.
Hence $\dagger$ is a contravariant functor, making $\gc_{\fc}$ a \emph{dagger category}. 
Recall from \autoref{def:gadgetops} and the proofs of \autoref{thm:generatepn} and 
\autoref{thm:quantumholant} that $\dagger$ flips signatures' orientations 
between clockwise and counterclockwise. We think of $\dagger$'s entrywise conjugation effect as
corresponding to this reversal of orientation.

Our other key graphical gadget manipulation is pivoting dangling edges between output and input
using $\e^{2,0}$ and $\e^{0,2}$,
in the sense of \autoref{lem:rotategadget}.
\autoref{rem:frobenius} motivates the inclusion of $M(\e^{2,0}) = E^{2,0}$
in the definition of a tensor category with duals; indeed, in the context of monoidal/tensor categories,
$\e^{2,0}$ is the important \emph{counit} gadget $\e^{2,0} = \epsilon_1: 1 \otimes 1 \to 0$. 
More generally, 
$\epsilon_n: n \otimes n \to 0$ consists of $n$ nested copies of $\e^{2,0}$, and the \emph{unit} gadgets
$\eta_n = \epsilon_n^{\dagger}: 0 \to n \otimes n$ consist of $n$ nested copies of $\e^{0,2}$.
Clearly $\eta_n, \epsilon_n \in \pn$. The units and counits satisfy the following
identity, illustrated in \autoref{fig:unitcounit} for $n = 3$:
\begin{equation}
    (\eta_n \otimes 1_n) \circ (1_n \otimes \epsilon_n) = 1_n =
    (1_n \otimes \eta_n) \circ (\epsilon_n \otimes 1_n). 
    \label{eq:unitcounit}
\end{equation}

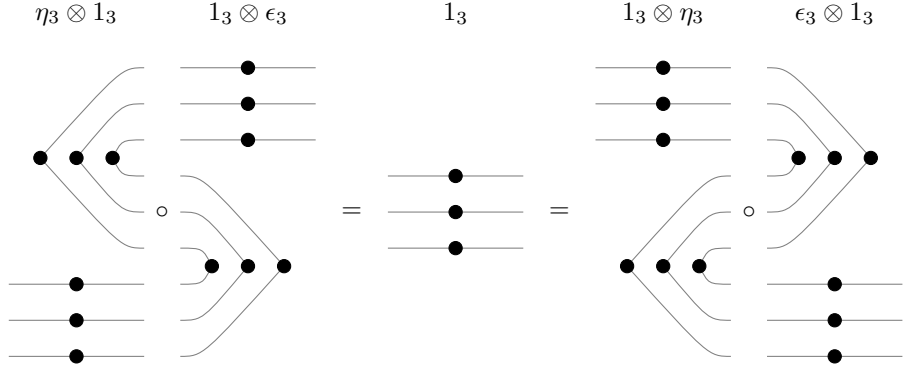
\begin{figure}[ht!]
    \center
    \begin{tikzpicture}[scale=.6]
    \GraphInit[vstyle=Classic]
    \SetUpEdge[style=-]
    \SetVertexMath
    \tikzset{VertexStyle/.style = {shape=circle, fill=black, minimum size=5pt, inner sep=1pt, draw}}

    \def\wirelen{1.5}
    \def\wiregap{0.8}
    \def\xa{0}
    \def\xb{\xa + 2*\wirelen + \wiregap}
    \def\xc{\xb + 2*\wirelen + 2*\wiregap}
    \def\xd{\xc + 2*\wirelen + 2*\wiregap}
    \def\xe{\xd + 2*\wirelen + \wiregap}

    \def\curvegap{0.8}

    \def\ygap{0.8}
    \def\ya{10}
    \def\yb{\ya - \ygap}
    \def\yc{\yb - \ygap}
    \def\yd{\yc - \ygap}
    \def\ye{\yd - \ygap}
    \def\yf{\ye - \ygap}
    \def\yg{\yf - \ygap}
    \def\yh{\yg - \ygap}
    \def\yi{\yh - \ygap}

    \node at (\xb-\wiregap/2, \ye) {$\circ$};
    \node at (\xc-\wiregap, \ye) {$=$};
    \node at (\xd-\wiregap, \ye) {$=$};
    \node at (\xe-\wiregap/2, \ye) {$\circ$};

    \node at (\xa+\wirelen, \ya+1.2) {$\eta_3 \otimes 1_3$};
    \node at (\xb+\wirelen, \ya+1.2) {$1_3 \otimes \epsilon_3$};
    \node at (\xc+\wirelen, \ya+1.2) {$1_3$};
    \node at (\xd+\wirelen, \ya+1.2) {$1_3 \otimes \eta_3$};
    \node at (\xe+\wirelen, \ya+1.2) {$\epsilon_3 \otimes 1_3$};

    \def\ymidhigh{\yc - \ygap/2}
    \def\ymidlow{\yf - \ygap/2}

    \draw[thin, color=gray] ({\xa + (2*\wirelen)}, \ya) 
        .. controls ({\xa + (2*\wirelen)-0.5}, \ya)
        .. (\xa+\wirelen-\curvegap, \ymidhigh) 
        .. controls ({\xa + (2*\wirelen)-0.5}, \yf) 
        .. ({\xa + (2*\wirelen)}, \yf);
    \draw[thin, color=gray] ({\xa + (2*\wirelen)}, \yb) 
        .. controls ({\xa + (2*\wirelen)-0.5}, \yb)
        .. (\xa+\wirelen, \ymidhigh) 
        .. controls ({\xa + (2*\wirelen)-0.5}, \ye) 
        .. ({\xa + (2*\wirelen)}, \ye);
    \draw[thin, color=gray] ({\xa + (2*\wirelen)}, \yc) 
        .. controls ({\xa + (2*\wirelen)-0.5}, \yc)
        .. (\xa+\wirelen+\curvegap, \ymidhigh) 
        .. controls ({\xa + (2*\wirelen)-0.5}, \yd) 
        .. ({\xa + (2*\wirelen)}, \yd);
    \Vertex[x=\xa+\wirelen-\curvegap,y=\ymidhigh,NoLabel]{v}
    \Vertex[x=\xa+\wirelen,y=\ymidhigh,NoLabel]{v}
    \Vertex[x=\xa+\wirelen+\curvegap,y=\ymidhigh,NoLabel]{v}

    \draw[thin, color=gray] (\xa, \yg) -- ({\xa + (2*\wirelen)}, \yg);
    \draw[thin, color=gray] (\xa, \yh) -- ({\xa + (2*\wirelen)}, \yh);
    \draw[thin, color=gray] (\xa, \yi) -- ({\xa + (2*\wirelen)}, \yi);
    \Vertex[x=\xa+\wirelen,y=\yg,NoLabel]{v}
    \Vertex[x=\xa+\wirelen,y=\yh,NoLabel]{v}
    \Vertex[x=\xa+\wirelen,y=\yi,NoLabel]{v}

    \draw[thin, color=gray] (\xb, \ya) -- ({\xb + (2*\wirelen)}, \ya);
    \draw[thin, color=gray] (\xb, \yb) -- ({\xb + (2*\wirelen)}, \yb);
    \draw[thin, color=gray] (\xb, \yc) -- ({\xb + (2*\wirelen)}, \yc);
    \Vertex[x=\xb+\wirelen,y=\ya,NoLabel]{v}
    \Vertex[x=\xb+\wirelen,y=\yb,NoLabel]{v}
    \Vertex[x=\xb+\wirelen,y=\yc,NoLabel]{v}

    \draw[thin, color=gray] (\xb, \yd) 
        .. controls ({\xb +0.5}, \yd)
        .. (\xb+\wirelen+\curvegap, \ymidlow) 
        .. controls ({\xb+0.5}, \yi) 
        .. (\xb, \yi);
    \draw[thin, color=gray] (\xb, \ye) 
        .. controls ({\xb +0.5}, \ye)
        .. (\xb+\wirelen, \ymidlow) 
        .. controls ({\xb+0.5}, \yh) 
        .. (\xb, \yh);
    \draw[thin, color=gray] (\xb, \yf) 
        .. controls ({\xb +0.5}, \yf)
        .. (\xb+\wirelen-\curvegap, \ymidlow) 
        .. controls ({\xb+0.5}, \yg) 
        .. (\xb, \yg);
    \Vertex[x=\xb+\wirelen-\curvegap,y=\ymidlow,NoLabel]{v}
    \Vertex[x=\xb+\wirelen,y=\ymidlow,NoLabel]{v}
    \Vertex[x=\xb+\wirelen+\curvegap,y=\ymidlow,NoLabel]{v}

    \draw[thin, color=gray] (\xc, \yd) -- ({\xc + (2*\wirelen)}, \yd);
    \draw[thin, color=gray] (\xc, \ye) -- ({\xc + (2*\wirelen)}, \ye);
    \draw[thin, color=gray] (\xc, \yf) -- ({\xc + (2*\wirelen)}, \yf);
    \Vertex[x=\xc+\wirelen,y=\yd,NoLabel]{v}
    \Vertex[x=\xc+\wirelen,y=\ye,NoLabel]{v}
    \Vertex[x=\xc+\wirelen,y=\yf,NoLabel]{v}

    \draw[thin, color=gray] (\xd, \ya) -- ({\xd + (2*\wirelen)}, \ya);
    \draw[thin, color=gray] (\xd, \yb) -- ({\xd + (2*\wirelen)}, \yb);
    \draw[thin, color=gray] (\xd, \yc) -- ({\xd + (2*\wirelen)}, \yc);
    \Vertex[x=\xd+\wirelen,y=\ya,NoLabel]{v}
    \Vertex[x=\xd+\wirelen,y=\yb,NoLabel]{v}
    \Vertex[x=\xd+\wirelen,y=\yc,NoLabel]{v}
    \draw[thin, color=gray] ({\xd + (2*\wirelen)}, \yd) 
        .. controls ({\xd + (2*\wirelen)-0.5}, \yd)
        .. (\xd+\wirelen-\curvegap, \ymidlow) 
        .. controls ({\xd + (2*\wirelen)-0.5}, \yi) 
        .. ({\xd + (2*\wirelen)}, \yi);
    \draw[thin, color=gray] ({\xd + (2*\wirelen)}, \ye) 
        .. controls ({\xd + (2*\wirelen)-0.5}, \ye)
        .. (\xd+\wirelen, \ymidlow) 
        .. controls ({\xd + (2*\wirelen)-0.5}, \yh) 
        .. ({\xd + (2*\wirelen)}, \yh);
    \draw[thin, color=gray] ({\xd + (2*\wirelen)}, \yf) 
        .. controls ({\xd + (2*\wirelen)-0.5}, \yf)
        .. (\xd+\wirelen+\curvegap, \ymidlow) 
        .. controls ({\xd + (2*\wirelen)-0.5}, \yg) 
        .. ({\xd + (2*\wirelen)}, \yg);
    \Vertex[x=\xd+\wirelen-\curvegap,y=\ymidlow,NoLabel]{v}
    \Vertex[x=\xd+\wirelen,y=\ymidlow,NoLabel]{v}
    \Vertex[x=\xd+\wirelen+\curvegap,y=\ymidlow,NoLabel]{v}

    \draw[thin, color=gray] (\xe, \ya) 
        .. controls ({\xe +0.5}, \ya)
        .. (\xe+\wirelen+\curvegap, \ymidhigh) 
        .. controls ({\xe+0.5}, \yf) 
        .. (\xe, \yf);
    \draw[thin, color=gray] (\xe, \yb) 
        .. controls ({\xe +0.5}, \yb)
        .. (\xe+\wirelen, \ymidhigh) 
        .. controls ({\xe+0.5}, \ye) 
        .. (\xe, \ye);
    \draw[thin, color=gray] (\xe, \yc) 
        .. controls ({\xe +0.5}, \yc)
        .. (\xe+\wirelen-\curvegap, \ymidhigh) 
        .. controls ({\xe+0.5}, \yd) 
        .. (\xe, \yd);
    \Vertex[x=\xe+\wirelen-\curvegap,y=\ymidhigh,NoLabel]{v}
    \Vertex[x=\xe+\wirelen,y=\ymidhigh,NoLabel]{v}
    \Vertex[x=\xe+\wirelen+\curvegap,y=\ymidhigh,NoLabel]{v}
    \draw[thin, color=gray] (\xe, \yg) -- ({\xe + (2*\wirelen)}, \yg);
    \draw[thin, color=gray] (\xe, \yh) -- ({\xe + (2*\wirelen)}, \yh);
    \draw[thin, color=gray] (\xe, \yi) -- ({\xe + (2*\wirelen)}, \yi);
    \Vertex[x=\xe+\wirelen,y=\yg,NoLabel]{v}
    \Vertex[x=\xe+\wirelen,y=\yh,NoLabel]{v}
    \Vertex[x=\xe+\wirelen,y=\yi,NoLabel]{v}

\end{tikzpicture}
    \caption{Illustrating \eqref{eq:unitcounit} for $n = 3$.}
    \label{fig:unitcounit}
\end{figure}

Define another contravariant functor $*$ as follows: $k^* = k$ for objects $k$ and, for morphism
$\k: k \to \ell$, let $\k^*: \ell^* \to k^*$ be the gadget defined by pivoting all output dangling edges
to become inputs and all input dangling edges to become outputs, using the procedure from
\autoref{lem:rotategadget}. Equivalently, 
\begin{align*}
    \k^* &= 
    (1_{k} \otimes \eta_{\ell}) \circ 
    (1_{k} \otimes \k \otimes 1_{\ell}) \circ 
    (\epsilon_{k} \otimes 1_{\ell}) \\
     &= (\eta_{\ell} \otimes 1_{k}) \circ
    (1_{\ell} \otimes \k \otimes 1_{k}) \circ 
    (1_{\ell} \otimes \epsilon_k).
    \numberthis\label{eq:sweep}
\end{align*}
See \autoref{fig:fstar}. The LHS of \eqref{eq:sweep} corresponds to pivoting the output and input
dangling edges around the bottom and top of the gadget, respectively, as shown in \autoref{fig:fstar}.
\begin{figure}[ht!]
    \center
    \begin{tikzpicture}[scale=0.8]
    \tikzset{VertexStyle/.style = {shape=circle, fill=black, minimum size=5pt, inner sep=1pt, draw}}
    \GraphInit[vstyle=Classic]
    \SetUpEdge[style=-]
    \SetVertexMath

    \def\wirelen{1}
    \def\xb{5}
    \def\xc{10}

    \node at (0, 2.8) {$\k$};
    \node at (\xb, 2.8) {$\k^{\dagger}$};
    \node at (\xc, 2.8) {$\k^*$};

    \draw[thin, color=gray] (-1-\wirelen,-0.4) .. controls (-1-\wirelen/2,-0.4) .. (-1,0);
    \draw[thin, color=gray] (-1-\wirelen,0.4) .. controls (-1-\wirelen/2,0.4) .. (-1,0);
    \draw[thin, color=gray] (-1-\wirelen,1) -- (-0.5,1);
    \draw[thin, color=gray] (1,0) -- (1+\wirelen,0);
    \draw[thin, color=gray] (1,1) -- (1+\wirelen,1);

    \draw[thin, color=gray] (\xc-1,0) .. controls (\xc-0.7,-1) and (\xc-0.5,-1.5) .. (\xc+2,-1.5);
    \draw[thin, color=gray] (\xc-1,0) .. controls (\xc-0.7, -1.5) and (\xc-0.5,-2) .. (\xc+2,-2);
    \draw[thin, color=gray] (\xc-0.5,1) .. controls (\xc-2.8, 0) and (\xc-1.5,-2.5) .. (\xc+2,-2.5);

    \draw[thin, color=gray] (\xc+1,0) .. controls (\xc+2.5,1) and (\xc+1,2.2) .. (\xc-2,2.2);
    \draw[thin, color=gray] (\xc+1,1) .. controls (\xc,1.4) and (\xc-1,1.7) .. (\xc-2,1.7);

    \foreach \xs/\num in {0/1,\xc/2} {

        \Vertex[x=-1+\xs,y=0,NoLabel]{a\num}
        \Vertex[x=-0.5+\xs,y=1,NoLabel]{b\num}
        \Vertex[x=1+\xs,y=0,NoLabel]{c\num}
        \Vertex[x=0+\xs,y=-1,NoLabel]{d\num}
        \Vertex[x=1+\xs,y=1,NoLabel]{e\num}
    };

    \foreach \xs/\num in {0/1,\xshift/2} {
        \Edges(a\num,b\num,c\num,d\num,a\num,c\num,e\num)
    };

    \draw[thin, color=gray] (\xb+1+\wirelen,-0.4) .. controls (\xb+1+\wirelen/2,-0.4) .. (\xb+1,0);
    \draw[thin, color=gray] (\xb+1+\wirelen,0.4) .. controls (\xb+1+\wirelen/2,0.4) .. (\xb+1,0);
    \draw[thin, color=gray] (\xb+1+\wirelen,1) -- (\xb+0.5,1);
    \draw[thin, color=gray] (\xb-1,0) -- (\xb-1-\wirelen,0);
    \draw[thin, color=gray] (\xb-1,1) -- (\xb-1-\wirelen,1);

    \Vertex[x=\xb+1,y=0,NoLabel]{a3}
    \Vertex[x=\xb+0.5,y=1,NoLabel]{b3}
    \Vertex[x=\xb-1,y=0,NoLabel]{c3}
    \Vertex[x=\xb,y=-1,NoLabel]{d3}
    \Vertex[x=\xb-1,y=1,NoLabel]{e3}

    \Edges(a3,b3,c3,d3,a3,c3,e3)
\end{tikzpicture}
    \caption{Illustrating $\k$, $\k^{\dagger}$, and $\k^*$.}
    \label{fig:fstar}
\end{figure}
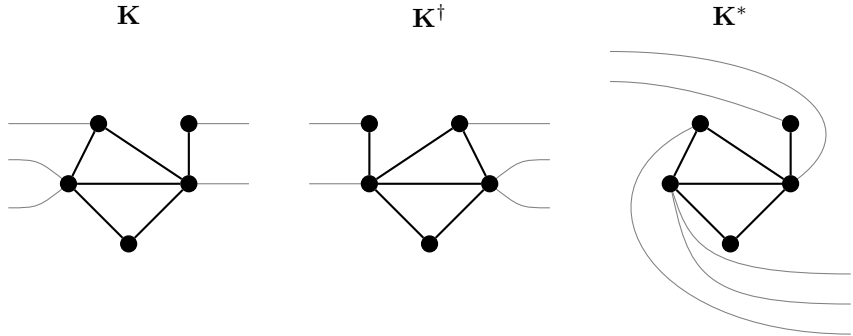
The RHS of \eqref{eq:sweep} pivots the outputs and inputs around the top and bottom, respectively.
$\k^*$ is called the \emph{dual} of $\k$, hence the terminology ``tensor category with duals.'' $*$ is a rotation, so, unlike the reflection $\dagger$, 
it preserves
$\k$'s cyclic input order; we do not apply entrywise conjugation as with $\dagger$. We have
$1_k^* = 1_k$, $(\k_1 \circ \k_2)^* = \k_2^* \circ \k_1^*$, and $\k^{**} = \k$, so $*$, like $\dagger$,
is indeed an involutive contravariant functor. This fact, along with the identities \eqref{eq:unitcounit}
and \eqref{eq:sweep}, make $\gc_{\fc}$ a \emph{pivotal category}, hence a \emph{pivotal dagger category}
\cite{temperley}.

\printbibliography
\end{document}